\documentclass[aps,pra,showpacs,amssymb,nofootinbib,superscriptaddress,twocolumn,longbibliography]{revtex4-1}
\usepackage{tikz}
\usetikzlibrary{shapes,snakes,backgrounds,fit,decorations.pathreplacing}
\usepackage{bbm,times}
\usepackage{bm}
\usepackage{amsbsy}
\usepackage{amsthm}
\usepackage{amssymb}
\usepackage{amsfonts}
\usepackage{amsmath}
\usepackage{dsfont} 
\usepackage{graphicx} 
\usepackage{epsfig}
\usepackage{epstopdf}
\usepackage{dsfont}
\usepackage{multibib}
\usepackage{color}
\usepackage{enumerate}
\usepackage{multirow}

\usepackage[colorlinks]{hyperref}
\makeatletter
\newcommand\org@hypertarget{}
\let\org@hypertarget\hypertarget
\renewcommand\hypertarget[2]{%
  \Hy@raisedlink{\org@hypertarget{#1}{}}#2%
  }
\makeatother
\usepackage[figure,table]{hypcap}
\usepackage{MnSymbol}
\usepackage{enumerate}
\usepackage{float}
\hypersetup{
	bookmarksnumbered,
	pdfstartview={FitH},
	citecolor={darkgreen},
	linkcolor={darkred},
	urlcolor={darkblue},
	pdfpagemode={UseOutlines}}
\definecolor{darkgreen}{RGB}{50,190,50}
\definecolor{darkblue}{RGB}{0,0,190}
\definecolor{darkred}{RGB}{238,0,0}
\usepackage{soul}
\renewcommand*{\vec}[1]{\mathbf{#1}}
\newcommand{\pr}{^{\prime}}

\newcommand{\ket}[1]{\ensuremath{\left|\right.\!{#1}\!\left.\right\rangle}}

\newcommand{\bra}[1]{\ensuremath{\left\langle\right.\!{#1}\!\left.\right|}}

\newcommand{\ketbra}[2]{\ensuremath{|{#1}\rangle\langle{#2}|}}

\newcommand{\nr}{\ensuremath{\hspace*{0.5pt}}}

\newcommand{\subtiny}[3]{\ensuremath{_{\hspace{#1 pt}\protect\raisebox{#2 pt}{\tiny{$ #3$}}}}}
\newcommand{\suptiny}[3]{\ensuremath{^{\hspace{#1 pt}\protect\raisebox{#2 pt}{\tiny{$ #3$}}}}}
\newcommand{\SA}{\ensuremath{_{\hspace{-1pt}\protect\raisebox{0pt}{\tiny{$A$}}}}}

\newcommand{\SB}{\ensuremath{_{\hspace{-1pt}\protect\raisebox{0pt}{\tiny{$B$}}}}}

\newcommand{\SAB}{\ensuremath{_{\hspace{-1pt}\protect\raisebox{0pt}{\tiny{$A\hspace*{-0.5pt}B$}}}}}

\newcommand{\ignore}[1]{}

\newcommand{\dv}{\ensuremath{|\hspace*{-0.75pt}|}}

\DeclareMathOperator{\diag}{diag}

\newcommand{\tr}{\textnormal{Tr}}
\newcommand{\djj}{d\kern-0.4em\char"16\kern-0.1em}

\newtheorem{lemma}{Lemma}
\newtheorem{claim}{Claim}

\newtheorem{prop}{Proposition}[]

\renewcommand{\thesection}{\arabic{section}}
\renewcommand{\thesubsection}{\arabic{section}.\Alph{subsection}}
\makeatletter
\renewcommand{\p@subsection}{}
\renewcommand{\p@subsubsection}{}
\makeatother

\usepackage[customcolors]{hf-tikz}
\tikzset{style green/.style={
    set fill color=green!50!lime!60,
    set border color=white,
  },
  style cyan/.style={
    set fill color=cyan!90!blue!60,
    set border color=white,
  },
  style orange/.style={
    set fill color=orange!80!red!60,
    set border color=white,
  },
  style hordash/.style={
    set fill color=white,
    set border color=black,
  },
  hor/.style={
    above left offset={-0.09,0.25},
    below right offset={0.09,-0.05},
    #1
  },
  ver/.style={
    above left offset={-0.09,0.35},
    below right offset={0.09,-0.1},
    #1
  }
}

\usepackage[framemethod=tikz]{mdframed}
\definecolor{mycolor}{rgb}{0.122, 0.435, 0.698}
\newmdenv[innerlinewidth=0.5pt, roundcorner=4pt,linecolor=mycolor,innerleftmargin=6pt,
innerrightmargin=6pt,innertopmargin=6pt,innerbottommargin=6pt]{mybox}
\usepackage{paracol}
\usepackage{tcolorbox}
\tcbuselibrary{theorems}

\newtcbtheorem{Definitions}{Definition}%
{colback=mycolor!5,colframe=mycolor,fonttitle=\bfseries,
left=4pt,
right=4pt,
top=5pt,
bottom=5pt}{defi}

\newtcbtheorem{Lemmas}{Lemma}%
{colback=green!5,colframe=green!50!blue,fonttitle=\bfseries,
left=4pt,
right=4pt,
top=5pt,
bottom=5pt,
title={#2},
}{lemma}

\newtcbtheorem{Theorems}{Theorem}%
{colback=blue!3,colframe=blue!35!black,fonttitle=\bfseries,
left=4pt,
right=4pt,
top=5pt,
bottom=5pt,
title={#2},
}{theorem}
\newtcbtheorem{Questions}{Question}%
{colback=orange!5,colframe=orange!50!blue,fonttitle=\bfseries,
left=4pt,
right=4pt,
top=5pt,
bottom=5pt,
title={#2},
}{question}

\newtcbtheorem{Conjectures}{Conjecture}%
{colback= brown!5,colframe=brown!80!black,fonttitle=\bfseries,
left=4pt,
right=4pt,
top=5pt,
bottom=5pt,
title={#2},
}{conjecture}

\newtcbtheorem{Corollaries}{Corollary}%
{colback= purple!5,colframe=purple!50!black,fonttitle=\bfseries,
left=4pt,
right=4pt,
top=5pt,
bottom=5pt,
title={#2},
}{corollary}
\newtcolorbox[blend into=figures]{boxdefi}[3][]
{ float*=ht,width=\textwidth,lower separated=false, 
title={#2},label= def:#3,#1}
\newtcbtheorem{WideBoxes}{Box}%
{width=\textwidth,
fonttitle=\bfseries,
left=4pt,
right=4pt,
top=5pt,
bottom=5pt,float*=ht!}{boxbox}
\renewcommand{\thesection}{\Roman{section}}
\renewcommand{\thesubsection}{\Roman{section}.\arabic{subsection}}
\renewcommand{\thesubsubsection}{\Roman{section}.\arabic{subsection}.\arabic{subsubsection}}
\makeatletter
\renewcommand{\p@subsection}{}
\renewcommand{\p@subsubsection}{}
\makeatother

\usepackage{filecontents}

\begin{document}

\title{Thermodynamically optimal creation of correlations}
\author{Faraj Bakhshinezhad}
\thanks{These authors contributed equally to this work.}
\affiliation{Department of Physics, Sharif University of Technology, Tehran 14588, Iran}
\affiliation{Institute for Quantum Optics and Quantum Information - IQOQI Vienna, Austrian Academy of Sciences, Boltzmanngasse 3, 1090 Vienna, Austria}
\author{Fabien Clivaz}
\thanks{These authors contributed equally to this work.}
\affiliation{Department of Applied Physics, University of Geneva, 1211 Geneva 4, Switzerland}
\affiliation{Institute for Quantum Optics and Quantum Information - IQOQI Vienna, Austrian Academy of Sciences, Boltzmanngasse 3, 1090 Vienna, Austria}
\author{Giuseppe Vitagliano}
\affiliation{Institute for Quantum Optics and Quantum Information - IQOQI Vienna, Austrian Academy of Sciences, Boltzmanngasse 3, 1090 Vienna, Austria}
\author{Paul Erker}
\affiliation{Institute for Quantum Optics and Quantum Information - IQOQI Vienna, Austrian Academy of Sciences, Boltzmanngasse 3, 1090 Vienna, Austria}
\author{Ali T. Rezakhani}
\affiliation{Department of Physics, Sharif University of Technology, Tehran 14588, Iran}
\author{Marcus Huber}
\email{marcus.huber@univie.ac.at}
\affiliation{Institute for Quantum Optics and Quantum Information - IQOQI Vienna, Austrian Academy of Sciences, Boltzmanngasse 3, 1090 Vienna, Austria}
\author{Nicolai Friis}
\email{nicolai.friis@univie.ac.at}
\affiliation{Institute for Quantum Optics and Quantum Information - IQOQI Vienna, Austrian Academy of Sciences, Boltzmanngasse 3, 1090 Vienna, Austria}

\begin{abstract}
Correlations lie at the heart of almost all scientific predictions. It is therefore of interest to ask whether there exist general limitations to the amount of correlations that can be created at a finite amount of invested energy. Within quantum thermodynamics such limitations can be derived from first principles. In particular, it can be shown that establishing correlations between initially uncorrelated systems in a thermal background has an energetic cost. This cost, which depends on the system dimension and the details of the energy-level structure, can be bounded from below but whether these bounds are achievable is an open question. Here, we put forward a framework for studying the process of optimally correlating identical (thermal) quantum systems. The framework is based on decompositions into subspaces that each support only states with diagonal (classical) marginals. Using methods from stochastic majorisation theory, we show that the creation of correlations at minimal energy cost is possible for all pairs of three- and four-dimensional quantum systems. For higher dimensions we provide sufficient conditions for the existence of such optimally correlating operations, which we conjecture to exist in all dimensions.
\end{abstract}

\maketitle


\section{Introduction}\label{sec:introduction}


Correlations can be regarded as the fundamental means of obtaining information: From a physical perspective, obtaining knowledge requires correlation of physical variables in an observed system with physical variables in a measurement apparatus~\cite{GuryanovaFriisHuber2018}. As long as one regards system and measurement apparatus as separate entities which are not persistently strongly interacting~\cite{FriisHuberPerarnauLlobet2016}, this correlation requires an investment of energy. Conversely, all correlations imply extractable work~\cite{PerarnauLlobetHovhannisyanHuberSkrzypczykBrunnerAcin2015}. Since correlations and energy (work) can be considered to be resources in (quantum) information theory and thermodynamics, respectively, these observations establish one of the fundamental connections between these theories. Indeed, both theories already share a common framework using statistical ensembles to capture knowledge about collections of physical systems. This knowledge can then be harnessed to facilitate the most efficient use of the available resources, e.g., energy, for the tasks at hand. The young field of quantum thermodynamics combines features from both fields and investigates their interplay in the quantum domain~\cite{GooldHuberRieraDelRioSkrzypczyk2016, VinjanampathyAnders2016, MillenXuereb2016}.

Here, we explore the specific quantitative relation between correlations and energy~\cite{VitaglianoKloecklHuberFriis2019}. While general qualitative insights can help to understand some quandaries arising from Maxwell's demon or Szilard's engine~\cite{Bennett1982, LeffRex2003, MayuramaNoriVedral2009}, and correlations play various interesting roles in quantum thermodynamics (see, e.g.,~\cite{HuberPerarnauHovhannisyanSkrzypczykKloecklBrunnerAcin2015, PerarnauLlobetHovhannisyanHuberSkrzypczykBrunnerAcin2015, BruschiPerarnauLlobetFriisHovhannisyanHuber2015, FriisHuberPerarnauLlobet2016, BinderVinjanampathyModiGoold2015, AlipourBenattiBakhshinezhadAfsaryMarcantoniRezakhani2016, BeraRieraLewensteinWinter2017, Mueller2018, BeraRieraLewensteinBaghaliWinter2019, SapienzaCerisolaRoncaglia2019}), precise quantitative statements about the trade-off between work and correlations are generally complicated. For example, by allowing arbitrarily slow quasi-static operations and perfect control over arbitrarily many auxiliary systems, one may provide tight lower bounds on the work cost of creating bipartite correlations as measured by the mutual information~\cite{JevticJenningsRudolph2012,JevticJenningsRudolph2012b,HuberPerarnauHovhannisyanSkrzypczykKloecklBrunnerAcin2015, BruschiPerarnauLlobetFriisHovhannisyanHuber2015, VitaglianoKloecklHuberFriis2019}. However, for two identical systems, these bounds are tight only in the case when specific so-called symmetrically thermalizing unitaries (STUs) exist, i.e., unitaries that map initial thermal states to locally thermal states at higher local effective temperatures. Moreover, how well two systems can be correlated for a given energy in finite time and with limited control is generally not known. In particular, the possibility of optimal conversion of energy into correlations for the interesting special case of fully controlled, closed systems with two identical subsystems rests entirely on the assumed existence of STUs. In this sense, central open questions at the very foundations of quantum thermodynamics concern the quantitative correspondence between correlations and energy.

In this paper, we put forward a framework for investigating potential marginal spectra in the unitary orbit of any quantum state via decompositions into locally classical subspaces (LCSs) and majorisation relations. Here, LCSs are subspaces of a bipartite Hilbert space that support only states that are locally classical, i.e., states which have diagonal marginals with respect to a chosen basis. We use this approach to investigate the question of existence of STUs for arbitrary bipartite quantum systems. In particular, we show that STUs exist for two identical copies of arbitrary three- and four-dimensional systems, thereby extending previous results on the special case of Hamiltonians with equally-spaced energy gaps~\cite{JevticJenningsRudolph2012b, HuberPerarnauHovhannisyanSkrzypczykKloecklBrunnerAcin2015}. In addition, we formulate sufficient conditions for the existence of STUs for identical subsystems with arbitrary local dimension and for multiple copies of the initial state. Indeed, we show that STUs exist in the limit of infinitely many copies.

This paper is structured as follows: In Sec.~\ref{sec:framework}, we describe the conceptual framework of our investigation, formulate our central question, and briefly review the state of the art. In Sec.~\ref{sec:asymmetric case}, we further discuss that, while STUs do not in general exist for asymmetric situations where the two systems have different Hamiltonians, this does not resolve the problem of maximizing correlations for asymmetric systems. We then focus on the open case of two identical subsystems in Sec.~\ref{subsec:ent subspace} and discuss a family of unitary transformations that result in symmetric, diagonal marginals. These unitaries have a block-diagonal structure with respect to a specific choice of LCSs, and give rise to a rich structure of potential marginal transformations. In particular, they allow showing that STUs exist in the asymptotic limit of infinitely many copies. In Sec.~\ref{sec:two-qutrit case}, we then formulate three alternative approaches to demonstrating the existence of STUs based on unitaries in these LCSs. As we show, all three approaches allow proving the existence of STUs in local dimension $d=3$. However, only two of these approaches are successful (to different extent) in local dimension $d=4$ and provide general sufficient conditions for the existence of STUs for higher dimensions. Finally, in Sec.~\ref{sec:conclusion} we provide a summary and discussion and put forward a hypothesis for general dimensions.


\section{Framework}\label{sec:framework}

\subsection{Main question}\label{sec:main question}

Let us start our investigation with a detailed description of the problem: We consider the process of correlating two previously uncorrelated quantum systems. More precisely, we study the controlled interaction of two quantum systems $A$ and $B$ with Hamiltonians $H\SA$ and $H\SB$, respectively, which are initially in a tensor product state $\varrho\SA\otimes\varrho\SB$, with the goal of increasing the correlations between them. As we want to obtain optimal bounds, we consider optimal coherently controlled systems, i.e. the ability to externally engineer and time any interaction Hamiltonian between the two quantum systems, resulting in unitary dynamics on the system and we thus study the global unitary orbit of product states. Furthermore, we are interested in the energy needed to establish correlations between these initially uncorrelated quantum systems.

One could generalise the question from perfect system control and thus unitaries on the system, to correlating maps generated by unitaries on the system and arbitrary auxiliary systems, inducing completely positive and trace preserving (CPTP) maps on the system. However, some transformations (e.g., cooling to the ground state) may only be achievable asymptotically (i.e., using infinite time or energy~\cite{ClivazSilvaHaackBohrBraskBrunnerHuber2019a, MasanesOppenheim2017, WilmingGallego2017, ClivazSilvaHaackBohrBraskBrunnerHuber2019b}). Moreover, the implementation of CPTP maps generally requires access to and control over auxiliary systems, resulting in an implicit energy cost for preparation and manipulation of the auxiliaries that is obscured by the use of the CPTP maps instead of the explicit description of the corresponding unitaries on the total Hilbert space. Consequently, it is of interest to understand the limits of unitary correlating processes, which allow for a complete account of all the energetic changes within the system. Moreover, to give a fair account of the energy that needs to be supplied to the system, it is important to consider initial states that do not implicitly supply free energy. That is, we assume that the initial states of the subsystems are thermal states at the same temperature $T$, i.e., $\varrho\SA=\tau\SA(\beta)$ and $\varrho\SB=\tau\SB(\beta)$ with $\beta=1/T$ (we use units where $k\subtiny{0}{0}{\mathrm{B}}=1$ throughout), where $\tau(\beta):=e^{-\beta H}/\mathcal{Z}$ and $\mathcal{Z}=\tr(e^{-\beta H})$. In this way, the initial state $\tau\SA(\beta)\otimes\tau\SB(\beta)$ minimises the local energies for given entropies and is also completely passive~\cite{PuszWoronowicz1978}, i.e., a state from which no energy can be extracted unitarily even when taking multiple copies.

Correlations can be quantified in different ways. Qualitatively, one may describe correlations between two subsystems as the feature that more information is available about properties of the joint system than about properties of the subsystems. If there is no preference as to the specific properties that are to be correlated, i.e., if one is interested in quantifying any correlations, it is most useful to consider entropic measures. The most paradigmatic among them is the mutual information $\mathcal{I}(\varrho\SAB)=S(\varrho\SA)+S(\varrho\SB)-S(\varrho\SAB)$, based on the von Neumann entropy $S(\varrho):=-\tr\bigl(\varrho\ln\varrho\bigr)$. Although we focus on the von Neumann entropy, in principle one can replace it with any other measure of local `mixedness'. And indeed, our methods are framed in the context of majorisation relations, and thus in principle also open an avenue towards studying other entropic measures.

When focusing on the mutual information, the problem of unitarily creating correlations can be reduced to the task of maximising the sum of the marginal entropies under local energy constraints. That is, the invested energy $\Delta E:=\text{Tr}\bigl[H\SAB\bigl( U\SAB\tau\SA(\beta)\otimes\tau\SB(\beta)U^{\dagger}\SAB-\tau\SA(\beta)\otimes\tau\SB(\beta)\bigr)\bigr]$, where $H\SAB:=\,H\SA\,+\,H\SB,$ only depends on the local marginals, while the global entropy is invariant under unitary operations. The change in mutual information thus reduces to $\Delta \mathcal{I}=\Delta S\SA+\Delta S\SB$. We thus arrive at the following question:

\begin{Questions}{Maximal mutual information under energy constraint}{1}
For a given pair of local Hamiltonians $H\SA$ and $H\SB$ and initial temperature $T=1/\beta$, what is the maximum value of
\begin{align}
   \Delta \mathcal{I}=\Delta S\SA+\Delta S\SB
    \label{Ques:DeltaMuUuuuuuh}
\end{align}
for an invested energy of at most $\Delta E$?
\end{Questions}

The fact that thermal states maximise the local entropies under local energy constraints suggests that
the maximal value of $\Delta \mathcal{I}$ is always achieved when the final state marginals are thermal at a higher temperature $T\pr\geq T$ (or, equivalently, for $\beta\pr\leq \beta$), i.e., $\tilde{\varrho}\SA=\tau\SA(\beta')$ and $\tilde{\varrho}\SB=\tau\SB(\beta')$, provided that the corresponding STUs exist. This is indeed the case as shown in Appendix \ref{app:upperbound} and is what leads us to the question of the existence of STUs. While when $H\SA=H\SB$, one might still hope for STUs to exist for all dimensions, Hamiltonians and temperatures, such a statement cannot be made when $H\SA\neq H\SB$. Before introducing our framework for constructing STUs, let us therefore briefly discuss general bounds on the achievable local entropies.


\subsubsection{Asymmetric case}\label{sec:asymmetric case}

Let us consider Question~\ref{question:1} in the case $H\SA\neq H\SB$. Although the problem is in general complex, note that nontrivial constraints on marginal spectra for given global spectra exist in the form of entropy inequalities. In particular, both the von Neumann entropy $S(\varrho)$ and the R{\'e}nyi $0$-entropy $S_{0}(\varrho):=-\log_2(\text{rank}(\varrho))$ satisfy subadditivity~\cite{LiebRuskai1973, CadneyHuberLindenWinter2014}, while no other entropy satisfies (dimension independent) linear inequalities~\cite{LindenMosonyiWinter2013}. Since we are interested in thermodynamic questions, where the rank of the involved states is typically full, only subadditivity of the von Neumann entropy, which we can write as
\begin{align}
    S(\varrho\SAB) &\geq |S(\varrho\SA)-S(\varrho\SB)|,
    \label{triangle ineq}
\end{align}
provides a nontrivial constraint. Also observe that, given an initial state $\tau\SA(\beta)\otimes \tau\SB (\beta)$ subject to a unitary evolution, the entropy of the total system is fixed to $S(\varrho\SAB) = S(\tau\SA(\beta))+S(\tau\SB (\beta))$. If we assume that the final marginals are thermal states with equal higher temperature i.e., $\beta'\leq \beta$, inequality~(\ref{triangle ineq}) can be rewritten as
\begin{align}
  |S(\tau\SA(\beta^{\prime}))-S(\tau\SB (\beta^{\prime}))| &\leq\, S(\tau\SA(\beta))+ S(\tau\SB (\beta)).
\end{align}
This already implies a nontrivial constraint, and also provides a counterargument to the naive assumption that two thermal marginals at the same temperature are reachable in general.

While this makes the structure complicated to deal with in general, one can still solve the optimisation in some special cases. In particular, for a system initialised in the ground state, i.e., a product of two pure states $\tau\SAB (\beta \to \infty)= \ket{0\SA\,0\SB}\!\!\bra{0\SA\,0\SB}$, the total state remains pure during the unitary evolution and hence has a Schmidt decomposition. The final state can therefore be written as
\begin{align}
    \ket{\tilde{\psi}}  &= \sum_{i=0}^{d} \sqrt{p_i}\ket{\varphi\suptiny{0}{0}{A}_{i},  \varphi\suptiny{0}{0}{B}_{i}},
    \label{eq:schmidt decomp}
\end{align}
for some probabilities $p_i$, with $d= \mathrm{min}\{d\SA,d\SB\}$, and $d\SA$ and $d\SB$ the dimensions of the respective Hilbert spaces $\mathcal{H}\SA$ and $\mathcal{H}\SB$. And the final marginals are
\begin{equation}
\tilde{\varrho}\SA
=\sum_{i=0}^{d-1} p_i \ket{\varphi\suptiny{0}{0}{A}_{i}}\!\!\bra{\varphi\suptiny{0}{0}{A}_{i}},\ \ \
\tilde{\varrho}\SB=
\sum_{i=0}^{d-1} p_i \ket{\varphi\suptiny{0}{0}{B}_{i}}\!\!\bra{\varphi\suptiny{0}{0}{B}_{i}}.
\label{substateA}
\end{equation}
This already means that the marginals (since they have equal rank) cannot both be thermal states if $d\SA\neq d\SB$.

However, in this case it is still possible to maximise the amount of correlation subject to a given maximal amount of energy consumption. In other words, one can find optimal points of the following optimisation problem:
\begin{equation}
\label{eq:puremaxS}
\begin{aligned}
&\max_{\tilde{\varrho}\SA, \tilde{\varrho}\SB} S(\tilde{\varrho}\SA)+S(\tilde{\varrho}\SB)\\
&\text{subject to}\ \tr(\tilde{\varrho}\SA H\SA)+ \tr(\tilde{\varrho}\SB H\SB) \leq c\\
& \hphantom{\text{subject to}}\ \tilde{\varrho}\SA\ \text{and}\ \tilde{\varrho}\SB\ \text{as in Eq.~(\ref{substateA}).}
\end{aligned}
\end{equation}
To do so one first solves the (loosely speaking) inverse problem of minimising the energy consumption for a given amount of correlation. This problem can be turned into two instances of the well-known passivity problem~\cite{PuszWoronowicz1978}, for which passive states are the solutions. Using this fact, the problem reduces to a simple problem that can be solved using Lagrange multipliers. One then shows that the solution found for that problem is strictly monotonically increasing in the constraint, allowing us to reverse the problem and show that it is a solution of the problem in (\ref{eq:puremaxS}). For more details see Appendix~\ref{app:asymmetric pure state}. See also Ref.~\cite{PiccioneMilitelloNapoliBellomo2019} for an alternative approach to the energy minimisation problem. The solutions obtained for (\ref{eq:puremaxS}) are thermal states at inverse temperature $\beta(c)$ (uniquely defined by the constant $c$) of the Hamiltonians $\tilde{H}\SA$ and $\tilde{H}\SB$, respectively, defined as
\vspace*{-1mm}
\begin{equation}
\begin{aligned}
    \tilde{H}\SA&= \sum_{i=0}^{d-1} (E_{i}\suptiny{0}{0}{A}+E_{i}\suptiny{0}{0}{B}) \ket{i}\!\!\bra{i}\SA\\
    \tilde{H}\SB&= \sum_{i=0}^{d-1} (E_{i}\suptiny{0}{0}{A}+E_{i}\suptiny{0}{0}{B}) \ket{i}\!\!\bra{i}\SB,
\end{aligned}
\end{equation}
where
\vspace*{-1.5mm}
\begin{equation}
\begin{aligned}
    H\SA&= \sum_{i=0}^{d\SA-1} E_{i}\suptiny{0}{0}{A} \ket{i}\!\!\bra{i}\SA, \quad E_{i}\suptiny{0}{0}{A} \leq E_{i+1}\suptiny{0}{0}{A},\\
    H\SB&= \sum_{i=0}^{d\SB-1} E_{i}\suptiny{0}{0}{B} \ket{i}\!\!\bra{i}\SB, \quad E_{i}\suptiny{0}{0}{B} \leq E_{i+1}\suptiny{0}{0}{B}.\\
\end{aligned}
\end{equation}
The solutions hence take the form
\begin{equation}
\begin{aligned}
    \tilde{\varrho}_{\text{opt},\tiny{A}} (\beta(c)) &= \frac{\tilde{\Pi}\SA e^{-\beta(c) \tilde{H}\SA}\tilde{\Pi}\SA}{\tr(\tilde{\Pi}\SA e^{-\beta(c) \tilde{H}\SA})},\\
        \tilde{\varrho}_{\text{opt},\tiny{B}} (\beta(c)) &= \frac{\tilde{\Pi}\SB e^{-\beta(c) \tilde{H}\SB}\tilde{\Pi}\SB}{\tr(\tilde{\Pi}\SB e^{-\beta(c) \tilde{H}\SB})},
    \end{aligned}
\end{equation}
where $\tilde{\Pi}\subtiny{0}{0}{A/B}:= \sum_{i=0}^{d-1} \ket{i}\!\!\bra{i}\subtiny{0}{0}{A/B}$.
The maximal amount of correlation achievable given a maximal amount of energy $c$ is then
\begin{align}
    2 S(\tilde{\varrho}\SA) &= 2 \beta(c) c+ 2 \ln \left[ \tr(e^{-\beta(c) \tilde{H}\SA})\right].
    \label{finmut1}
\end{align}
\vspace*{-2mm}


\subsubsection{Symmetric case}\label{sec:preliminaries}
\vspace*{-2mm}

Let us now consider the symmetric case $H\SA=H\SB$. Following our considerations in the previous section and Appendix~\ref{app:upperbound}, it is clear that an upper bound for $\Delta\mathcal{I}$ is in this case achieved when $\Delta S\SA=\Delta S\SB$, and since the states with maximal entropy given a fixed energy are thermal states, Question~\ref{question:1} can be substantially simplified to that of the existence of symmetrically thermalizing unitaries.\\

\begin{Questions}{Existence of STUs}{2}
Does there exist a unitary $U\SAB$ on $\mathcal{H}\SAB$ such that
\begin{align}
   \tilde{\varrho}\SA &=\tr\SB\bigl(U\SAB\tau\SAB(\beta) U\SAB^{\dagger}\bigr)=\tau\SA(\beta\pr),\nonumber\\
    \tilde{\varrho}\SB &=\tr\SA\bigl(U\SAB\tau\SAB(\beta) U\SAB^{\dagger}\bigr)=\tau\SB(\beta\pr),
    \label{Ques:Uandmarginals}
\end{align}
for every pair of local Hamiltonians $H\SA=H\SB$, for all final temperatures $T\pr=1/\beta\pr$ and all initial temperatures $T=1/\beta\leq T\pr$?
\end{Questions}

If such STUs exist, then they are automatically the optimally correlating unitaries for a given system. It is already known that Question~\ref{question:1} can be answered affirmatively when the subsystem Hamiltonians 
are equally spaced, i.e., $H\SA=H\SB=\sum_{n=0}^{d-1}E_{n}\ket{n}\!\!\bra{n}$ with $E_{n+1}-E_{n}=\omega ~\forall n\in\{0,1,\ldots,d-1\}$ for some constant $\omega$ (with appropriate units)~\cite{HuberPerarnauHovhannisyanSkrzypczykKloecklBrunnerAcin2015} (see also the more recent formulation of the derivation in Ref.~\cite{McKayRodriguezEduardo2018}). In particular, this implies that such optimally correlating unitaries always exist for two qubits, i.e., when $d=2$. However, there is (yet) no answer for the general situation of arbitrary Hamiltonians and dimensions. To fill this gap, we present a framework that generalises previous approaches~\cite{HuberPerarnauHovhannisyanSkrzypczykKloecklBrunnerAcin2015} and allows investigating this general question in the symmetric cases for any dimension. This approach is based on unitary operations in LCSs, as we will explain in Sec.~\ref{subsec:ent subspace}, and enables us to show that STUs exist for the simplest nontrivial cases of two qutrits ($d=3$) and two ququarts ($d=4$).


\subsection{Unitary operations on locally classical subspaces}\label{subsec:ent subspace}


\begin{WideBoxes}{Unitary operations on locally classical subspaces for two qutrits}{Boxdim3}
For any tensor product of two identical states $\tau\SAB=\sum_{i,j=0}^{2}p_{ij}\,\ketbra{ij}{ij}$, we can separate the local diagonal elements (i.e., the eigenvalues, since the matrix is diagonal) through an economic notation where we denote the diagonal entries of the marginals as entries in vectors
\begin{equation}
    \vec{p}\SA:= \vec{\tau}\SA=
    \begin{pmatrix}
        \textcolor{blue!50!black}{p_{00}}+\textcolor{green!50!black}{p_{01}}+\textcolor{red!50!black}{p_{02}}\\
        \textcolor{blue!50!black}{p_{11}}+\textcolor{green!50!black}{p_{12}}+\textcolor{red!50!black}{p_{10}}\\
        \textcolor{blue!50!black}{p_{22}}+\textcolor{green!50!black}{p_{20}}+\textcolor{red!50!black}{p_{21}}\\
    \end{pmatrix}=\textcolor{blue!50!black}{\vec{q}}+\textcolor{green!50!black}{\vec{r}_1}+\textcolor{red!50!black}{\vec{r}_2}, \quad
    \vec{p}\SB:=\vec{\tau}\SB=
    \begin{pmatrix}
        \textcolor{blue!50!black}{p_{00}}+\textcolor{red!50!black}{p_{10}}+\textcolor{green!50!black}{p_{20}}\\
        \textcolor{blue!50!black}{p_{11}}+\textcolor{red!50!black}{p_{21}}+\textcolor{green!50!black}{p_{01}}\\
        \textcolor{blue!50!black}{p_{22}}+\textcolor{red!50!black}{p_{02}}+\textcolor{green!50!black}{p_{12}}\\
    \end{pmatrix}=\textcolor{blue!50!black}{\vec{q}}+\textcolor{red!50!black}{\Pi^2\,\vec{r}_2}+\textcolor{green!50!black}{\Pi\,\vec{r}_1},
\end{equation}
where $\Pi=(\Pi_{ij})$ is a cyclic permutation matrix with components $\Pi_{ij}=\,\delta_{i,j+1\operatorname{mod}3}$. Employing the unitary
$U\SAB=\textcolor{blue!50!black}{U_q}\oplus\textcolor{green!50!black}{U_{r_1}}\oplus \textcolor{red!50!black}{U_{r_2}}$, where $U_q$, $U_{r_{1}}$, and $U_{r_{2}}$ act unitarily on the subspaces
\begin{equation}
    \textcolor{blue!50!black}{\mathcal{H}_{q}:=\mathrm{span}\{ \ket{00},\ket{11},\ket{22}\}},\quad
    \textcolor{green!50!black}{\mathcal{H}_{r_1}:=\mathrm{span}\{\ket{01},\ket{12},\ket{20}\}},\quad
    \textcolor{red!50!black}{\mathcal{H}_{r_2}:=\mathrm{span}\{\ket{02}, \ket{10},\ket{21}\}},
\end{equation}
respectively, the transformation of the marginals can be described by unistochastic matrices
\begin{equation}
    \vec{\tilde{p}}\SA=\textcolor{blue!50!black}{M_q\,\vec{q}}+\textcolor{green!50!black}{M_{r_1}\vec{r}_1}+\textcolor{red!50!black}{M_{r_2}\vec{r}_2}, \quad
    \vec{\tilde{p}}\SB=\textcolor{blue!50!black}{M_{q}\vec{q}}+\textcolor{red!50!black}{\Pi^2\,M_{r_2}\vec{r}_2}+\textcolor{green!50!black}{\Pi \,M_{r_1}\vec{r}_1}.
\end{equation}
The components of the unistochastic matrices are determined by the moduli squared of corresponding unitary matrices, $(M_{\alpha})_{kl}=|(U_{\alpha})_{kl}|^{2}$ for $\alpha=q,r_{1},r_{2}$. Due to the symmetry $p_{ij}=p_{ji}$, we can further identify ${\vec{r}_2}=\Pi {\vec{r}_1}$ to obtain
\begin{equation}
    \vec{\tilde{p}}\SA=\textcolor{blue!50!black}{M_q\,\vec{q}}+\textcolor{green!50!black}{M_{r_1}\vec{r}_1}+\textcolor{red!50!black}{M_{r_2}\Pi \vec{r}_1}, \quad
    \vec{\tilde{p}}\SB=\textcolor{blue!50!black}{M_{q}\vec{q}}+\textcolor{red!50!black}{\Pi^2\,M_{r_2}\Pi\,\vec{r}_1}+\textcolor{green!50!black}{\Pi\,M_{r_1}\vec{r}_1}.
\end{equation}
\end{WideBoxes}

Let us now discuss our general framework for marginal transformations. The central idea of this approach is to decompose the diagonal elements of the marginals into elements originating from different subspaces, with the property that any unitary that leaves the division into these subspaces invariant never creates local off-diagonal elements.

More precisely, we consider a pair of $d$-dimensional systems $A$ and $B$ with matching local Hamiltonians $H\SA=H\SB=\sum_{i=0}^{d-1}E_{i}\ket{i}\!\!\bra{i}$, where we set $E_{0}=0$ without loss of generality. Moreover, let the energy eigenvalues, sorted in non-decreasing order (with $E_{i+1}\geq E_{i}\ \forall\,i$), be measured in units of $E_{1}$. The initial thermal states of $A$ and $B$ are 
\begin{align}
    \tau\SA(\beta) &=\,\tau\SB(\beta)\,
    =\,\sum_{i=0}^{d-1}\,p_{i}(\beta)\,\ket{i}\!\!\bra{i},
    \label{eq:thermal state}
\end{align}
where $p_{i}(\beta)=e^{-\beta E_i}/\mathcal{Z}(\beta)$ and $\mathcal{Z}(\beta)=\sum_{i}e^{-\beta E_{i}}$. For convenience, we use the shorthand $p_{i}\equiv p_{i}(\beta)$. The joint initial state $\tau\SAB=\tau\SA(\beta)\otimes\tau\SB (\beta)$ is also diagonal with respect to the energy eigenbasis and its diagonal entries are products of the diagonal entries of the reduced states, i.e., we write $p_{ij}:=p_{i}p_{j}$ such that $\tau\SAB=\sum_{i,j=0}^{d-1}p_{ij}\,\ketbra{ij}{ij}$. Since the unitaries that we use throughout the manuscript do not create coherence in the marginals, it is convenient to introduce a vectorised notation for the diagonal entries of the reduced states. That is, for arbitrary diagonal joint states $\vec{p}\SAB=\diag\{p_{ij}\}_{i,j=0}^{d-1}$ where the $p_{ij}$ need not factorise with respect to $i$ and $j$, the reduced states $\varrho\SA=\tr\SB(\varrho\SAB)$ and $\varrho\SB=\tr\SA(\varrho\SAB)$ are diagonal, and the diagonal entries can be collected into vectors $\vec{p}\SA,\vec{p}\SB\in\mathbb{R}^{d}$ with nonnegative components and unit $1$-norm, $\dv \vec{p}\SA\dv=\tr(\varrho\SA)=1=\dv \vec{p}\SB\dv$.


\subsubsection{Locally classical subspaces}\label{sec:LCSs}

Having expressed the diagonals of the marginals in this vectorised form, we further choose particular subspace vector decompositions, as we illustrate for $d=3$ in Box~\ref{boxbox:Boxdim3}. For general local dimension $d\geq3$, we write $\vec{p}\SA$ and $\vec{p}\SB$ as sums of $d$ vectors according to
\begin{align}
    \vec{p}\SA \,=\,\sum_{i=0}^{d-1}\vec{r}\suptiny{0}{0}{A}_{i}
    \,=\,\sum_{i=0}^{d-1}\vec{r}_{i},
    \ \ \
    \vec{p}\SB \,=\,\sum_{i=0}^{d-1}\vec{r}\suptiny{0}{0}{B}_{i}=\sum_{i=0}^{d-1}\Pi^i\,\vec{r}_{i},
    \label{deco. marginal}
\end{align}
where $\vec{r}\suptiny{0}{0}{A}_{i} =\vec{r}_{i}= \sum_{j=0}^{d-1} p_{j\, j+i}\,\vec{e}_{j}$ and $\vec{r}\suptiny{0}{0}{B}_{i}=\sum_{j=0}^{d-1} p_{j-i\, j}\,\vec{e}_{j}$. Here, $\{\vec{e}_{j}\}_{j=0}^{d-1}$ denotes the chosen orthonormal basis of $\mathbb{R}^{d}$ and all indices are to be taken modulo $d$. In the second equality, we have further used the fact that the $i$-th vector $\vec{r}\suptiny{0}{0}{B}_{i}$ in the decomposition of $\vec{p}\SB$ can be related to the $i$-th vector $\vec{r}\suptiny{0}{0}{A}_{i}$ in the decomposition of $\vec{p}\SA$ via the $i$-th power of the cyclic permutation matrix
$\Pi=(\Pi_{ij})$ with components $\Pi_{ij}=\,\delta_{i,j+1\operatorname{mod}d}$.

One further observes that the decomposition of $\vec{p}\SA$ and $\vec{p}\SB$ into these vectors corresponds to the selection of a total of $d$ subspaces $\mathcal{H}_{q}$ and $\{\mathcal{H}_{r_i}\}_{i=1}^{d-1}$ of the joint Hilbert space $\mathcal{H}\SAB$,
\begin{align}
    \mathcal{H}_{q} := \mathcal{H}_{r_0}=\,\operatorname{span}\{\ket{j\,j}\}_{j=0}^{d-1},\ \ \
    \mathcal{H}_{r_i} =\,\operatorname{span}\{\ket{j\,j+i}\}_{j=0}^{d-1},
    \label{eq:subspace dec}
\end{align}
with $\mathcal{H}\SAB=\mathcal{H}_{q}\oplus_{i=1}^{d-1}\mathcal{H}_{r_i}$, such that arbitrary unitaries $U_{q}$ and $U_{r_i}$ applied in either of the $d$ subspaces cannot lead to nonzero off-diagonal elements in the reduced states $\varrho\SA$ or $\varrho\SB$, provided that none are present to begin with. In other words, any state on $\mathcal{H}\SAB$ that has support on only one of these subspaces, or which has a block-diagonal structure with respect to this subspace decomposition, is locally classical, i.e., has diagonal marginals with respect to the local bases $\{\ket{j}\}_{j=0}^{d-1}$. We hence call these subspaces \emph{locally classical}. However, note that the choice of LCSs we make here is not unique, as discussed below in Sec.~\ref{sec:general LCSs}.

A unitary transformation $\varrho\SAB\mapsto\tilde{\varrho}\SAB=U\varrho\SAB U^{\dagger}$ with $U=U_{q}\oplus_{i=1}^{d-1} U_{r_i}$ hence does not preclude off-diagonal elements from appearing in the joint state $\tilde{\varrho}\SAB$ but leaves the reduced states diagonal, implying that we can directly read off the spectra of the marginals. In such a case, it is still useful to describe the marginals by $d$-component vectors $\vec{\tilde{p}}\SA$ and $\vec{\tilde{p}}\SB$ collecting their diagonal elements, and the transformation of these vectors can be represented as
\begin{align}
    \vec{p}\SA &\mapsto\,\tilde{\vec{p}}\SA\,=\,M_{q}\vec{q}+\sum_{i=1}^{d-1}M_{r_i}\vec{r}_{i},\label{eq: gen evo. marginal A }\\
    \vec{p}\SB &\mapsto\,\tilde{\vec{p}}\SB\,=\,M_{q}\vec{q}+\sum_{i=1}^{d-1}\Pi^i M_{r_i}\vec{r}_{i},
    \label{gen. evolution marginal}
\end{align}
where $\vec{q}=\vec{r}\suptiny{0}{0}{A}_{0}=\vec{r}\suptiny{0}{0}{B}_{0}$, and $M_{q}$ and $M_{r_i}$ are unistochastic $d\times d$ matrices, i.e., doubly stochastic (entries in rows and columns sum to $1$) square matrices $M_{\alpha}$ (with $\alpha\in\{q,r_i\}$) whose components can be understood as the moduli squared of unitary matrices, $(M_{\alpha})_{kl}=|(U_{\alpha})_{kl}|^{2}$.

As a technical remark, note that the Schur-Horn theorem implies a weak convexity property for unistochastic matrices, namely that for any $\vec{v}\in\mathbb{R}^{d}$, the set of vectors obtained by applying the set of unistochastic $d\times d$ matrices to $\vec{v}$ is equivalent to the set of vectors obtained by applying the set of doubly stochastic $d\times d$ matrices to $\vec{v}$, see Ref.~\cite{BengtssonEricssonKusTadejZyczkowski2005}. This means that for any doubly stochastic matrix $M$ and vector $\vec{v}$, there exists a unistochastic matrix $M_{\vec{v}}$ such that $M\,\vec{v}=M_{\vec{v}}\vec{v}$. Here, this property permits us to conclude that the vectorised marginals of any state reachable by unitaries of the form $U=U_{q}\oplus_{i=1}^{d-1} U_{r_i}$ can be written in terms of the action of doubly stochastic matrices $M_{q}$ and $M_{r_{i}}$ on $\vec{q}$ and $\vec{r}_{i}$, respectively.

In general, Eq.~(\ref{eq: gen evo. marginal A }) can be rewritten as
\begin{align}
    \tilde{\vec{p}}\SA  &=
    \begin{cases}
        M_{q}\vec{q} + \sum\limits_{i=1}^{k} (M_{r_{i}}\vec{r}_{i} + M_{r_{d-i}}\vec{r}_{d-i})
        & d\ \text{odd}\\[2mm]
        M_q\vec{q} + \sum\limits_{i=1}^{k-1} (M_{r_{i}}\vec{r}_{i} + M_{r_{d-i}}\vec{r}_{d-i})+M_{r_{d/2}} \vec{r}_{d/2}
        & d\ \text{even}
    \end{cases},
    \label{gen evol marginal A}
\end{align}
where $k=\tfrac{d-1}{2}$ if $d$ is odd and $k=\tfrac{d}{2}$ if $d$ is even. Further taking into account the symmetry $p_{ij} = p_{ji}$, which implies $\vec{r}_{d-i}= \sum_{j=0}^{d-1} p_{j\, j+d-i} \Pi^{i}\vec{e}_{j-i}=\Pi^{i} \vec{r}_{i}$, the vectorised form of the marginal $A$ can be written in a compact way as
\begin{align}
    \tilde{\vec{p}}\SA &=\,M_{q}\vec{q} + \sum\limits_{i=1}^{k} (\lfloor\tfrac{2i}{d}\rfloor+1)^{-1}(M_{r_i} + M_{r_{d-i}}\Pi^i)\vec{r}_{i},
    \label{eq:gen evo. marg. A 1}
\end{align}
where $\lfloor x \rfloor$ denotes the floor function of $x$, and the prefactor $(\lfloor\tfrac{2i}{d}\rfloor+1)^{-1}$ is equal to $1$ unless $d$ is even and $i=k$, in which case it is $\tfrac{1}{2}$. Using  Eq.~(\ref{gen. evolution marginal}), one may obtain the transformed marginal $B$ as
\begin{align}
    \tilde{\vec{p}}\SB & = M_q\vec{q} + \sum\limits_{i=1}^{k} (\lfloor\tfrac{2i}{d}\rfloor+1)^{-1} (\Pi^i M_{r_i} + \Pi^{-i}M_{r_{d-i}}\Pi^i)\, \vec{r}_{i},
\end{align}
where we have used the property $\Pi^{-i}=\Pi^{d-i}$ of the $d$-dimensional cyclic permutation. To satisfy the conditions of Question~\ref{question:2}, we need to further restrict ourselves to a subgroup of transformations which generate the same marginals, $\tilde{\vec{p}}\SA=\tilde{\vec{p}}\SB$. This requirement results in the condition
\begin{align}
    M_{r_{d-i}}=\Pi^{i} M_{r_i} \Pi^{-i}\ \ \ \ \text{for}\ \ \ i=1,\dots,k,
\end{align}
for the doubly stochastic matrices mentioned in Eq.~(\ref{eq:gen evo. marg. A 1}).
With this, we arrive at the form
\begin{align}
    \vec{\tilde{p}}=\vec{\tilde{p}}\SA=\vec{\tilde{p}}\SB
    &=\, M_q \vec{q} + \sum_{i=1}^{k} (\lfloor\tfrac{2i}{d}\rfloor+1)^{-1} (\mathds{1}+ \Pi^i)M_{r_i} \vec{r}_i
    \label{eq:EqualEvoMarginal C}
\end{align}
for both vectorised marginals. Now that we have ensured that both marginals transform in the same way, we may concentrate on one of them, say $\tilde{\vec{p}}\SA$ (dropping the subscript $A$ for ease of notation) and use the established framework to investigate the existence of STUs as specified in Question~\ref{question:2}.


\subsubsection{General locally classical subspaces}\label{sec:general LCSs}

While we focus here on the decomposition into the specific LCSs from Eq.~(\ref{eq:subspace dec}), we note that this is by far not the only option. Indeed, let $\{\mathcal{P}_{i}\}_{i=0,\dots,d-1}$ be a set of permutations of $d$ elements $j=0,\dots,d-1$. Then a sufficient condition for obtaining an LCS decomposition $\mathcal{H}\SAB=\bigoplus_{i=0}^{d-1}\tilde{\mathcal{H}}_{\tilde{r}_i}$ into LCSs of the form $\tilde{\mathcal{H}}_{\tilde{r}_i}=\operatorname{span}\{\ket{j\,\mathcal{P}_{i}(j)}\}_{j=0}^{d-1}$ is that the $d\times d$ matrix $\Gamma$ with components $\Gamma_{ij}=\mathcal{P}_{i}(j)$ is a Latin square, i.e., every entry $j=0,\dots,d-1$ appears exactly once in each row and in each column. For every such choice of LCS decomposition, one may then apply unitaries $\tilde{U}=\bigoplus_{i=0}^{d-1}U_{\tilde{r}_i}$ that leave the LCSs invariant, i.e., where $U_{\tilde{r}_i}$ acts nontrivially only on the subspace $\tilde{\mathcal{H}}_{\tilde{r}_i}$. Denoting the corresponding vector decomposition of the first vectorised marginal as $\vec{p}\SA=\sum_{i=0}^{d-1}\tilde{\vec{r}}_{i}$, with $(\tilde{\vec{r}}_{i})_{j}=p_{j\mathcal{P}_{i}(j)}$, we have $\vec{p}\SB=\sum_{i=0}^{d-1}\mathcal{P}_{i}^{-1}\tilde{\vec{r}}_{i}$ since $(\vec{r}_{i}\suptiny{0}{0}{B})_{j}=p_{\mathcal{P}^{-1}_{i}(j)\nr j}=(\mathcal{P}_{i}^{-1}\tilde{\vec{r}}_{i})_{j}$. In this case, unitaries that leave the LCSs invariant transform the vectorised marginals according to
\vspace*{-2mm}
\begin{align}
    \vec{p}\SA &\mapsto\,\vec{\tilde{p}}\SA \,=\,\sum_{i=0}^{d-1}M_{\tilde{r}_i}\tilde{\vec{r}}_{i},
    \label{eq: gen evo. marginal A general}\\
    \vec{p}\SB &\mapsto\,\vec{\tilde{p}}\SB \,=\,\sum_{i=0}^{d-1}\mathcal{P}^{-1}_{i} M_{\tilde{r}_i}\tilde{\vec{r}}_{i},
\label{gen. evolution marginal general}
\end{align}
where the matrices $M_{\tilde{r}_i}$ are the unistochastic matrices corresponding to the unitaries $U_{\tilde{r}_i}$. This implies the existence of symmetric marginal transformations, e.g., for unistochastic matrices $M$ that commute with all the $\mathcal{P}_{i}$, such that the marginals take the form $M\vec{p}\SA =M\vec{p}\SB=\vec{\tilde{p}}$. For the remainder of this work we focus again on the specific special case we consider in Eq.~(\ref{eq:subspace dec}), where $\mathcal{P}_{i}=(\Pi^{-1})^{i}$.


\subsubsection{Asymptotic case}\label{subsubsec: asymcase}

Before we go further into detail on trying to answer Question~\ref{question:2}, let us briefly showcase that this is indeed a problem connected to the finite size of the system. More specifically, we can consider a scenario where one wishes to correlate multiple copies of the initial thermal states $\tau\SA(\beta)$ and $\tau\SB(\beta)$ via a joint unitary, such that the final state of $n$ copies is
$U\tau\SAB(\beta)^{\otimes n}U^{\dagger}$, with $\tau\SAB(\beta)^{\otimes n}=\tau\SA(\beta)^{\otimes n}\otimes\tau\SB(\beta)^{\otimes n}$. As we will show now, STUs such that $\tilde{\varrho}\SA=\tilde{\varrho}\SB=\tau(\beta\pr)$ exist for all $\beta\pr$, $\beta$ such that $\beta\pr\leq\beta$, and for all local Hamiltonians in the limit of infinitely many copies, $n\rightarrow\infty$.

To see this, note that for any $n$ we can find a unitary such that the marginals $\varrho{\ensuremath{_{\hspace{-1pt}\protect\raisebox{0pt}{\tiny{$A/B$}}}}}= \text{Tr}{\ensuremath{_{\hspace{-1pt}\protect\raisebox{0pt}{\tiny{$B/A$}}}}}( U\tau\SA(\beta)^{\otimes n}\otimes\tau\SB(\beta)^{\otimes n}U^{\dagger} )$ are passive states whose entropy equals that of the thermal state with the target temperature, i.e., $S(\varrho{\ensuremath{_{\hspace{-1pt}\protect\raisebox{0pt}{\tiny{$A/B$}}}}})=S(\tau{\ensuremath{_{\hspace{-1pt}\protect\raisebox{0pt}{\tiny{$A/B$}}}}}(\beta\pr))$. This is a consequence of Eq.~\eqref{eq:EqualEvoMarginal C} and the continuity of the von~Neumann entropy. More specifically, note that a trivial way of obtaining marginals with the same final spectrum is to select $M_{q}$ and all $M_{r_{i}}$ in Eq.~\eqref{eq:EqualEvoMarginal C} to be circulant, i.e., convex combinations of cyclic permutations. Since these matrices commute with $\Pi^{i}$, one may reach any marginal whose vectorised marginal is any cyclic permutation of the original marginal vector, or indeed, any convex combination of cyclic permutations of the original marginal vector. In particular, the equally weighted convex combination of all cyclic permutations yields maximally mixed marginals, $\tilde{\vec{p}}\SA=\tilde{\vec{p}}\SB=\tfrac{1}{d}\sum_{i=0}^{d-1}\vec{e}_{i}$, corresponding to maximal local entropy $\beta\pr\rightarrow0$. The convex structure of the set of vectors reachable from a given vector via application of unistochastic matrices then suggests that the points corresponding to $\vec{p}\SA$ and $\tilde{\vec{p}}\SA$ are connected by a continuous line of states reachable from $\vec{p}\SA$ via circulant unistochastic matrices. The entropy must vary continuously along this line from the initial value $S\bigl(\tau\SA(\beta)\bigr)$ to $S\bigl(\varrho\SA\bigr)=S\bigl(\tau\SA(\beta\pr)\bigr)$. Thus, one may obtain marginals $\varrho\SA=\varrho\SB$ that are passive states with the desired entropy, but which might have a different spectra than $\tau\SA(\beta\pr)=\tau\SB(\beta\pr)$.

Given this fact, it is sufficient to note that for $n\rightarrow \infty$ we can convert passive states with a given entropy $S(\varrho)$ into thermal states with the same entropy only by using local operations~\cite{PuszWoronowicz1978}. Because the application of local operations to the subsystems leaves the mutual information $\mathcal{I}(\varrho\SAB)$ invariant, this consequently proves the existence of STUs for the asymptotic case $n\rightarrow \infty$. Further discussion of the case of finitely many copies is presented in Appendix~\ref{passive}.


\section{Optimally correlating unitaries for bipartite states with matching Hamiltonians}\label{sec:two-qutrit case}

In this section, we present three approaches to construct STUs for bipartite systems in an (uncorrelated) initial thermal state at temperature $T=1/\beta$ of $H=H\SA+H\SB$ with matching Hamiltonians $H\SA=H\SB$. We focus on bipartite $3\times 3$ and $4 \times 4$-dimensional systems, and show explicit constructions for such STUs, thereby proving that Question~\ref{question:2} can be answered affirmatively for $d=3$ and $d=4$. In the $d=3$ case, we show the existence of STUs via all three alternative approaches as a basis for generalisations to higher dimensions. We then turn to the particular case of dimension $d=4$, highlighting the challenges and strengths encountered for each of these methods, before proving the $d=4$ case using a geometric argument.

Concretely, referring to Sec.~\ref{subsec:ent subspace} as our framework, we start from Eq.~(\ref{eq:EqualEvoMarginal C}), which provides a transformations of the reduced states that leaves both marginals equal and diagonal with respect to the chosen basis. In general, $M_q$ and all the $M_{r_i}$ are arbitrary doubly stochastic matrices, and therefore one can try to make use of majorisation conditions to prove the existence of STUs via the Hardy-Littlewood-P{\'o}lya (HLP) theorem~\cite[p.~75]{HardyLittlewoodPolya1952} or~\cite[p.~33]{MarshallOlkinArnold2011}. The HLP theorem states that a necessary and sufficient condition that $y \succ x$ is that there exist a doubly stochastic matrix $M$ such that $x = My$. This, in fact, is the first argument that we use below for $d=3$, where it allows proving a statement that is actually even stronger than what is required for STUs. However, the theorem that we are going to prove holds only for the $d=3$ case.


\subsection{Majorised marginals approach}\label{sec:commutingM}

In the particular case of $d=3$ the final expression for the (equal) marginals from Eq.~(\ref{eq:EqualEvoMarginal C}) is
\begin{align}
    \vec{\tilde{p}}=\vec{\tilde{p}}\SB\,=\vec{\tilde{p}}\SA\,=\,M_{q}\vec{q}+(\openone+\Pi)M_{r}\vec{r},
    \label{eq:gen. trans. 3d}
\end{align}
where, as in Eq.~(\ref{gen. evolution marginal}), we have introduced the vectors $\vec{q}$ and $\vec{r}$, with components
\begin{align}
    \vec{q} &=\,\begin{pmatrix} p_{00},\ p_{11},\ p_{22} \end{pmatrix}^T,\ \
    \vec{r}=\,\begin{pmatrix} p_{01},\ p_{12},\ p_{20} \end{pmatrix}^T.
    \label{vecqrs}
\end{align}
Based on the HLP theorem and Eq.~(\ref{eq:gen. trans. 3d}), we are going to prove the existence of STUs in $d=3$ via the following lemma:

\begin{lemma}\label{lemma:swappingM}
For any $3\times 3$ doubly stochastic matrix $M$, there exists a $3\times 3$ doubly stochastic matrix $\tilde{M}$ such that
\begin{equation}\label{eq:commutingM}
    M (\mathds{1}+\Pi)=(\mathds{1}+\Pi) \tilde{M}.
\end{equation}
\end{lemma}

The proof of Lemma~\ref{lemma:swappingM} is presented in Appendix~\ref{App:Proof:lemma:swappingM}. We are now ready to state the first theorem for $d=3$.\\

\begin{Theorems}{Majorised marginals in $d=3$}{2}
For every pair of states $\varrho$ and $\tilde{\varrho}$ in a $3$-dimensional Hilbert space which satisfy the condition $\vec{\lambda}(\tilde{\varrho})\prec \vec{\lambda}({\varrho})$, where $\vec{\lambda}(\varrho)$ is the vector of eigenvalues of $\varrho$, there exists a unitary $U\SAB$ on $\mathcal{H}\SAB$ such that
\begin{equation}
\begin{aligned}
   \tilde{\varrho}_{\hspace{-1pt}\protect\raisebox{0pt}{\tiny{$A$}}} &=\tr_{\hspace{-1pt}\protect\raisebox{0pt}{\tiny{$B$}}}\bigl(U\SAB\varrho \otimes \varrho U\SAB^{\dagger}\bigr)= \tilde{\varrho},\nonumber\\[1mm]
     \tilde{\varrho}_{\hspace{-1pt}\protect\raisebox{0pt}{\tiny{$B$}}}&=\tr_{\hspace{-1pt}\protect\raisebox{0pt}{\tiny{$A$}}}\bigl(U\SAB\varrho \otimes \varrho U\SAB^{\dagger}\bigr)=\tilde{\varrho}\nonumber.
\end{aligned}
\end{equation}
\end{Theorems}

\textbf{Proof.}\ The unitary is given by $U\SAB=U_{\rm loc}U_{\rm ent}$, i.e., by a product of an entangling unitary $U_{\rm ent}$ which acts on the global system preserving the equal marginals as discussed in Sec.~\ref{subsec:ent subspace}, and a local unitary of the form $U_{\rm loc}=U\otimes U$, such that the marginals are kept equal. From Lemma~\ref{lemma:swappingM} we then see that, given any doubly stochastic matrix $M_q$, one can find $M_r$ such that $M_q (\mathds{1} + \Pi)= (\mathds{1}+\Pi) M_r$ and consequently a mapping to each probability vector $\vec{\tilde{p}}$ such that
\begin{align}
  \vec{\tilde{p}}&= M_q \vec{q}+ (\mathds{1}+\Pi) M_r \vec{r}
  \,=\,M_q\, [ \vec{q}+ (\mathds{1}+\Pi)\vec{r}]=M_q \, \vec{p}.
\end{align}
From the HLP theorem, we thus know that it is possible to produce all the vectors $\vec{\tilde{p}}$ that are majorised by the vector $\vec{p}$ corresponding to the initial state. Thus, we proved that through the entangling unitary, all the marginals ${\varrho^\prime}=\sum_{i=0}^{2}\tilde{p}_i \ket{i}\!\!\bra{i}$ such that $\vec{\tilde{p}}\prec \vec{ p}$ can be reached from the initial state. Then, if we can also arbitrarily change their eigenstates in a symmetric way, Theorem~\ref{theorem:2} is proven. For that, we use a local unitary $U \otimes U$ with $U =\sum_{i=0}^2 \ket{\phi_i}\!\!\bra{i}$ to change the eigenstates of the marginal to an arbitrary set of eigenstates $\{ \ket{\phi_i}\}_{i=0}^{2}$, such that finally the marginals become $\tilde{\varrho}=\sum_{i=0}^{2}\tilde{p}_i \ket{\phi_i}\!\!\bra{\phi_i}$ in which the set $\{\ket{\phi_i}\}$ can be any orthonormal basis in $d=3$.\qed

As a corollary of Theorem~\ref{theorem:2}, the existence of STUs is proven for the two-qutrit case for initial and final thermal states with inverse temperatures $\beta$ and $\beta'\leq \beta$, respectively, since $\vec{p}(\beta\pr)\prec \vec{ p}(\beta)$ holds whenever $\beta' \leq\beta$.

Unfortunately, generalisations of Lemma~\ref{lemma:swappingM} to higher dimensions fail, as we explain in more detail in Appendix~\ref{sec:generalisations}. Despite the general statement of Theorem~\ref{theorem:2}, any attempts at generalisations must hence be based on a different approach, which does not require the existence of a doubly stochastic matrix $\tilde M$ as in Eq.~(\ref{eq:commutingM}) for all doubly stochastic matrices $M$.


\subsection{Alternative approach: ``passing on the norm''}\label{sec:3_passing}

Here, we explore another possible approach that makes use of the HLP theorem and tries to overcome the difficulties encountered for generalising the previous approach. Let us first discuss the method for two $d$-dimensional systems, before turning to the special cases $d=3$ and $d=4$.

To begin, we once again employ the transformation presented in Eq.~(\ref{eq:EqualEvoMarginal C}) to generate equal marginals. Recalling that the marginals of our final target state are thermal states at inverse temperature $\beta'<\beta$, we decompose the final state vector into the form
\begin{align}
    \vec{\tilde{p}} &= \vec{a} + \sum_{i=1}^{k} (\lfloor\tfrac{2i}{d}\rfloor+1)^{-1} (\mathds{1}+ \Pi^i) \vec{b}_{i},
    \label{eq:EqualEvoMarginal Cprime}
\end{align}
with $\vec{a}:= \vec{q}(\beta^\prime)= \sum_{j=0}^{d-1} p_{j\nr j}(\beta^\prime)\,\vec{e}_{j}$ and $\vec{b}_i:=\, \vec{r}_i(\beta^\prime)= \sum_{j=0}^{d-1} p_{j\nr j+i}(\beta^\prime)\,\vec{e}_{j}$. Further, recall that the initial states have the same decomposition with inverse temperature $\beta$. One way of transforming $\vec{p}$ into $\vec{\tilde{p}}$ is to transform $\vec{q}$ into $\vec{a}$ and $\vec{r}_i$ into $\vec{b}_i$ for all $i=1,\dots,k$. Looking at Eq.~(\ref{eq:EqualEvoMarginal C}), one realises that this is possible if
\begin{align}
    M_q \vec{q} &=\vec{a},\label{eq:intuitivePassNorm}\\[1mm]
    M_{r_i} \vec{r}_i   &=\vec{b}_i, \quad i=1,\dots,k. \label{eq:intuitivePassNorm2}
\end{align}
Unfortunately, for any $\beta'<\beta$ we have $\lVert \vec{q} \rVert > \lVert \vec{a} \rVert$, where $\lVert \cdot \rVert$ denotes the $1$-norm, see Appendix~\ref{app:PassNorm}, and since doubly stochastic transformations conserve the norm of vectors, Eq.~(\ref{eq:intuitivePassNorm}) cannot hold true. The norms of the $\vec{r}_i$ also generically differ from those of the $\vec{b}_i$, meaning Eq.~(\ref{eq:intuitivePassNorm2}) does not hold either. However, it is not clear whether $\lVert \vec{r}_i(\beta) \rVert $ is a monotonic function of $\beta$, and so $\lVert \vec{r}_i \rVert$ may in principle be smaller or larger than $\lVert \vec{b}_i \rVert$, depending on $\beta$, $\beta'$ and $i$. Loosely speaking, the vector $\vec{q}$ has `too much norm'. Intuitively, one may thus make use of the excessive norm of $\vec{q}$ by transforming the required `amount' of $\vec{q}$ into $\vec{a}$, and redistributing its excessive norm to the rest of $\tilde{\vec{p}}$, or in other words by passing the excessive norm of $\vec{q}$ on to the rest of $\tilde{\vec{p}}$, giving rise to the name of the approach. Formally, one achieves this by splitting $M_q$ into a convex combination of two doubly stochastic matrices, i.e.,
\begin{equation}
    M_q= \alpha_0 M_{q \rightarrow a}+ \alpha_1 M_{q \rightarrow b},
\end{equation}
where $\alpha_{0}, \alpha_{1} \geq 0$, $\alpha_0 +\alpha_1 =1$, and $M_{q \rightarrow a}$ and $M_{q \rightarrow b}$ are doubly stochastic matrices such that
\begin{align}
   M_{q \rightarrow a} \frac{\vec{q}}{\lVert \vec{q} \rVert}&= \frac{\vec{a}}{\lVert \vec{a} \rVert},\\
   M_{q \rightarrow b} \frac{\vec{q}}{\lVert \vec{q
   }\rVert}&= \vec{f}(\vec{b}_{1}, \dots, \vec{b}_{k}),
\end{align}
where $\vec{f}$ depends on the vectors $\vec{b}_i$. The main idea of this approach is that if one chooses $\vec{f}$ and $M_{q \rightarrow b}$ in an appropriate way, the $M_{r_i}$ may potentially be chosen such that
\begin{equation}
    M_{r_i} \frac{\vec{r}_i}{\lVert \vec{r}_i \rVert} = \frac{\vec{b}_i}{\lVert \vec{b}_i \rVert}, \quad i=1,\dots,k,
\end{equation}
in order for $\vec{p}$ to transform into $\vec{\tilde{p}}$ as desired. One way of `passing on the norm' of $\vec{q}$ in this way, i.e. choosing $\vec{f}$ and $M_{q \rightarrow b}$, is to pass it to each $\vec{b}_i$ individually
\begin{align}
    M_{q \rightarrow b}&= \sum_{i=1}^{k} \tilde{\alpha}_i (\mathds{1}+\Pi^i) M_{q \rightarrow b_i},\\
     \text{with}\ \ M_{q \rightarrow b_i} \frac{\vec{q}}{\lVert \vec{q} \rVert}&= \frac{\vec{b}_i}{\lVert \vec{b}_i \rVert}, \quad i=1,\dots,k,
\end{align}
which implies $\vec{f}(\vec{b}_{1}, \dots, \vec{b}_{k})=\sum_{i=1}^{k} \tilde{\alpha}_i \frac{\vec{b}_i}{\lVert \vec{b}_i \rVert}$, where the $\{\tilde{\alpha}_i\}_{i}$ constitutes a probability distribution. However, this is not always possible as the condition $\vec{q} / \lVert \vec{q} \rVert \succ \vec{b}_i / \lVert \vec{b}_i \rVert$, necessary for the existence of $M_{q \rightarrow b_{i}}$, cannot be satisfied in general. See Appendix~\ref{app:claim:qmajri} for details. One is therefore forced to pass the norm of $\vec{q}$ in a more clever way. A good candidate is
\begin{align}
    M_{q \rightarrow b}&= \sum_{i=1}^{k} \tilde{\alpha}_i M_{q \rightarrow 2b_i},\\
     \text{with}\ \ M_{q \rightarrow 2b_i} \frac{\vec{q}}{\lVert \vec{q} \rVert}&= \frac{(\mathds{1}+\Pi^i)\vec{b}_i}{2 \lVert \vec{b}_i \rVert}, \quad i=1,\dots,k.
\end{align}
However, this separation of $M_q$ into one part ``passing on the norm'' from $\vec{q}$ to $\vec{a}$ and another part from $\vec{q}$ to the $\vec{b}_i$ is not strictly needed. One may instead only require the existence of a doubly stochastic matrix $M_q$ such that
\begin{equation}\label{eq:generalMq}
    M_q \vec{q}= \vec{a} + \tilde{\vec{f}}(\vec{r}_1,\vec{b}_1,\dots,\vec{r}_k,\vec{b}_k),
\end{equation}
with
\begin{equation}
    \tilde{\vec{f}}(\vec{r}_1,\vec{b}_1,\dots,\vec{r}_k,\vec{b}_k)=\sum_{i=1}^{k} (\lfloor\tfrac{2i}{d}\rfloor+1)^{-1} (1-\frac{\lVert \vec{r}_i \rVert }{\lVert \vec{b}_i \rVert }) (\mathds{1}+\Pi^i) \vec{b}_i.
\end{equation}
Unfortunately, Eq.~(\ref{eq:generalMq}) is challenging to check, e.g., even confirming whether $\vec{a} + \tilde{\vec{f}}(\vec{r}_1,\vec{b}_1,\dots,\vec{r}_k,\vec{b}_k)$ is a vector of nonnegative components is already complicated. In summary, the approach requires checking the two following conditions:
\begin{enumerate}[(i)]
    \item \label{passing con i} For any pair of vectors $(\vec{r}_i,\vec{b}_i)$, there exists a doubly stochastic matrix $M_{r_i}$ such that
    \begin{equation}
       \frac{\vec{b}_i}{\lVert \vec{b}_i\rVert} =M_{r_i} \frac{\vec{r}_i}{\lVert \vec{r}_i \rVert},\; i=0,\dots,k,
    \end{equation}
    while at the same time $\lVert \vec{q} \rVert \geq \lVert \vec{ a} \rVert$, where $\vec{a}:=\vec{q}(\beta')$.
    \item \label{passing con ii} There exists a doubly stochastic matrix $M_q$ such that
\begin{equation}
    M_q \vec{q}= \vec{a} + \sum_{i=1}^{k} (\lfloor\tfrac{2i}{d}\rfloor+1)^{-1} (1-\frac{\lVert \vec{r}_i \rVert }{\lVert \vec{b}_i \rVert }) (\mathds{1}+\Pi^i) \vec{b}_i,
\end{equation}
or the stronger version: $ \lVert \vec{r}_i\rVert \leq \lVert \vec{b}_i\rVert$ for all $i=1,\dots,k$ and there exist doubly stochastic matrices $M_{q \rightarrow 2 b_i},\; i=1,\dots,k$ such that
\begin{equation}\label{eq:passnormcondII2}
     M_{q \rightarrow 2b_i} \frac{\vec{q}}{\lVert \vec{q} \rVert}= \frac{(\mathds{1}+\Pi^i)\vec{b_i}}{2 \lVert \vec{ b_i} \rVert} ~~~~\, i=1,\dots,k.
\end{equation}
\end{enumerate}
Condition~(\ref{passing con i}) indeed holds, and we thus state it below as the following lemma, the proof of which is presented in Appendix~\ref{app:PassNorm}.

\begin{lemma}\label{lemma:GenPassNormD}
For any $d$-dimensional system with arbitrary Hamiltonian, for any $\beta'< \beta$, and for all $i=0,\dots, d-1$ we have\\
\vspace*{-6mm}
\begin{align}
    \lVert \vec{q} \rVert &\geq\, \lVert \vec{ a} \rVert\,,
    \label{lemma:rismallerbi}\\
    \frac{\vec{r}_i}{\lVert \vec{r}_i \rVert} &\succ\,  \frac{\vec{b}_i}{\lVert \vec{b}_i \rVert}\,.
    \label{lemma:rimajbi}
\end{align}
\end{lemma}

Given the above lemma and assuming that condition~(\ref{passing con ii}) also holds, one can prove the existence of STUs for all dimensions by using the following majorisation relations. The majorisation relation (\ref{lemma:rimajbi}) from Lemma~\ref{lemma:GenPassNormD} ensures the existence of doubly stochastic matrices $M_{r_i}$ which satisfy $M_{r_i}\, \vec{r}_i / \lVert \vec{r}_i\rVert=\vec{b}_i / \lVert \vec{b}_i\rVert$. To reach any thermal state with higher temperature, i.e., $\vec{\tilde{p}}=\, \vec{p}(\beta\pr)$, we then need the existence of $M_q$ such that
\begin{equation}
 \vec{p}(\beta^{\prime})\,=\,M_{q}\,\vec{q}+\sum_{i=1}^{k} (\lfloor\tfrac{2i}{d}\rfloor+1)^{-1}\frac{\lVert \vec{r}_i \rVert }{\lVert \vec{b}_i \rVert } (\mathds{1}+\Pi^i)\, \vec{b}_i,
\end{equation}
therefore compensating for the required norm of the vectors associated to the subspace $\{\mathcal{H}_{r_i}\}_{i=1}^{d-1}$. This is precisely what is ensured by condition~(\ref{passing con ii}).

However, since condition~(\ref{passing con ii}) is in general cumbersome to prove, one can check whether Eq.~(\ref{eq:passnormcondII2}) holds instead. See Appendix~\ref{app:strong condition II} for more details about how to prove the existence of STUs in general in this way. In the particular case of $d=3$, the strong version of condition~(\ref{passing con ii}) can be reduced to a (``norm passing'') requirement that we formulate in the following lemma:

\begin{lemma}\label{lemma:PassNorm3d}
In $d=3$, for every choice of $E_{1}$ and $E_{2}$ with $E_{2}\geq E_{1}$ and for any $\beta\pr\leq\beta$, the following majorisation relation holds:\\
\vspace*{-4mm}
\begin{align}
    \frac{\vec{q}}{\lVert \vec{q} \rVert} \succ \frac{(\mathds{1} + \Pi)\vec{b}}{2 \lVert \vec{ b} \rVert}:=\frac{\vec{c}}{\lVert \vec{c} \rVert}.
\end{align}
\end{lemma}

The proof of Lemma~\ref{lemma:PassNorm3d} is presented in Appendix~\ref{App:proof:passnorminD}. Together, Lemmas~\ref{lemma:GenPassNormD} and~\ref{lemma:PassNorm3d} then confirm the existence of STUs for the $d=3$ case.

Now, let us examine the case $d=4$. Consider a bipartite system with equal local Hamiltonians $H=\sum_{i=0}^{3}E_i\ket{i}\!\!\bra{i}$ in which the energy eigenvalues are ordered in increasing order and $E_{0}=0$. Following Eq.~(\ref{deco. marginal}), the vectorised form of the marginals for the initial uncorrelated thermal state is given by $\vec{p}= \vec{q} +\vec{r}_1+ \vec{r}_2+\vec{r}_3$ with
\begin{equation}
    \vec{q}=\begin{pmatrix}p_{00}\\ p_{11}\\ p_{22} \\ p_{33} \end{pmatrix}, \; \vec{r}_1=\begin{pmatrix}p_{01}\\ p_{12}\\ p_{23} \\ p_{30} \end{pmatrix},\;\vec{r}_2=\begin{pmatrix}p_{02}\\ p_{13}\\ p_{20} \\ p_{31} \end{pmatrix},\; \vec{r}_3=\begin{pmatrix}p_{03}\\ p_{10}\\ p_{21} \\ p_{32} \end{pmatrix}=\Pi \vec{r}_1,
    \label{eq:d4 vectors}
\end{equation}
where $\vec{r}_i$ is a shorthand for $\vec{r}_i(\beta)$. Furthermore, we again denote the vector decomposition of any thermal state with higher temperature $\beta\pr \leq \beta$ as $\vec{p}(\beta\pr)=\vec{a}+\sum_{i}\vec{b}_{i}$ with $\vec{a}:=\vec{q}(\beta')$
and $\vec{b}_i:=\vec{r}_i(\beta\pr)$. Unitaries on the LCSs that generate the same marginals, according to Eq.~(\ref{eq:EqualEvoMarginal C}), lead to the transformation
\begin{equation}
 \vec{p}\mapsto\vec{\tilde{p}}\,=\,M_{q}\,\vec{q}+(\openone+\Pi)\,M_{r_1}\vec{r}_1+\frac{1}{2}(\openone+\Pi^2)M_{r_2}\,\vec{r}_2,
 \label{eq:gen. trans. 4d}
\end{equation}
where the $M_{r_i}$ are arbitrary doubly stochastic matrices. To achieve a thermal marginal with higher temperature, one can transform each $\vec{r}_i$ to $\vec{b}_i$ as prescribed by condition~\ref{passing con i} and pass the extra norm of the vector $\vec{q}$ as dictated by the stronger version of condition~(\ref{passing con ii}). This is possible at least under some restrictions on the energy level spacings $\delta_i:=E_{i+1}-E_i$. Let us phrase this statement more precisely in the following Lemma~\ref{lemma:PassNorm4d}, a detailed proof of which is presented in Appendix~\ref{App:proof:passnorminD}.

\begin{lemma}\label{lemma:PassNorm4d}
In $d=4$, for every choice of the set $\{E_{i}\}_{i=0}^{3}$ with $E_{i+1}\geq E_i$ and $\delta_{i+1} \leq \delta_{i}$ and for any $\beta\pr\leq\beta$, the following relations hold for $i=1,2$:
\begin{align}
    \frac{\vec{q}}{\lVert \vec{q} \rVert} &\succ\, \frac{\mathds{1}+\Pi^i}{2} \frac{\vec{b}_i}{\lVert \vec{b}_i \rVert},
    \label{lemma:qmaj2ri4d}\\[1mm]
    \lVert \vec{r}_i \rVert &\leq\, \lVert \vec{b}_i \rVert.
    \label{lemma:rismallerbi4d}
\end{align}
\end{lemma}

With Lemma~\ref{lemma:PassNorm4d} at hand, we are ready to state the result achieved for the $d=4$ with this method.\\

\begin{Theorems}{Existence of STUs in $d=4$}{PassNorm4d}
In any $4$-dimensional system, for every set of energy eigenvalues $\{E_{i}\}_{i=0 }^3$ with $E_{i+1}\geq E_i$ and $\delta_{i+1} \leq \delta_{i}$, and for any $\beta\pr\leq\beta$, there exists a set of doubly stochastic matrices $\{M_{r_i}\}_{i=0}^2$ such that
\begin{equation}
\vec{p}(\beta\pr)=\,M_{q}\,\vec{q}+(\openone+\Pi)\,M_{r_1}\vec{r}_1+\frac{1}{2}(\openone+\Pi^2)M_{r_2}\,\vec{r}_2,
\end{equation}
which implies the existence of STUs in $d=4$ for symmetric Hamiltonians with decreasing energy gaps.
\end{Theorems}

\textbf{Proof}. The statement follows from Lemmas~\ref{lemma:GenPassNormD} and~\ref{lemma:PassNorm4d}. \qed

We have thus seen that the approach discussed in this section does, at least partially, generalise to local dimension $4$. However, the proof of Theorem~\ref{theorem:PassNorm4d}, which exploits the stronger version of condition~(\ref{passing con ii}) fails for $E_3 \gg 1$. Note that this failure is not an artifact of even dimensions, but rather continues to persist in   subsequent higher dimensions. Still, the possibility remains that above results can be generalised by proving the weaker version of condition~(\ref{passing con ii}).


\subsection{Geometric approach}\label{sec:two-qutrit case_geometric}

Our attempts to generalise the approaches of Secs.~\ref{sec:commutingM} and~\ref{sec:3_passing} to higher dimensions have shown that the problem can be recast in terms of different sets of conditions. However, checking these conditions has proven to be increasingly complex with growing dimension, and has thus only provided partial results even for the case $d=4$. We therefore now turn to a third approach which at least provides a complete proof for $d=4$. This approach is centred around the geometric structure generated by doubly stochastic matrices. More specifically, recall that the Schur-Horn theorem implies that for any $v\in\mathbb{R}^{n}$, the set of vectors obtained by applying the set of unistochastic $n\times n$ matrices to $v$ is a convex polytope given by the convex hull of all permutations of the entries of $v$, see, e.g., Ref.~\cite{BengtssonEricssonKusTadejZyczkowski2005}. Here, we can apply this idea to the vectors $\vec{q}$ and $\vec{r}_{i}$ and the matrices $M_{q}$ and $M_{r_{i}}$, respectively, to generate matching diagonal marginals according to Eq.~(\ref{eq:EqualEvoMarginal C}). In other words, the set of all possible symmetric marginal vectors reachable by unitaries that are block-diagonal with respect to the chosen LCS decomposition is the polytope with vertices given by the set of points
\begin{align}
 \mathbf{\tilde{p}}\suptiny{0}{0}{(i_{0},i_{1},\ldots,i_{k})}
    &= \Pi\suptiny{0}{0}{(i_{0})} \mathbf{q} +
    \sum_{n=1}^{k}
    (\lfloor\tfrac{2n}{d}\rfloor+1)^{-1}
    (\mathds{1}+ \Pi^{n})\,\Pi\suptiny{0}{0}{(i_{n})} \mathbf{r}_{n},
\label{eq:EqualEvoMarginal}
\end{align}
with $i_{j}\in\{1,2,\ldots,d!\}$ for all $j=0,\ldots,k$ with $k$ as in Eq.~(\ref{gen evol marginal A}). Here, $\Pi\suptiny{0}{0}{(i)}$ for $i=1,\ldots,d!$ are the possible permutations of $d$ elements. The question about the existence of STUs is then equivalent to asking if the curve defined by the set of points $\{\vec p(\beta\pr)|\beta\geq\beta\pr\geq0\}$ is enclosed within the polytope corresponding to the convex hull of the $(d!)^{k+1}$ points $\mathbf{\tilde{p}}\suptiny{0}{0}{(i_{0},i_{1},\ldots,i_{k})}$.

To answer this question, we then proceed in the following way. First, we note that the $d$-component vectors $\mathbf{p}=(p_{i})$ only have $d-1$ independent entries due to the normalisation condition $\sum_{i=0}^{d-1}p_{i}=1$. The problem can thus be reduced to $d-1$ dimensions by choosing coordinates $\{x_{i}\}_{i=0,1,\ldots,d-2}$ with
\begin{align}
    x_{n}   &=\,(n+1)p_{n+1}-\sum_{i=0}^{n}p_{i},
    \label{eq:new coords}
\end{align}
for $n=0,1,\ldots,d-2$, while the additional last coordinate $x_{d-1}=-\sum_{i=0}^{d-1}p_{i}=-1$ is fixed by normalisation and can thus be disregarded. In these coordinates, the point obtained for infinite temperature ($\beta\rightarrow0$) is the origin, $p(\beta\rightarrow0)=(0,0,\ldots,0,-1)$, the point for vanishing temperature is $p(\beta\rightarrow\infty)=(-1,-1,\ldots,-1)$ and all thermal states lie on a continuous curve connecting these points that is strictly confined to negative coordinate regions, $x_{i} \,\leq\, 0~\forall i$. However, note that there are points corresponding to reachable marginals that have positive values for some of the new coordinates.

For any given initial temperature $1/\beta$ and dimension $d$, a sufficient set of conditions for the inclusion of the curve defined by the points $\{\vec p(\beta\pr)|\,\beta\geq\beta\pr\geq0\}$ in the reachable polytope is as follows:
\begin{enumerate}[(I)]
\item{\label{item i}Inclusion of the vertices
\begin{align}
    \vec{v}_0&:=\vec{p}(\beta)=(x_{0}(\beta),x_{1}(\beta),\ldots,x_{d-3}(\beta),x_{d-2}(\beta),-1) , \nonumber\\
    \vec{v}_1&:=(0,x_{1}(\beta),\ldots,x_{d-3}(\beta),x_{d-2}(\beta),-1) ,  \nonumber\\
    \vec{v}_2&:=(0,0,x_2(\beta)\ldots,x_{d-3}(\beta),x_{d-2}(\beta),-1) ,  \nonumber\\
    &\ \ \vdots  \nonumber\\
    \vec{v}_{d-2}&:=(0,0,\ldots,0,x_{d-2}(\beta),-1)  ,  \nonumber\\
    \vec{v}_{d-1}&:=\vec{p}(\beta\rightarrow 0) = (0,\dots,0,-1) ,
    \label{eq:polyvertgeom}
\end{align}
in the polytope of achievable marginals.
}
\item{\label{item ii}Inclusion of all points $\vec p(\beta\pr)$ with $\beta\pr\leq \beta$ in the simplex corresponding to convex hull of the set of vertices $\{\vec{v}_{i}|\,i=0,\ldots,d-1\}$.
}
\end{enumerate}

\begin{figure*}
(a)\includegraphics[width=0.5\textwidth,trim={0cm 0mm 0cm 0mm}]{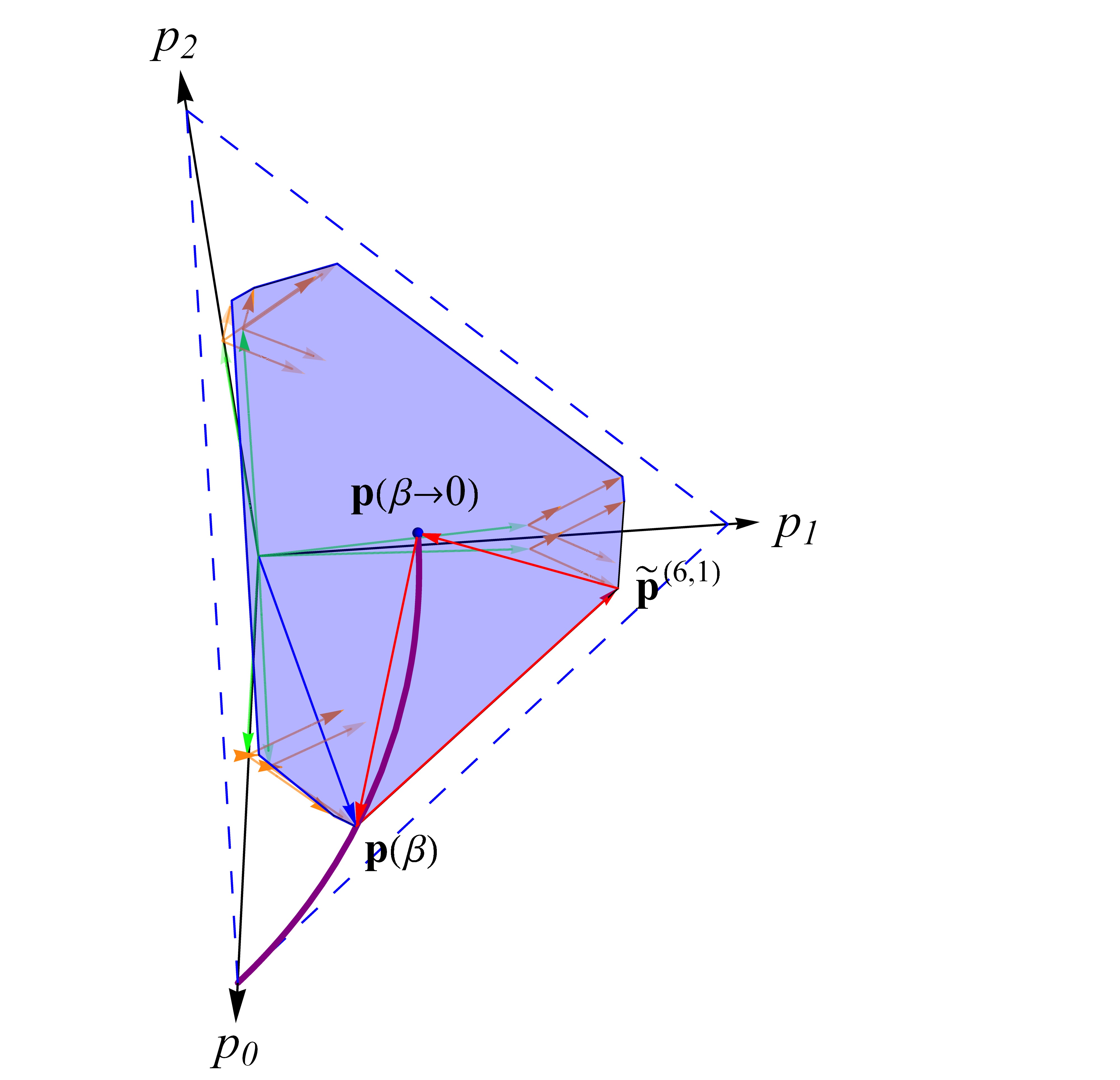}
(b)\ \includegraphics[width=0.35\textwidth,trim={0cm 0mm 0cm 0mm}]{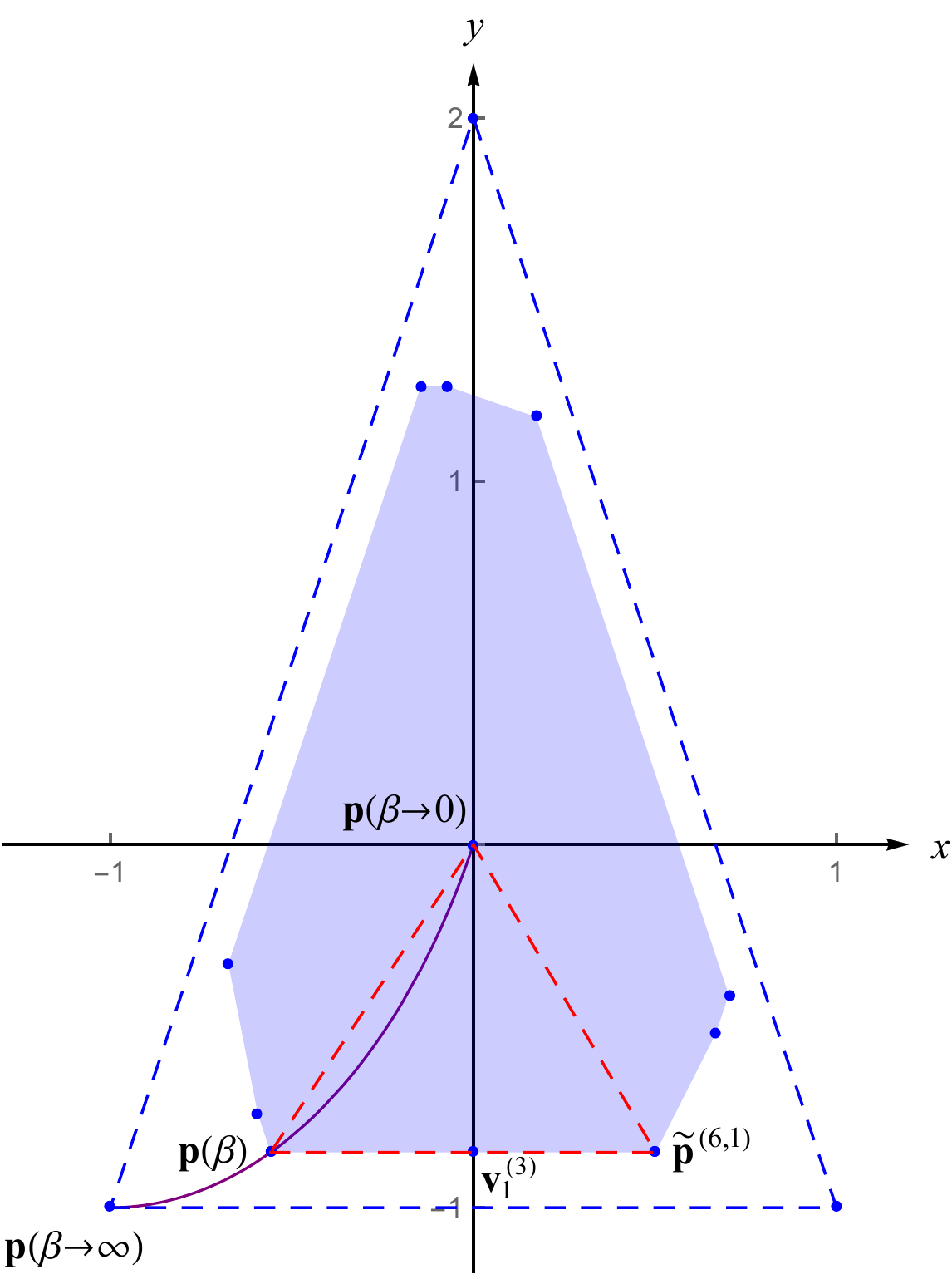}
\caption{\textbf{Polytope of reachable symmetric diagonal qutrit marginals}. (a) The axes show the components $p_{i}$ (with $i=0,1,2$) with respect to the original basis $\{\vec{e}_{i}\}_{i=0,1,2}$ and the simplex of all possible states (not necessarily reachable from a given state) is indicated by the dashed blue triangle. The parameter values chosen for the illustration are $\beta=1.35 E_{1}$ and $E_{2}=2 E_{1}$. For these values, the family of thermal states is shown as a solid purple curve from $\mathbf{p}(\beta\rightarrow\infty)=(1,0,0)^{T}$ to $\mathbf{p}(\beta\rightarrow0)=(\tfrac{1}{3},\tfrac{1}{3},\tfrac{1}{3})^{T}$. The initial state of the marginals is represented by the point $\mathbf{p}(\beta)$ indicated by a blue arrow. The shaded blue area shows the polytope of diagonal reduced states with $\tilde{\mathbf{p}}\SA=\tilde{\mathbf{p}}\SB$ that is reachable by the application of a unistochastic matrix $M$ on $\mathbf{q}$ in combination with a circulant unistochastic matrix $\tilde{M}$ applied to $\mathbf{r}$. This polytope is the convex hull of the points obtained from combining any (out of 6 possible) permutations of $q$ (green arrows) with cyclic permutations (out of 3 possible) of $\mathbf{r}$ (orange arrows). We have chosen to restrict to cyclic permutations on $\mathbf{r}$ here to illustrate that this is enough for $d=3$, whereas this is no longer the case when $d=4$. The red arrows between $\mathbf{p}(\beta)$, $\tilde{\mathbf{p}}\suptiny{0}{0}{(6,1)}$, and $\mathbf{p}(\beta\rightarrow0)$ delineate a triangle, which we show contains all points $\mathbf{p}(\beta\pr)$ that correspond to thermal states with temperatures higher than the original temperature, $\beta>\beta\pr$. (b) The polytope is shown in terms of the independent coordinates $x\equiv x_{0}=-p_{0}+p_{1}$ and $y\equiv x_{1}=-p_{0}-p_{1}+2p_{2}$ from Eq.~(\protect\ref{eq:new coords}). The vertex $\vec{v}\suptiny{0}{0}{(3)}_{1}$ is located at the intersection of the $y$ axis with the line connecting $\mathbf{p}(\beta)$ and $\tilde{\mathbf{p}}\suptiny{0}{0}{(6,1)}$.}
\label{fig:polytope}
\end{figure*}

In general, confirming condition~(\ref{item ii}) is relatively straightforward, either by proving the positivity of the partial derivatives
$\tfrac{\partial x_{m}}{\partial x_{n}}   =\bigl(\tfrac{\partial x_{n}}{\partial \beta}\bigr)^{-1}\tfrac{\partial x_{m}}{\partial \beta}\geq0$ and $\tfrac{\partial^{2} x_{m}}{\partial x_{n}^{2}}   =\tfrac{\partial^{2} x_{m}}{\partial \beta^{2}}- \bigl(\tfrac{\partial x_{m}}{\partial \beta}\bigr)\bigl(\tfrac{\partial x_{n}}{\partial \beta}\bigr)^{-1} \tfrac{\partial^{2} x_{n}}{\partial \beta^{2}}\geq0$ (which we show explicitly for $d=4$ in Appendix~\ref{app:CEforConditionii}), or by showing that all points $\vec p(\beta\pr)$ can be written as a convex combination of the vertices (\ref{eq:polyvertgeom}), which we prove in general for dimension $d$ in Appendix~\ref{app:CEforConditionii}. Let us therefore state condition~(\ref{item ii}) as the following lemma:
\begin{lemma}\label{lemma:condIgeom}
In the coordinates defined by Eq.~(\ref{eq:new coords}), the curve of thermal states at inverse temperature $\beta\pr\leq\beta$ satisfies condition~(\ref{item ii}) for all $\beta$ in all dimensions $d$.
\end{lemma}

Besides satisfying condition~(\ref{item ii}), it is also easy to see that the vertices $\vec{v}_{0}$ and $\vec{v}_{d-1}$ from condition~(\ref{item i}), which represent the initial thermal state and maximally mixed state in the new coordinates, are reachable for all dimensions. However, showing the inclusion of the rest of the points in~(\ref{item i}) is increasingly difficult, due to the rapidly growing number of possible polytope vertices and the difficulty of visualizing the $(d-1)$-dimensional polytope beyond $d=4$. In the following, we prove that condition~(\ref{item i}) holds (at least) in the particular cases of $d=3$ and $d=4$. See also Fig.~\ref{fig:polytope} for an illustration in dimension $3$ and Appendices~\ref{sec:two qutrit case} and~\ref{sec:two ququart case} for further details.

\begin{Theorems}{Geometric approach in $d=3$ and $d=4$}{dim3result}
In $d=3$ and $d=4$ systems, for every choice of Hamiltonians and initial inverse temperature $\beta$, the set of thermal states
with $\beta\pr \leq \beta$ is contained within the polytope with vertices defined in Eq.~(\ref{eq:polyvertgeom}), which proves the existence of STUs in the symmetric two-qutrit and two-quqart cases.
\end{Theorems}

\textbf{Proof.}\ For $d=3$ we have to prove that the point $\vec{v}\suptiny{0}{0}{(3)}_{1}=(0,x_1(\beta),-1)$ can be reached with transformations of the type $\vec{p}\mapsto\tilde{\vec{p}}=M_{q}\vec{q}+(\openone +\Pi)M_r\vec{r}$ for some doubly stochastic $3\times3$ matrices $M_{q}$ and $M_{r}$, where $\vec{q}$ and $\vec{r}$ are as in Eq.~(\ref{vecqrs}). This is indeed the case, e.g., for $M_r=\openone$ being the identity matrix and
\begin{align}
    M_{q}   &=
    \begin{pmatrix} m & 1-m & 0 \\ 1-m & m & 0 \\ 0 & 0 & 1 \end{pmatrix},
\end{align}
with $m=1-1/[2(p_{0}+p_{1})]$, which is a doubly stochastic matrix since $p_{0}+p_{1}\geq 1/2$ for $d=3$. Note that, geometrically, this choice of $M_{q}$ and $M_{r}$ represents a convex combination of the points $\vec{p}(\beta)$ and $\tilde{\vec{p}}^{(6,1)}$ shown in Fig.~(\ref{fig:polytope}). For this transformation, we have $\tilde{\vec{p}}=\sum_{i=0}^{2}\tilde{p}_{i}\vec{e}_{i}$, where the components with respect to the original basis $\{\vec{e}_{j}\}_{j=0,1,2}$ are $\tilde{p}_{0}=\tilde{p}_{1}=(p_{0}+p_{1})/2$ and $\tilde{p}_{2}=p_{2}$, which means that the components of $\tilde{\vec{p}}$ with respect to to the new coordinates of Eq.~(\ref{eq:new coords}) are $\tilde{\vec{p}}=(0,2-3(p_{0}+p_{1}),-1)^{T}=(0,x_{1}(\beta),-1)^{T}=\vec{v}\suptiny{0}{0}{(3)}_{1}$. This concludes the proof for $d=3$.

For $d=4$, we have to show that the two points $\vec{v}\suptiny{0}{0}{(4)}_{1}=(0,x_1(\beta),x_2(\beta),-1)$ and
$\vec{v}\suptiny{0}{0}{(4)}_{2}=(0,0,x_2(\beta),-1)$ can be reached with transformations of the type $\vec{p}\mapsto\tilde{\vec{p}}=M_q\vec{q}+(\openone+\Pi)M_{r_1}\vec{r}_1+\tfrac{1}{2} (\openone+\Pi^2)M_{r_2}\vec{r}_2$ for some doubly stochastic $4\times4$ matrices $M_{q}$, $M_{r_{1}}$, and $M_{r_{2}}$, where $\vec{q}$ and $\vec{r}_{1}$ are as in Eq.~(\ref{eq:d4 vectors}).

The point $\vec{v}\suptiny{0}{0}{(4)}_{1}$ can be reached with the equivalent of the transformation used to reach $\vec{v}\suptiny{0}{0}{(3)}_{1}$ above, that is, using $M_{r_{1}}=M_{r_{2}}=\openone$ and
\begin{align}
    M_{q}   &=
    \begin{pmatrix}
    m & 1-m & 0 & 0 \\
    1-m & m & 0 & 0 \\
    0 & 0 & 1 & 0 \\
    0 & 0 & 0 & 1 \end{pmatrix}  ,
\end{align}
with $m=1-1/[2(p_{0}+p_{1})]$, which is again doubly stochastic since $p_{0}+p_{1}\geq 1/2$ also holds for $d=4$.

To prove that $\vec{v}\suptiny{0}{0}{(4)}_{2}$ can be reached, one can see that it can be obtained as a convex combination of (at most) $5$ of the vertices in Eq.~(\ref{eq:EqualEvoMarginal}), as we show in detail in Appendix~\ref{sec:two ququart case}. \qed

For higher dimensions, as we already mentioned, the problem becomes more and more complex, but one can try to build a recursive approach based on the above lower dimensional proofs and borrowing some ideas from the ``passing on the norm'' approach. In particular, we outline a possible route for such an approach for the case $d=5$ in Appendix~\ref{app:higherdgeom}, where we can show the existence of STUs for $d=5$ for a subset of all possible Hamiltonians.


\section{Conclusions}\label{sec:conclusion}

We have investigated the generation of correlations in initially thermal, uncorrelated systems. For two identical $d$-dimensional systems, the conversion of energy into correlations as measured by the mutual information is optimal when the final state can be reached unitarily and both marginals of the final state are thermal at the same effective temperature. For any given system, the possibility of such an optimal conversion for all input energies (or desired amout of correlations) hence hinges upon the existence of symmetrically thermalizing unitaries for all initial temperatures and effective local final temperatures. This gives rise to the central question: \emph{Is it possible to find unitaries (STUs) transforming thermal marginals to other thermal marginals with higher temperature for any local Hamiltonian?}

In asymmetric cases, where the two local Hamiltonians are different, this is generally not possible, as we have shown via constraints on the subsystem entropies. For the symmetric case (equal local Hamiltonians), we have provided a framework based on locally classical subspaces in $d\times d$-dimensional systems to address this question beyond previous partial results for equally gapped Hamiltonians~\cite{HuberPerarnauHovhannisyanSkrzypczykKloecklBrunnerAcin2015}. In particular, we have shown that STUs exist for all (locally matching) Hamiltonians in local dimensions $d=3$ and $d=4$, and we conjecture that STUs exist in all local dimensions.

To showcase the complexity and interesting features of the problem, as well as to provide further guidance for proving (or disproving) our conjecture, we have discussed three approaches operating within our framework. Using the ``majorised marginals" approach we showed for two qutrits ($d=3$) that, not only do STUs generically exist for any local Hamiltonian at any temperature, but also it is indeed possible to symmetrically reach any marginal that is majorised by the initial marginals. However, since this approach fails to be generalised to higher dimensions, we introduced two alternative approaches that we call ``passing on the norm" and ``geometric approach", respectively. Both allow proving the existence of STUs in the two-qutrit case. Using the ``passings on the norm" approach, we were further able to show that STUs exist for $d=4$ when the local Hamiltonians satisfy specific conditions on their energy gaps, i.e., $\delta_{i+1}\leq \delta_i$. Finally, we have used the ``geometric" method to prove the existence of STUs in local dimension $d=4$ for all symmetric Hamiltonians, and we formulate a set of conditions to extend this approach to higher dimensions.

Our work addresses a fundamental question in quantum thermodynamics, whether correlations can always be created energetically optimally, or not. Besides addressing a question about the conversion between thermodynamic and information-theoretic resources, the problem at hand can be considered a part of the quantum marginal problem. What kind of marginals can be unitarily reached from (are compatible with) a particular global state? The framework we put forward in terms of locally classical subspaces is more general than symmetric marginal transformations and also goes beyond mere majorisation relations of marginal eigenvalues. As such, it may also be relevant for other variants of this question, such as addressing the catalytic entropy conjecture of~\cite{BoesEisertGallegoMuellerWilming2019}. Just using one of the many possible such subspaces, we managed to resolve our main question for dimensions $d=3$ and $d=4$ and it would be interesting to know if all marginal eigenvalue distributions could potentially be reached by operating in locally classical subspaces only. Finally, a significant challenge lies in specializing from the creation of arbitrary correlations to generating entanglement~\cite{HuberPerarnauHovhannisyanSkrzypczykKloecklBrunnerAcin2015, BruschiPerarnauLlobetFriisHovhannisyanHuber2015, PiccioneMilitelloNapoliBellomo2019, GuhaAlimuddinParashar2019}.


\begin{acknowledgments}
We are grateful to Paul Boes, Mirdit Doda, Christian Gogolin, Claude Kl{\"o}ckl, and Alex Monras for fruitful discussions. We acknowledge support by the Austrian Science Fund (FWF) through the START project Y879-N27, the Lise-Meitner project M 2462-N27, the project P 31339-N27, the Zukunftskolleg ZK03 and the joint Czech-Austrian project MultiQUEST (I 3053-N27 and GF17-33780L). F.B. acknowledges support by the Ministry of Science, Research, and Technology of Iran (through funding for graduate research visits) and Sharif University of Technology's Office of Vice President for Research and Technology through Grant QA960512. F.C. acknowledges funding from the Swiss National Science Foundation (SNF) through the AMBIZIONE grant PZ00P2$\_$161351 and Grant No. 200021$\_$169002.
\end{acknowledgments}


\bibliography{bibfile}


\hypertarget{sec:appendix}
\appendix
\renewcommand{\thesection}{}
\renewcommand\appendixname{}
\renewcommand{\thesection}{}
\renewcommand{\thesection}{A}
\renewcommand{\thesubsection}{A.\Roman{subsection}}
\renewcommand{\thesubsubsection}{A.\Roman{subsection}.\alph{subsubsection}}
\setcounter{equation}{0}
\numberwithin{equation}{section}


\section*{Appendices}

In these appendices, we provide detailed proofs of the lemmas supporting the main theorems, as well as additional detailed calculations and counterexamples mentioned in the main text. In Appendix~\ref{app:asymmetric pure state}, we investigate the maximal amount of correlations unitarily achievable for a fixed amount of energy that can be created between two arbitrary asymmetric systems initialised in a pure state. In Appendix~\ref{passive}, we propose a scheme to transform finitely many copies of bipartite thermal states with thermal marginals to states with symmetric thermal marginals at a higher temperature. In Appendix~\ref{App:Proof:lemma:swappingM}, we give a detailed proof of Lemma~\ref{lemma:swappingM}. We also discuss why Lemma~\ref{lemma:swappingM} cannot be generalised to higher dimensions in Appendix~\ref{sec:generalisations}. In Appendix~\ref{app:claim:qmajri}, we show via counterexample that it is in general not possible to map the normalised versions of the vectors $\vec{q}_i$ to the vectors $\vec{r}_i$ by doubly stochastic matrices. In Appendix~\ref{app:PassNorm}, we present a detailed proof of Lemma~\ref{lemma:GenPassNormD} that confirms that condition~(\ref{passing con i}), used in the ``passing on the norm'' approach, holds in general. In Appendix~\ref{app:strong condition II}, we discuss how one can show the existence of STUs via the stronger version of condition~(\ref{passing con ii}). In Appendix~\ref{App:proof:passnorminD}, by proving Lemma~\ref{lemma:PassNorm3d} and Lemma~\ref{lemma:PassNorm4d}, we complete the proof of the existence of STUs via this approach in $d=3$ and under specific constraints on the energy gaps in $d=4$. Then, we turn our attention to the geometric approach and show the monotonicity and convexity of the thermal curve in Appendix~\ref{app:CEforConditionii}. The detailed proofs of the existence of STUs using the geometric approach in dimensions $3$ and $4$ are presented in Appendix~\ref{sec:two qutrit case} and Appendix~\ref{sec:two ququart case}, respectively. Finally, we discuss the possibility of generalising the geometric method to higher dimensions in Appendix~\ref{app:higherdgeom}.


\subsection{Upper bound on correlation}\label{app:upperbound}

In this appendix, we show that if STUs exist in general, i.e., in particular for the desired temperature, they provide an upper bound for the amount of correlation that can be achieved unitarily. Using the same notation as in the main text, we are interested in solving the problem
\begin{equation}
    \max_U S(\tilde{\varrho}\SA) + S(\tilde{\varrho}\SB) \quad\ \text{s.t.}\ \ \tr(\tilde{\varrho}\SA H\SA) + \tr(\tilde{\varrho}\SB H\SB) \leq c.
\end{equation}
This can be rewritten as
\begin{align}
   & \max_U S(\tilde{\varrho}\SA \otimes \tilde{\varrho}\SB)\\
   &\text{s.t.}\ \ \tr(\tilde{\varrho}\SA \otimes \tilde{\varrho}\SB (H\SA \otimes \mathds{1} + \mathds{1} \otimes H\SB)) \leq c.
   \nonumber
\end{align}
According to Jaynes' principle, the maximum is obtained for
\begin{equation}
    \tau\SA(\bar{\beta}) \otimes \tau\SB(\bar{\beta}),
\end{equation}
where $\bar{\beta}$ is chosen such that
\begin{equation}
    \tr (\tau\SA(\bar{\beta}) \otimes \tau\SB(\bar{\beta}) (H\SA \otimes \mathds{1} + \mathds{1} \otimes H\SB)) = c.
\end{equation}
This solution can be found using Lagrange multipliers by considering the $n\times n$ matrix $\varrho\SAB$ as an vector with $n^{2}$ components. We therefore have as desired
\begin{align}
   &\max_U S(\tilde{\varrho}\SA) + S(\tilde{\varrho}\SB)\ \text{s.t.}\ \tr (\tilde{\varrho}\SA H\SA) + \tr (\tilde{\varrho}\SB H\SB) \leq c.
   \nonumber\\
  &\leq \, S(\tau\SA(\bar{\beta}))+S( \tau\SB(\bar{\beta})),
\end{align}
with $\tr (\tau\SA(\bar{\beta}) H\SA )+ \tr ( \tau\SB(\bar{\beta}) H\SB) = c$.


\subsection{Maximal amount of correlation for a pure state in the asymmetric case}\label{app:asymmetric pure state}

Here, we discuss and solve the problem of maximising the correlations under an energy constraint in the asymmetric case with an initial pure state. That is, we want to solve
\begin{align}
    \max_{\tilde{\varrho}\SA, \tilde{\varrho}\SB} S(\tilde{\varrho}\SA)+ S(\tilde{\varrho}\SB),
\end{align}
subject to the constraint $\tr(\tilde{\varrho}\SA H\SA) + \tr(\tilde{\varrho}\SB H\SB) \leq c$, where
\begin{align}
    \tilde{\varrho}\SA=\sum_{i=0}^{d-1} p_{i} \ket{\varphi\suptiny{0}{0}{A}_{i}}\!\!\bra{\varphi\suptiny{0}{0}{A}_{i}},\ \ \
    \tilde{\varrho}\SB=\sum_{i=0}^{d-1} p_{i} \ket{\varphi\suptiny{0}{0}{B}_{i}}\!\!\bra{\varphi\suptiny{0}{0}{B}_{i}},
\label{sub state A}
\end{align}
with $p_i \geq 0$ for $i=0, \dots, d-1$, $\sum_{i=0}^{d-1} p_{i}=1$, $d=\min\{d\SA,d\SB\}$, and where $\{\ket{\varphi\suptiny{0}{0}{A}_{i}}\}_{i=0}^{d\SA-1}$ and $\{\ket{\varphi\suptiny{0}{0}{B}_{i}}\}_{i=0}^{d\SB-1}$ are orthonormal bases of $\mathcal{H}\SA$ and $\mathcal{H}\SB$, respectively. Without loss of generality we then assume $d=d\SA$. Note that $S(\tilde{\varrho}\SA)=S(\tilde{\varrho}\SB)$ and $\tilde{\varrho}\SB = U \tilde{\varrho}\SA U^{\dagger}$, where we write $U= \sum_{i=0}^{d-1} \ket{\varphi\suptiny{0}{0}{B}_{i}}\!\!\bra{\varphi\suptiny{0}{0}{A}_{i}}$ (in a slight abuse of notation) such that the problem may be rewritten as
\begin{align}
\label{equ:maxS}
    \max_{\varrho, U} S(\varrho) \quad \text{s.t.}\ \ \ \tr[\varrho (H\SA+U^{\dagger} H\SB U)]\,\leq\, c,
\end{align}
where we have dropped the tilde and subscript $A$ on $\varrho$ for ease of notation. To solve this problem we consider the converse problem
\begin{align}
\label{equ:minEn}
    \min_{\varrho, U} \tr [ \varrho (H\SA+U^{\dagger} H\SB U)] \quad \text{s.t.}\ \ \ S(\varrho)= \kappa,
\end{align}
and show that (at least a family of) optimal points of (\ref{equ:minEn}) are optimal points of (\ref{equ:maxS}). To simplify the notation further let us write
\begin{align}
    H\SA &= \sum_{i=0}^{d-1} E_{i}\suptiny{0}{0}{A} \ket{i}\!\!\bra{i}\SA\\
    H\SB &= \sum_{i=0}^{d_B-1} E_{i}\suptiny{0}{0}{B} \ket{i}\!\!\bra{i}\SB,
\end{align}
with $E_{i}\suptiny{0}{0}{A} \leq E_{i+1}\suptiny{0}{0}{A}$ and $E_{i}\suptiny{0}{0}{B} \leq E_{i+1}\suptiny{0}{0}{B}$.

\begin{prop}
The pair $(\varrho_{\mathrm{opt}}(\kappa), U_{\mathrm{opt}})$, given by
\begin{align}
    U_{\mathrm{opt}}    &:= \sum_{i=0}^{d-1} \ket{i}\SB \bra{i}\SA,\\
    \varrho_{\mathrm{opt}}(\beta(\kappa))&:= \frac{e^{-\beta(\kappa) \tilde{H}}}{\tilde{Z}},
\end{align}
where $\tilde{Z}=\mathrm{Tr}(e^{-\beta(\kappa) \tilde{H}})$ and $\tilde{H}= \sum_{i=0}^{d-1} (E_{i}\suptiny{0}{0}{A}+E_{i}\suptiny{0}{0}{B}) \ket{i}\!\!\bra{i}\SA$, is a solution of the minimisation in~\emph{(\ref{equ:minEn})}.
\end{prop}

\begin{proof}
Denoting the spectrum of $\varrho$ as $\lambda_{\varrho}=(\lambda_0,\dots,\lambda_{d-1})$, we first show that $(\varrho_{\mathrm{opt}}(\beta(\kappa)), U_{\mathrm{opt}})$ is a solution of the following minimisation problem
\begin{align}
    \min_{\lambda} \left( \min_{\varrho\ \text{s.t.}\ \lambda_{\varrho}=\lambda} \tr (\varrho H\SA) +  \min_{U, \varrho\ \text{s.t.}\ \lambda_{\varrho}=\lambda} \tr (U \varrho U^{\dagger} H\SB) \right),
    \label{equ:minEdecomp}
\end{align}
subject to the constraint $H(\lambda)=\kappa$, where $H(\lambda)$ denotes the Shannon entropy of the probability distribution $\lambda$. Since all density matrices of a given spectrum are unitarily related we have
\begin{equation}
     \min_{\varrho\ \text{s.t.}\ \lambda_{\varrho}=\lambda} \tr (\varrho H\SA)= \min_{V} \tr(V \operatorname{diag}\{\lambda_\varrho\} V^{\dagger} H\SA),
\end{equation}
where the minimisation on the right-hand side is over all unitaries $V$. The passive state with spectrum $\lambda$ is well-known to solve this minimisation. Adopting the notation $v^{\downarrow}=(v^{\downarrow}_{i})$ to denote the vector obtained by arranging the components of the vector $v=(v_1,\dots,v_n)$ in decreasing order, i.e., such that $v_1^{\downarrow} \geq v_2^{\downarrow} \geq \dots \geq v_n^{\downarrow}$, we thus have
\begin{equation}
     \min_{\varrho\ \text{s.t.}\ \lambda_{\varrho}=\lambda} \tr (\varrho H\SA)= \sum_{i=0}^{d-1} \lambda_{i}^{\downarrow} E_{i}\suptiny{0}{0}{A}.
\end{equation}

For the second minimisation problem in (\ref{equ:minEdecomp}), note that since $\{\ket{\varphi_{i}\suptiny{0}{0}{A}}\}_{i=0}^{d\SA-1}$ and $\{\ket{\varphi_{i}\suptiny{0}{0}{B}}\}_{i=0}^{d\SB-1}$ are orthonormal bases, the matrix representation of $U$ can be extended to a unitary $d\SB \times d\SB$ matrix $V$. Similarly, the matrix representation of $\varrho$ can be extended to a positive $d\SB\times d\SB$ matrix $\bar{\varrho}$ by padding it with zeroes. With this we have
\begin{equation}
    \min_{U, \varrho\ \text{s.t.}\ \lambda_{\varrho}=\lambda} \tr (U \varrho U^{\dagger} H\SB)\geq \min_{V, \bar{\varrho}\ \text{s.t.}\  \lambda_{\bar{\varrho}}=(\lambda,0,\dots,0)} \tr (V \bar{\varrho} V^{\dagger} H\SB).
    \label{eq:fabien's proof 1}
\end{equation}
Again the passive state with spectrum $(\lambda,0,\dots,0)$ solves the right-hand side of~(\ref{eq:fabien's proof 1}) and we obtain
\begin{equation}
    \min_{V, \bar{\varrho}\ \text{s.t.}\ \lambda_{\bar{\varrho}}=(\lambda,0,\dots,0)} \tr (V \bar{\varrho} V^{\dagger} H\SB)
    = \sum_{i=0}^{d-1} \lambda_{i}^{\downarrow} E_{i}\suptiny{0}{0}{B}.
\end{equation}
This solution is in fact also an attainable solution of the left-hand side of~(\ref{eq:fabien's proof 1}). Hence
\begin{equation}
     \min_{U, \varrho\ \text{s.t.}\ \lambda_{\varrho}=\lambda} \tr (U \varrho U^{\dagger} H\SB)= \sum_{i=0}^{d-1} \lambda_{i}^{\downarrow} E_{i}\suptiny{0}{0}{B}.
\end{equation}

We have therefore reduced the minimisation problem of~(\ref{equ:minEdecomp}) to solving
\begin{equation}
     \min_{\lambda} \sum_{i=0}^{d-1} \lambda_{i}^{\downarrow} (E_{i}\suptiny{0}{0}{A}+ E_{i}\suptiny{0}{0}{B}),\ \text{s.t.}\
     H(\lambda)=\kappa,\
     \sum_{i} \lambda_{i}^{\downarrow}=1.
\end{equation}
This problem, in turn, can be solved by means of Lagrange multipliers, which yields
\begin{equation}
     \lambda_{i}^{\downarrow}= \frac{e^{-\beta (E_{i}\suptiny{0}{0}{A}+ E_{i}\suptiny{0}{0}{B})}}{\sum_{i=0}^{d-1} e^{-\beta (E_{i}\suptiny{0}{0}{A}+ E_{i}\suptiny{0}{0}{B})}},
\end{equation}
which is precisely what is delivered by the solution $(\varrho_{\mathrm{opt}}(\beta(\kappa)),U_{\mathrm{opt}})$.

We can further check that for every $\kappa \in [0, \ln(d)]$ there exists a unique $\beta \in [0,\infty]$ such that $H(\lambda_{\varrho_{\mathrm{opt}}})=S(\varrho_{\mathrm{opt}}(\beta))=\kappa$, i.e., that the notation $\beta(\kappa)$ is well defined and can be understood as a function. This can be seen from
\begin{align}
     S(\varrho_{\mathrm{opt}}(0)) &=\,\ln(d),\\
     S(\varrho_{\mathrm{opt}}(\infty))  &=\,0,\\
     \frac{d}{d\beta} S(\varrho_{\mathrm{opt}}(\beta))&=\, \beta (\tr(\varrho \tilde{H})^{2}- \tr(\varrho \tilde{H}^2)) < 0 , \;
     \forall \beta \in (0, \infty).
\end{align}
Strictly speaking, the last line is not valid when $\tilde{H} \propto \mathds{1}$, but in that case one can straightforwardly check that $(U_{\mathrm{opt}}, \frac{\mathds{1}}{d})$ solves our original problem (\ref{equ:maxS}) for any allowed $c$. We hence (tacitly) discard it from the start.

Having established this fact about $(\varrho_{\mathrm{opt}}(\beta(\kappa)),U_{\mathrm{opt}})$, let $(\varrho,U)$ be such that $S(\varrho)=\kappa$. Then $\varrho$ has some spectrum $\{\mu_0,\dots, \mu_{d-1}\}$. From the above we thus have
\begin{align}
    &\tr (\varrho H\SA) + \tr (U \varrho U^{\dagger} H\SB) \,\geq\, \sum_{i=0}^{d-1} \mu_i^{\downarrow} E_{i}\suptiny{0}{0}{A}+\sum_{i=0}^{d-1} \mu_i^{\downarrow} E_{i}\suptiny{0}{0}{B}\\
    &\geq\, \sum_{i=0}^{d-1} \frac{e^{-\beta(\kappa) (E_{i}\suptiny{0}{0}{A}+ E_{i}\suptiny{0}{0}{B})}}{\sum_{i=0}^{d-1} e^{-\beta(\kappa)(E_{i}\suptiny{0}{0}{A}+E_{i}\suptiny{0}{0}{B})}} (E_{i}\suptiny{0}{0}{A}+E_{i}\suptiny{0}{0}{B})
    \nonumber\\
    &=\,\tr(\varrho_{\mathrm{opt}}(\beta(\kappa)) H\SA) + \tr ( U_{\mathrm{opt}} \varrho_{\mathrm{opt}}(\beta(\kappa)) U_{\mathrm{opt}}^{\dagger} H\SB),
    \nonumber
\end{align}
which proves our claim, because $S(\varrho_{\mathrm{opt}}(\beta(\kappa)))=\kappa$.
\end{proof}

Now let us define the function
\begin{equation}
    f(\kappa)= \tr(\varrho_{\mathrm{opt}}(\beta(\kappa)) H\SA) + \tr ( U_{\mathrm{opt}} \varrho_{\mathrm{opt}}(\beta(\kappa)) U_{\mathrm{opt}}^{\dagger} H\SB).
\end{equation}
We can then establish the following proposition.

\begin{prop}
\label{claim:strictmon}
If $f(\kappa_1) < f(\kappa_2)$ then $\kappa_1 < \kappa_2$.
\end{prop}

The above is saying that if $f$ has an inverse then that inverse is strictly monotonically increasing. This is indeed how the proof proceeds.

\begin{proof}
First, note that we have
\begin{align}
    \frac{d}{d\kappa} f(\kappa) &=\, \frac{d}{d \beta} f(\beta) \frac{d \beta}{d\kappa}\,=\,
    \frac{d}{d \beta} f(\beta) \left(\frac{d}{d\beta} S(\varrho_{\mathrm{opt}}(\beta)\right)^{-1}.
\end{align}
Further, we have already seen that
\begin{equation}
    \frac{d}{d\beta} S(\varrho_{\mathrm{opt}}(\beta)\,=\,\beta (\tr(\varrho \tilde{H})^{2}- \tr(\varrho \tilde{H}^2)) \,<\, 0.
\end{equation}
Similarly one calculates
\begin{equation}
      \frac{d}{d\beta} f(\beta)\,=\, (\tr(\varrho \tilde{H})^{2}- \tr(\varrho \tilde{H}^2)) \,<\, 0.
\end{equation}
So for $\kappa \in (0, \ln(d))$ we have
\begin{equation}
    \frac{d}{d\kappa} f(\kappa) \,=\, \frac{1}{\beta} \,>\,0.
\end{equation}
This proves that $f$ is strictly monotonically increasing. It therefore has an inverse $f^{-1}$ that is also strictly monotonically increasing which proves our claim.
\end{proof}

We are now ready to solve the maximisation in (\ref{equ:maxS}).
\begin{prop}
There is a (unique) $\beta$ such that $(\varrho_{\mathrm{opt}}(\beta), U_{\mathrm{opt}})$ is a solution of \emph{(\ref{equ:maxS})}.
\end{prop}

\begin{proof}
First, if $c \geq \frac{1}{d} \sum_{i=0}^{d-1} E_{i}\suptiny{0}{0}{A}+E_{i}\suptiny{0}{0}{B}$, then $(\varrho_{\mathrm{opt}}(0),U_{\mathrm{opt}})$ is the sought after solution, because
\begin{equation}
     H((\tfrac{1}{d}, \dots, \tfrac{1}{d}))\,=\, S(\varrho_{\mathrm{opt}}(0)) \,\geq\, S(\varrho)\ \forall\, \varrho\,,
\end{equation}
and $\tr(\varrho_{\mathrm{opt}}(0) \tilde{H}) \leq c$. Second, if $c \leq \frac{1}{d} \sum_{i=0}^{d-1} E_{i}\suptiny{0}{0}{A}+E_{i}\suptiny{0}{0}{B}$ then
\begin{equation}
    \frac{d}{d \beta} f(\beta) \,<\, 0
\end{equation}
ensures that there exists a (unique) $\beta$, say $\bar{\beta}$, such that
\begin{equation}
     f(\bar{\beta}) \,=\, \tr ( \varrho_{\mathrm{opt}} (\bar{\beta}) \tilde{H} )\,=\,c.
\end{equation}
Let us now prove that $(\varrho_{\mathrm{opt}}(\bar{\beta}), U_{\mathrm{opt}})$ is the desired solution in this case. Consider a pair $(\varrho_1,U_1)$ such that
\begin{equation}
     \tr (\varrho_1 H\SA)+ \tr ( U_1 \varrho_1 U_1^{\dagger} H\SB) \,\leq\, c\,
     =\, \tr (\varrho_{\mathrm{opt}} (\bar{\beta}) \tilde{H}).
\end{equation}
Let us further consider $\varrho_{\mathrm{opt}}(\beta(\kappa))$, the solution of (\ref{equ:minEn}) and choose $\kappa=S(\varrho_1)$. Then it holds that
\begin{align}
     \tr( \varrho_{\mathrm{opt}}(\beta(\kappa)) \tilde{H}) &\leq\, \tr (\varrho_1 H\SA)+ \tr ( U_1 \varrho U_1^{\dagger} H\SB) \nonumber\\
     &\leq\, \tr (\varrho_{\mathrm{opt}} (\bar{\beta}) \tilde{H}),
\end{align}
which, using Proposition~(\ref{claim:strictmon}), implies
\begin{equation}
     S(\varrho_{\mathrm{opt}}(\beta(\kappa)) \,\leq\, S(\varrho_{\mathrm{opt}} (\bar{\beta})).
\end{equation}
But $S(\varrho_{\mathrm{opt}}(\beta(\kappa)) = S(\varrho_1)$ such that
\begin{equation}
     S(\varrho_1) \,\leq\, S(\varrho_{\mathrm{opt}} (\bar{\beta})).
\end{equation}
\end{proof}


\subsection{Correlating finitely many copies}\label{passive}

Here we discuss a protocol that can be applied to $n$ copies of the initial state, i.e., $\tau\SA (\beta)^{\otimes n} \otimes \tau\SB (\beta)^{\otimes n}$, to increase the temperature of the marginals in cases where STUs for single copies exist for small temperature differences, but are not attainable for larger temperature differences. For $n$ copies, the protocol consists of $n$ consecutive steps. In each step, a fixed STU achieving some finite temperature increase is applied to different LCSs that correspond to particular pairs of subsystems, such that for $i,j=1,\dots,n$, each subsystem $A_{i}$ interacts with one and only one subsystem $B_{j}$, and no subsystem interacts again with a subsystem it has previously interacted with. This ensures that all marginals are left diagonal and thermal after each step, and leads to a step-wise increase in the temperature of the marginals. In particular, in the $j$th step the unitary is applied to the subsystems $A_{i}$ and $B_{i + j -1 \, \operatorname{mod}(n) }$. This construction guarantees that the tensor product structure of thermal states is preserved for each pair, i.e., $\tau\SA (\beta_j) \otimes \tau\SB (\beta_j)$ where $\beta_j$ denotes the effective local inverse temperature after the $j$th step. Furthermore, this means that in each step the STU can be applied, as the marginals are in a thermal state with the same temperature and moreover, product to each other. While this protocol ensures that the desired structures are preserved, the necessary conditions for the protocol to achieve arbitrary final temperatures are still part of ongoing research. In general the question remains, whether it is possible to reach any arbitrary temperature difference, $\beta'< \beta$, within finitely many steps. This discussion contrasts the proof of the existence of STUs in the asymptotic case, discussed in Sec.~\ref{subsubsec: asymcase}.\\

To illustrate the protocol, let us illustrate the case of $n=4$ copies here. In the following, each dot reprsents a subsystem and the application of STUs is denoted by lines connnecting pairs of subsystems. Initially one is confronted with the following situation.

\begin{center}
\begin{tikzpicture}
\coordinate[label=left:$A_4$] (A1) at (0,1);
\coordinate[label=right:$B_4$] (B1) at (5,1);
\coordinate[label=left:$A_3$] (A2) at (0,2);
\coordinate[label=right:$B_3$] (B2) at (5,2);
\coordinate[label=left:$A_2$] (A3) at (0,3);
\coordinate[label=right:$B_2$] (B3) at (5,3);
\coordinate[label=left:$A_1$] (A4) at (0,4);
\coordinate[label=right:$B_1$] (B4) at (5,4);
\fill (A1) circle (2pt);
\fill (B1) circle (2pt);
\fill (A2) circle (2pt);
\fill (B2) circle (2pt);
\fill (A3) circle (2pt);
\fill (B3) circle (2pt);
\fill (A4) circle (2pt);
\fill (B4) circle (2pt);
\end{tikzpicture}
\end{center}
In the first step, we connect subsystems corresponding to the same copy, i.e., $A_i$ and $B_i$.
\begin{center}
\begin{tikzpicture}
\coordinate[label=left:$A_4$] (A1) at (0,1);
\coordinate[label=right:$B_4$] (B1) at (5,1);
\coordinate[label=left:$A_3$] (A2) at (0,2);
\coordinate[label=right:$B_3$] (B2) at (5,2);
\coordinate[label=left:$A_2$] (A3) at (0,3);
\coordinate[label=right:$B_2$] (B3) at (5,3);
\coordinate[label=left:$A_1$] (A4) at (0,4);
\coordinate[label=right:$B_1$] (B4) at (5,4);
\draw (A1) -- (B1);
\draw (A2) -- (B2);
\draw (A3) -- (B3);
\draw (A4) -- (B4);
\fill (A1) circle (2pt);
\fill (B1) circle (2pt);
\fill (A2) circle (2pt);
\fill (B2) circle (2pt);
\fill (A3) circle (2pt);
\fill (B3) circle (2pt);
\fill (A4) circle (2pt);
\fill (B4) circle (2pt);
\end{tikzpicture}
\end{center}
As specified, we are now in the situation where we have slightly decreased the inverse temperature of the marginals to $\tau(\beta_1)$. We now apply a unitary to the subsystems $A_{i}$ and $B_{i +1 \, \text{mod}(n) }$, further decreasing the inverse temperature of the marginals.
\begin{center}
\begin{tikzpicture}
\coordinate[label=left:$A_4$] (A1) at (0,1);
\coordinate[label=right:$B_4$] (B1) at (5,1);
\coordinate[label=left:$A_3$] (A2) at (0,2);
\coordinate[label=right:$B_3$] (B2) at (5,2);
\coordinate[label=left:$A_2$] (A3) at (0,3);
\coordinate[label=right:$B_2$] (B3) at (5,3);
\coordinate[label=left:$A_1$] (A4) at (0,4);
\coordinate[label=right:$B_1$] (B4) at (5,4);
\draw (A1) -- (B1);
\draw (A2) -- (B2);
\draw (A3) -- (B3);
\draw (A4) -- (B4);
\draw[dashed] (A1) -- (B4);
\draw[dashed] (A2) -- (B1);
\draw[dashed] (A3) -- (B2);
\draw[dashed] (A4) -- (B3);
\fill (A1) circle (2pt);
\fill (B1) circle (2pt);
\fill (A2) circle (2pt);
\fill (B2) circle (2pt);
\fill (A3) circle (2pt);
\fill (B3) circle (2pt);
\fill (A4) circle (2pt);
\fill (B4) circle (2pt);
\end{tikzpicture}
\end{center}
In the third step, the subsystems $A_{i}$ and $B_{i + 2 \, \text{mod}(n) }$ are connected.
\begin{center}
\begin{tikzpicture}
\coordinate[label=left:$A_4$] (A1) at (0,1);
\coordinate[label=right:$B_4$] (B1) at (5,1);
\coordinate[label=left:$A_3$] (A2) at (0,2);
\coordinate[label=right:$B_3$] (B2) at (5,2);
\coordinate[label=left:$A_2$] (A3) at (0,3);
\coordinate[label=right:$B_2$] (B3) at (5,3);
\coordinate[label=left:$A_1$] (A4) at (0,4);
\coordinate[label=right:$B_1$] (B4) at (5,4);
\draw (A1) -- (B1);
\draw (A2) -- (B2);
\draw (A3) -- (B3);
\draw (A4) -- (B4);
\draw[dashed] (A1) -- (B4);
\draw[dashed] (A2) -- (B1);
\draw[dashed] (A3) -- (B2);
\draw[dashed] (A4) -- (B3);
\draw[dotted] (A4) -- (B2);
\draw[dotted] (A3) -- (B1);
\draw[dotted] (A2) -- (B4);
\draw[dotted] (A1) -- (B3);
\fill (A1) circle (2pt);
\fill (B1) circle (2pt);
\fill (A2) circle (2pt);
\fill (B2) circle (2pt);
\fill (A3) circle (2pt);
\fill (B3) circle (2pt);
\fill (A4) circle (2pt);
\fill (B4) circle (2pt);
\end{tikzpicture}
\end{center}
Lastly, we connect $A_{i}$ and $B_{i + 3 \, \text{mod}(n) }$ as prescribed.
\begin{center}
\begin{tikzpicture}
\coordinate[label=left:$A_4$] (A1) at (0,1);
\coordinate[label=right:$B_4$] (B1) at (5,1);
\coordinate[label=left:$A_3$] (A2) at (0,2);
\coordinate[label=right:$B_3$] (B2) at (5,2);
\coordinate[label=left:$A_2$] (A3) at (0,3);
\coordinate[label=right:$B_2$] (B3) at (5,3);
\coordinate[label=left:$A_1$] (A4) at (0,4);
\coordinate[label=right:$B_1$] (B4) at (5,4);
\draw (A1) -- (B1);
\draw (A2) -- (B2);
\draw (A3) -- (B3);
\draw (A4) -- (B4);
\draw[dashed] (A1) -- (B4);
\draw[dashed] (A2) -- (B1);
\draw[dashed] (A3) -- (B2);
\draw[dashed] (A4) -- (B3);
\draw[dotted] (A4) -- (B2);
\draw[dotted] (A3) -- (B1);
\draw[dotted] (A2) -- (B4);
\draw[dotted] (A1) -- (B3);
\draw[dashdotted] (A4) -- (B1);
\draw[dashdotted] (A3) -- (B4);
\draw[dashdotted] (A2) -- (B3);
\draw[dashdotted] (A1) -- (B2);
\fill (A1) circle (2pt);
\fill (B1) circle (2pt);
\fill (A2) circle (2pt);
\fill (B2) circle (2pt);
\fill (A3) circle (2pt);
\fill (B3) circle (2pt);
\fill (A4) circle (2pt);
\fill (B4) circle (2pt);
\end{tikzpicture}
\end{center}
We have now used the four copies of the initial state, to increase the correlations between the two sides step by step, while retaining the product structure of the thermal marginals, arriving at the final reduced states $\rho_{A/B}=\tau(\beta_4)^{\otimes 4}$ for the subsystems $A$ and $B$.


\subsection{Proof of Lemma \ref{lemma:swappingM}}\label{App:Proof:lemma:swappingM}

In this section, we will first show that for any $3\times 3$ doubly stochastic matrix $M$, there exists another $3\times 3$ doubly stochastic matrix $\tilde{M}$ such that $M (\mathds{1}+\Pi)=(\mathds{1}+\Pi) \tilde{M}$, a statement we called Lemma~\ref{lemma:swappingM} in the main text.

\begin{proof}[Proof of Lemma \emph{\ref{lemma:swappingM}}]
In general, all doubly stochastic matrices can be written as convex combinations of permutation matrices as the following
\begin{equation}
     M=\sum_{i=1}^{d!} \alpha_i\, \Pi^{(i)},
     \label{eq:gen. doubly stoch}
\end{equation}
where $\Pi^{(i)}$ indicate permutation matrices in dimension $d$, and $\sum_i \alpha_i =1$. Since we need to show that the statement is true for any doubly stochastic matrix $M$, it is sufficient to prove that the statement is true for $M$ being a permutation matrix. In the following we will just focus on dimension $3$. In this dimension there are $3!=6$ permutation matrices where we collect the cyclic permutations $\Pi\suptiny{0}{0}{(1)}_{\mathrm{C}}=\mathds{1}$, $\Pi\suptiny{0}{0}{(2)}_{\mathrm{C}}=\Pi$, $\Pi\suptiny{0}{0}{(3)}_{\mathrm{C}}=\Pi^{2}$, and anticyclic permutations
\begin{align}
    \Pi\suptiny{0}{0}{(1)}_{\mathrm{AC}} &=\,\begin{pmatrix} 1 & 0 & 0 \\ 0 & 0 & 1 \\ 0 & 1 & 0 \end{pmatrix},\
    \Pi\suptiny{0}{0}{(2)}_{\mathrm{AC}} \,=\,\begin{pmatrix} 0 & 1 & 0 \\ 1 & 0 & 0 \\ 0 & 0 & 1 \end{pmatrix},\
    \Pi\suptiny{0}{0}{(3)}_{\mathrm{AC}} \,=\,\begin{pmatrix} 0 & 0 & 1 \\ 0 & 1 & 0 \\ 1 & 0 & 0 \end{pmatrix},\
\end{align}
in a particular order, i.e.,
\begin{align}
    \Pi\suptiny{0}{0}{(1)} &=\,\Pi\suptiny{0}{0}{(1)}_{\mathrm{C}},\
    \Pi\suptiny{0}{0}{(2)} \,=\,\Pi\suptiny{0}{0}{(1)}_{\mathrm{AC}},\
    \Pi\suptiny{0}{0}{(3)} \,=\,\Pi\suptiny{0}{0}{(3)}_{\mathrm{C}},\nonumber\\
    \Pi\suptiny{0}{0}{(4)} &=\,\Pi\suptiny{0}{0}{(3)}_{\mathrm{AC}},\
    \Pi\suptiny{0}{0}{(5)} \,=\,\Pi\suptiny{0}{0}{(2)}_{\mathrm{C}},\
    \Pi\suptiny{0}{0}{(6)} \,=\,\Pi\suptiny{0}{0}{(2)}_{\mathrm{AC}}.
    \label{permutationdim3dim}
\end{align}
where $\Pi\suptiny{0}{0}{(1)}$, $\Pi\suptiny{0}{0}{(5)}$, and $\Pi\suptiny{0}{0}{(3)}$ trivially commute with $\Pi\suptiny{0}{0}{(5)}=\Pi$. It is further straightforward to show that
\begin{align}
    \Pi\suptiny{0}{0}{(6)} (\mathds{1}+\Pi\suptiny{0}{0}{(5)}) &= (\mathds{1}+\Pi\suptiny{0}{0}{(5)}) \Pi\suptiny{0}{0}{(2)}\nonumber\\
    \Pi\suptiny{0}{0}{(2)} (\mathds{1}+\Pi\suptiny{0}{0}{(5)}) &= (\mathds{1}+\Pi\suptiny{0}{0}{(5)}) \Pi\suptiny{0}{0}{(4)}
    \label{eq:comm. relation 3}\\
    \Pi\suptiny{0}{0}{(4)} (\mathds{1}+\Pi\suptiny{0}{0}{(5)}) &= (\mathds{1}+\Pi\suptiny{0}{0}{(5)}) \Pi\suptiny{0}{0}{(6)}\nonumber.
\end{align}
For any $3\times 3$ doubly stochastic matrix $M$ given in the form of Eq.~(\ref{eq:gen. doubly stoch}), by using Eq.~(\ref{eq:comm. relation 3}), one may thus find a $3\times 3$ doubly stochastic matrix $\tilde{M}$ of the form
\begin{align}
     \tilde{M} &=\alpha_1\, \Pi\suptiny{0}{0}{(1)}+\alpha_3\, \Pi\suptiny{0}{0}{(3)}+\alpha_5\, \Pi\suptiny{0}{0}{(5)}\nonumber\\
     &+\alpha_2\, \Pi\suptiny{0}{0}{(4)}+\alpha_4\, \Pi\suptiny{0}{0}{(6)}+\alpha_6\, \Pi\suptiny{0}{0}{(2)},
\end{align}
which satisfies the required condition.
\end{proof}


\subsection{Generalisation of the majorised marginals approach}\label{sec:generalisations}

One of the approaches to investigat{ing} the existence of {the} unitaries mentioned in {Q}uestion~\ref{question:1} is {to} generalise the `majorised marginals approach' of Sec.~\ref{sec:commutingM}. This approach appears to be promising because of its simplicity and utility. For such a generalisation to higher dimensions to be successful, Eq.~(\ref{eq:EqualEvoMarginal}) would demand to prove the following claim.

\begin{claim}\label{claim:swappingM}
For {any} $d \times d$ doubly stochastic matrix $M$, {with} $d\geq 4$, there exists a doubly stochastic matrix \(\tilde{M}\) such that
\begin{equation}
    M (\mathds{1}+\Pi^i)=(\mathds{1}+\Pi^i) \tilde{M}_i~~ \forall i \in \{ 1, \dots , \lfloor \tfrac{d}{2} \rfloor\}.
    \label{eq:gen comm M}
\end{equation}
\end{claim}

Note that by permuting the chosen basis vectors in a particular manner, $(\mathds{1}+\Pi^i)$ can always be written in the form of $(\mathds{1}+\Pi)$. Additionally, we know that such a transformation of the basis transforms any doubly stochastic matrix to another doubly stochastic matrix. Since Eq.~(\ref{eq:gen comm M}) should be true for any doubly stochastic matrix $M$, these two facts imply that Eq.~(\ref{eq:gen comm M}) can be reduced to showing that for any doubly stochastic $d\times d$ matrix $M$ there exists a doubly stochastic matrix $\tilde{M}$ such that
\begin{equation}
    M (\mathds{1}+\Pi)\,=\,(\mathds{1}+\Pi) \tilde{M}.
    \label{eq:swapping relation}
\end{equation}
Using a counterexample, we show that Claim~\ref{claim:swappingM} does not hold in general for dimensions $d\geq 4$. In particular, we construct a counterexample in dimension $4$.\\


\noindent\textbf{Counterexample to Claim}~\ref{claim:swappingM}.\
Let
\begin{equation}
M=\begin{pmatrix}
1&0&0&0\\
0&0&0&1\\
0&1&0&0\\
0&0&1&0.
\end{pmatrix}{.}
\label{m-temp}
\end{equation}
As we will now show, there exists no doubly stochastic $\tilde{M}$ such that $M (\mathds{1} + \Pi)= (\mathds{1}+\Pi) \tilde{M}$. To do so, we first calculate
\begin{equation}
M(\mathds{1}+\Pi)=\begin{pmatrix}1&0&0&1\\
0&0&1&1\\
1&1&0&0\\
0&1&1&0 \end{pmatrix}.
\end{equation}
We are then interested in determining whether this can be equal to $(\mathds{1}+\Pi) \tilde{M}$ for some $\tilde{M}$, with elements $m_{ij} \geq 0$, for $i,j =0,\dots,3$ such that each column and row sums to $1$. From the first row of $(\mathds{1}+\Pi)\tilde{M}$, we obtain $m_{{00}}+m_{{30}}=1$, which implies $m_{{10}}=m_{{20}}=0$, $m_{{01}}=m_{{31}}=0$, and $m_{{02}}=m_{{32}}=0$, and further we obtain $m_{{03}}+m_{{33}}=1$, which implies $m_{{13}}=m_{{23}}=0$. Thus, $\tilde{M}$ must be of the form
\begin{equation}
\tilde{M}=\begin{pmatrix}
m_{00}&0&0&m_{03}\\
0&m_{11}&m_{12}&0\\
0&m_{21}&m_{22}&0\\
m_{30}&0&0&m_{33}\\
\end{pmatrix}.
\end{equation}
Then, from the second row of $(\mathds{1}+\Pi)\tilde{M}$ we obtain $m_{{00}}=0$, implying $m_{{30}}=1$, as well as $m_{{11}}=0$, which implies $m_{{21}}=1$. Moreover, we have $m_{{12}}=1$, which means $m_{{22}}=0$, while $m_{{03}}=1$ suggests $m_{{33}}=0$. We hence have
\begin{equation}
\tilde{M}=\begin{pmatrix}
0&0&0&1\\
0&0&1&0\\
0&1&0&0\\
1&0&0&0\\
\end{pmatrix}.
\end{equation}
But since $((\mathds{1}+\Pi) \tilde{M})_{31}=0 \neq (M ( \mathds{1}+\Pi))_{31}=0$ we arrive at a contradiction.


Since Lemma~\ref{lemma:swappingM} {does not hold in general}, one can try to relax it to a weaker statement which is still suitable for our purposes. One way to do so is to demand that an equivalent of Eq.~(\ref{eq:swapping relation}) holds only when applied to (certain) vectors, in the spirit of the observation that although the set of doubly stochastic matrices does not coincide with that of unistochastic ones, given a vector $\vec{v}$ with nonnegative components and a doubly stochastic matrix $M$, there is yet a unistochastic $M_U$ such that $M  \vec{v} = M_U  \vec{v}$. Thus, we investigate whether the following statement is true or not:

\begin{claim}\label{claim:swappingMv}
Given a vector $\vec{v}$ of nonnegative components and a doubly stochastic matrix $M$, there exists a doubly stochastic matrix $\tilde{M}_{v}$, which may depend on $\vec{v}$, such that
\begin{equation}
    M (\mathds{1}+\Pi) \vec{v}=(\mathds{1}+\Pi) \tilde{M}_{v}  \vec{v}.
\end{equation}
\end{claim}

The relaxation being clearly that now $\tilde{M}$ is allowed to depend on $\vec{v}$. Unfortunately, the previous counterexample carries over to Claim~\ref{claim:swappingMv}, as we shall see.\\

\noindent\textbf{Counterexample to Claim}~\ref{claim:swappingMv}.\
Let us consider $M$ as in Eq.~(\ref{m-temp}) and let \( \vec{v}=(1,0,0,0)^{{T}}\). Then we have
\begin{equation}\label{eq:Monv}
    M (\mathds{1}+\Pi)  \vec{v} = (1,\,0,\,1,\,0)^{T}.
\end{equation}
We then want to know if this equals $(\mathds{1}+\Pi) \tilde{M}_{v}  \vec{v}$ for some $\tilde{M}_{v}$ with elements $m_{ij} \geq 0 , \, i, j ={0},\dots,{3}$ such that each column and row sums to $1$. We obtain
\begin{equation}
    (\mathds{1}+\Pi) \tilde{M}_{v}  \vec{v} =
    \begin{pmatrix}
        m_{{00}}+m_{{30}}\\
        m_{{00}}+m_{{10}}\\
        m_{{10}}+m_{{20}}\\
        m_{{20}}+m_{{30}}
    \end{pmatrix}.
\end{equation}
Comparing this with Eq.~(\ref{eq:Monv}), we have $m_{00}+m_{30}=m_{10}+m_{20}$ and $m_{20} + m_{30} = m_{00} + m_{10}$, which yields $m_{30}=m_{10}$, and concurrently $m_{00} + m_{30}=0$ and $m_{00} + m_{30}=1$, which means we arrive at a contradiction.


\subsection{Counterexample for majorisation relations}\label{app:claim:qmajri}

Here, via counterexample, we show that the following claim (discussed in Sec.~\ref{sec:3_passing}) does not hold in general:

\begin{claim}\label{claim:qmajr}
The relation $\frac{\vec{q}}{\lVert \vec{q} \rVert} \succ \frac{\vec{b}_i}{\lVert \vec{b}_i \rVert}$ holds $\forall$ $d\geq2$ and $\beta' \leq \beta$.
\end{claim}

\noindent\textbf{Counterexample to Claim}~\ref{claim:qmajr}.\
Consider the case where the last eigenvalue of the local Hamiltonian is infinitely large, $E_{d-1}\to \infty$. In this case, we know that for a thermal state with any finite temperature, the corresponding probability weight vanishes, $p_{d-1}(\beta)\to 0$. Due to the definitions of the vectors $\vec{q}$ and $\vec{r}_i$ [Eq.~(\ref{deco. marginal})], it is clear that $p_{d-1}(\beta)$ contributes to one element of the vector $\vec{q}$, i.e., to $({\vec{q}})_{d-1}=p_{d-1}^2$, and two elements of the vector $\vec{r}_i$, $(\vec{r}_{{i}})_{d-1-i}=p_{d-1-i}\,p_{d-1}$ and $(\vec{r}_{{i}})_{d-1}=p_{d-1}p_{d-1+i}$. Hence, the vectors $\vec{q}$ and $\vec{b}_i=\vec{r_i}(\beta^\prime)$ have $d-1$ and $d-2$ nonzero elements, respectively. Now recalling from majorisation theory that no vector can be majorised by a vector with higher rank one can see that Claim~\ref{claim:qmajr} does not hold.


\subsection{Proof of Lemma \emph{\ref{lemma:GenPassNormD}}}
\label{app:PassNorm}

Here, we present the proof of Lemma~\ref{lemma:GenPassNormD} for any $d$-dimensional system, which ensures that condition~(\ref{passing con i}) holds in general.

\begin{proof}[Proof of Eq.~\emph{(\ref{lemma:rismallerbi})} in Lemma~\emph{\ref{lemma:GenPassNormD}}.]
To prove that Eq.~(\ref{lemma:rismallerbi}) holds, we need to show that for any $\beta> \beta^{\prime}$, $g(\beta ):= \lVert \vec{q} \rVert $ is greater than $g(\beta^{\prime} )= \lVert \vec{ a} \rVert$ in $d$-dimensional systems. To achieve this, one may use the positivity of the first derivative of the function $g(\beta)$. We therefore calculate
\begin{align}
    \partial_\beta \lVert \vec{q} \rVert &=\partial_\beta ( \frac{\sum_i e^{-2\beta E_i}}{\sum_{m,n} e^{-\beta (E_m+E_n)}})\\
    &=\,\tfrac{1}{z(\beta)^4}\sum_{i,m,n}(E_m+E_n-2E_i)e^{-\beta (2E_i+E_m+E_n)}\nonumber\\
    &=\,\tfrac{2}{z(\beta)^3}\sum_{i,m}(E_m-E_i)e^{-\beta (2E_i+E_m)}\nonumber\\
    &=\,\tfrac{2}{z(\beta)^3}\sum_{m\,>\, i}\big[(E_m-E_i)e^{-\beta (2E_i+E_m)}\nonumber\\
    &\ \ \ \ \ +\,(E_i-E_m)e^{-\beta (E_i+2E_m)}\big]\nonumber\\
    &\,=\tfrac{2}{z(\beta)^3}\sum_{m\,>\, i}(E_m-E_i)(e^{-\beta (2E_i+E_m)}-e^{-\beta (E_i+2E_m)}).\nonumber
\end{align}
Without loss of generality, we have ordered the energy eigenvalues in increasing order, $E_{i+1}\geq E_i$. Then, for any pair $(m,i)$ with $m> i$, we know that $(e^{-\beta (2E_i+E_m)}-e^{-\beta (E_i+2E_m)})$ is nonnegative, implying $\partial_\beta \lVert \vec{q} \rVert \geq 0$. Also note that the last inequality is strict unless $E_i=E_m$ for all $m>i$ implying $H \propto \mathds{1}$ for which the problem is already solved. For all practical purposes, the inequality is therefore strict.
\end{proof}

\begin{proof}[Proof of Eq.~\emph{(\ref{lemma:rimajbi})} in Lemma~\emph{\ref{lemma:GenPassNormD}}.]
Using our convention of LCSs, see Sec.~\ref{sec:LCSs}, we have
\begin{equation}
    \vec{r}_i:= \sum_{j=0}^{d-1} p_{j\, j+i} \,\vec{e}_{j}\ \ \
    \text{and}\ \ \
    \vec{b}_i:= \sum_{j=0}^{d-1} p_{j\, j+i}^\prime \,\vec{e}_{j},
\end{equation}
where $p_{j\, j+i}(\beta):=e^{-\beta (E_j+E_{j+i})}/\mathcal{Z}(\beta)^2$ and $p^\prime_{j\, j+i}:=p_{j\, j+i}(\beta\pr)$. We then calculate
\begin{align}
    \frac{\vec{r}_i}{\lVert \vec{r}_i \rVert} &= \sum_{j=0}^{d-1} \frac{p_{j,j+i}}{\sum_{k=0}^{d-1} p_{k\,k+i}} \,\vec{e}_{j}
    \,=\, \sum_{j=0}^{d-1} \frac{e^{-\beta (E_j+E_{j+i})}}{\sum_{k=0}^{d-1} e^{-\beta (E_k+E_{k+i})}} \,\vec{e}_{j}
    \nonumber\\
    &= \sum_{j=0}^{d-1} \frac{e^{-\beta A_j}}{\sum_{k=0}^{d-1} e^{-\beta A_k}} \,\vec{e}_{j},
\end{align}
where $ A_j:= E_j+E_{j+i}$ for all $j=0,\dots, d-1$. The last expression is nothing else than the vectorised form of a thermal state at inverse temperature $\beta$ with respect to the Hamiltonian $\sum_{j=0}^{d-1} A_j \ket{j}\!\!\bra{j}$. Since the vector of diagonal entries of any thermal state majorises the vectorised diagonal of any thermal state (with respect to the same Hamiltonian) with lower inverse temperature $\beta'< \beta$, which is nothing else than $\vec{b}_i/\lVert \vec{b}_i \rVert$, Eq.~(\ref{lemma:rimajbi}) holds, which concludes the proof.
\end{proof}


\subsection{The existence of STUs with the stronger version of condition~(\ref{passing con ii})}\label{app:strong condition II}

In this section, we will show how one can prove the existence of the STUs if condition~(\ref{passing con i}) and the stronger version of condition~(\ref{passing con ii}) from Sec.~\ref{sec:3_passing} hold.

condition~(\ref{passing con i}) ensures that there exist doubly stochastic matrices $M_{r_i}$ which transform $\vec{r}_i/ \lVert \vec{r}_i\rVert$ to $\vec{b}_i/ \lVert \vec{b}_i\rVert$ and also $\vec{q}/ \lVert \vec{q}\rVert$ to $\vec{a}/ \lVert \vec{a} \rVert$, i.e.,
\begin{equation}
    M_{r_i} \frac{\vec{r}_i}{\lVert \vec{r}_i\rVert}=\frac{\vec{b}_i}{\lVert \vec{b}_i\rVert} ~~~\mathrm{for} ~~ i=0,\dots,k.
    \label{eq:rimajbiddd}
\end{equation}
Using Eqs.~(\ref{eq:EqualEvoMarginal C}) and~(\ref{eq:rimajbiddd}), the marginals then become
\begin{align}
    \vec{\tilde{p}}
    &= M_q\vec{q} + \sum_{i=1}^{k} (\lfloor\tfrac{2i}{d}\rfloor+1)^{-1} (\mathds{1}+ \Pi^i) \frac{\rVert     \vec{r}_i\lVert}{\rVert \vec{b}_i\lVert}\vec{b}_i .
    \label{eq:EqualEvoMarginal Cprime11}
\end{align}
To reach a thermal state with higher temperature, one can compensate the required norm of $(\openone+\Pi^i)\,\vec{b}_i$ by extra norm from the vector $\vec{q}$. To achieve that, we need to have  $\lVert \vec{q}\rVert \geq \lVert \vec{ a} \rVert$, which is ensured by condition~(\ref{passing con i}). In addition, another requirement for shifting the norm is the existence of doubly stochastic matrices $M_{q\to2 {b}_i}$ to transform $\vec{q}/ \lVert \vec{q}\rVert$ to $(\openone+\Pi^i)\,\vec{b}_i/2\lVert \vec{b}_i\rVert$. Due to the HLP theorem, condition~(\ref{passing con ii}) ensures that such matrices exist. Therefore, $M_q$ can be written as
\begin{equation}
    M_q= \alpha_0 \,M_{q\to a}+
    \sum_{i=1}^{k}  \alpha_i\, M_{q\to 2 {b}_i} ,
    \label{eq: cov. com. M0}
\end{equation}
where $M_{q\to a}$ and the $M_{q\to {b}_i}$ are doubly stochastic matrices that map $\vec q$ to $\vec a$ and $\vec q$ to each of the $(\openone+\Pi^i)\,\vec{b}_i/2\lVert \vec{b}_i\rVert$, respectively. Applying $M_q$ to the vector $\vec{q}$ in Eq.~(\ref{eq:EqualEvoMarginal Cprime11}), we have
\begin{align}
    \vec{\tilde{p}}\,&=\alpha_0\,\frac{\lVert \vec{q} \rVert}{\lVert \vec{a} \rVert} \,\vec{a}+ \sum_{i=1}^{k} (\alpha_i\,\frac{\lVert \vec{q} \rVert}{2\lVert \vec{b}_i \rVert}+(\lfloor\tfrac{2i}{d}\rfloor+1)^{-1}\frac{\rVert \vec{r}_i\lVert}{\rVert \vec{b}_i\lVert}) (\openone+\Pi) \vec{b}_i.
    \label{eq:trans. 4d pass 2 }
\end{align}
In order to obtain a thermal state at inverse temperature $\beta\pr\leq \beta$, one must choose the coefficients $\alpha_i$ as
\begin{equation}
\label{eq:alphai}
    \alpha_0=\frac{\rVert \vec{a}\lVert}{\rVert \vec{q}\lVert},~~\alpha_i=2(\lfloor\tfrac{2i}{d}\rfloor+1)^{-1}\frac{(\rVert \vec{b}_i\lVert-\rVert \vec{r}_i\lVert)}{\rVert \vec{q}\lVert}.
\end{equation}
Using the fact that convex combinations of doubly stochastic matrices are doubly stochastic matrices, to ensure that $M_q$ is indeed a doubly stochastic matrix, one needs to show that the coefficients $\alpha_i$ as given by Eq.~(\ref{eq:alphai}) are positive and sum to one. The inequality $\rVert \vec{b}_i\lVert \geq \rVert \vec{r}_i\lVert$ in condition~(\ref{passing con ii}) ensures their positivity. Furthermore by using that $\lVert \vec{p} \rVert = \lVert \vec{\tilde{p}} \rVert$ we have that
\begin{equation}
    \lVert \vec{p}\rVert =\lVert \vec{q}\rVert + 2\sum_{i=1}^{k} (\lfloor\tfrac{2i}{d}\rfloor+1)^{-1}  {\rVert \vec{r}_i\lVert}=\lVert \vec{a}\rVert + 2\sum_{i=1}^{k} (\lfloor\tfrac{2i}{d}\rfloor+1)^{-1}  {\rVert \vec{b}_i\lVert},
\end{equation}
which shows that $\alpha_0 +\sum_{i=1}^{k} \alpha_i=1$, concluding the proof. \qed


\subsection{Proofs of Lemma \ref{lemma:PassNorm3d} and Lemma \ref{lemma:PassNorm4d}}\label{App:proof:passnorminD}

Here we present detailed proofs of Lemma~\ref{lemma:PassNorm3d} and Lemma \ref{lemma:PassNorm4d}.

\begin{proof} [Proof of Lemma~\emph{\ref{lemma:PassNorm3d}}]
To prove $\frac{\vec{q}}{\lVert \vec{q} \rVert} \succ \frac{(\mathds{1} + \Pi) b}{2 \lVert \vec{ b} \rVert}$, we can employ the following two majorisation relations:
\begin{align}
    \frac{\vec{q}}{\lVert \vec{q} \rVert} &\succ \frac{(\mathds{1} + \Pi) \vec{r}}{2 \lVert \vec{ r} \rVert},
    \label{eq:qmajrapp}\\
    \frac{(\mathds{1} + \Pi) \vec{r}}{2 \lVert \vec{ r} \rVert} &\succ \frac{(\mathds{1}+\Pi) \vec{b}}{2 \lVert \vec{ b} \rVert}.
    \label{eq:2rmaj2bapp}
\end{align}
If we show that these relations are true for any $\beta' \leq \beta$, we can say that our statement is proven.

To prove the relation in~(\ref{eq:qmajrapp}), we need to show that the greatest/smallest entry of $\vec{q}/\lVert \vec{q} \rVert$ is greater/smaller than the greatest/smallest entry of $\frac{\mathds{1}+\Pi}{2} \frac{\vec{r}}{\lVert \vec{ r} \rVert}$. That is, we first need to check that
\begin{align}
 \frac{p_{00}}{p_{00}+p_{11}+p_{22}} &\geq \frac{p_{01}+p_{20}}{2(p_{01}+p_{12}+p_{20})}.
\end{align}
Disregarding the trivial case where $p_{0}=1$, and $p_{1}=p_{2}=0$, this inequality can be transformed to
\begin{align}
   &e^{-\beta E_1} + 2 e^{-\beta E_1} e^{- \beta E_2} + e^{-\beta E_2}\nonumber\\
   &\ \ \geq ( e^{- \beta E_1 }+ e^{-\beta E_2})
   (e^{- 2 \beta E_1}+ e^{-2 \beta E_2})
   \label{eq:mystery inequality}
\end{align}
This inequality indeed holds, since the left-hand side is larger than (or equal to)  $(e^{-\beta E_1}+e^{-\beta E_2})^2$, which, in turn, is larger or equal to the righ-hand side of~(\ref{eq:mystery inequality}).

Then, second, we need to check the inequality
\begin{align}
 \frac{p_{22}}{p_{00}+p_{11}+p_{22}}    &{\leq}\frac{p_{20}+p_{12}}{2 (p_{01}+p_{12}+p_{20})},
\end{align}
where the relevant case (i.e., for $p_{1}\neq1$ such that at least $p_{1}\neq0$) can be rewritten as
\begin{align}
    & e^{-2 \beta E_2} (e^{-\beta (E_1+E_2)} + 2 e^{-\beta E_1} + e^{-\beta E_2})\nonumber\\
    & \ \ \leq
    (e^{-\beta E_2}+ e^{-\beta (E_1+E_2)}) (1+e^{- 2 \beta E_1}).
    \label{eq:mystery inequality 2}
\end{align}
This second inequality holds since the right-hand side of~(\ref{eq:mystery inequality 2}) is larger than (or equal to) $(e^{-\beta E_2}+e^{-\beta (E_1+E_2)})^2$, which, in turn, is larger or equal to the righ-hand side of~(\ref{eq:mystery inequality 2}).

To prove Eq.~(\ref{eq:2rmaj2bapp}), we show that the greatest entry of $\frac{(\mathds{1}+\Pi)\vec{r}}{2 \lVert \vec{ r} \rVert}$ is monotonically increasing with $\beta$, and that its smallest entry is monotonically decreasing with $\beta$. For the former we calculate
\begin{align}
    \partial_{\beta} &\frac{p_{01}+p_{20}}{2(p_{01}+p_{12}+p_{20})}= \partial_{\beta} \frac{e^{-\beta E_1}+e^{-\beta E_2}}{2( e^{- \beta E_1}+ e^{-\beta (E_1+E_2)}+e^{-\beta E_2})}\nonumber\\
    &=\, \frac{E_2 e^{-\beta (2E_1+E_2)}+E_1 e^{-\beta (E_1+2 E_2)}}{2(e^{-\beta E_1}+e^{-\beta (E_1+E_2)}+e^{-\beta E_2})^2}\geq 0.
\end{align}
And for the latter
\begin{align}
    \partial_{\beta} &\frac{p_{20}+p_{21}}{2(p_{01}+p_{12}+p_{20})}= \partial_{\beta} \frac{e^{-\beta E_2}+e^{-\beta (E_1+E_2)}}{2( e^{- \beta E_1}+ e^{-\beta (E_1+E_2)}+e^{-\beta E_2})}\nonumber\\
    &=\, \frac{(E_1-E_2) e^{-\beta (E_1+E_2)}-E_2 e^{-\beta (2E_1+ E_2)}}{2(e^{-\beta E_1}+e^{-\beta (E_1+E_2)}+e^{-\beta E_2})^2}\leq 0.
\end{align}
\end{proof}

\begin{proof}[Proof of Eq.~\emph{(\ref{lemma:qmaj2ri4d})} in Lemma~\emph{\ref{lemma:PassNorm4d}}.]
We prove {this} lemma under {the} stated condition on {the} gap structure of the local Hamiltonians, i.e., $\delta_{i}{\geq} \delta_{i+1}$. Here we {need to show} that
\begin{align}
    \frac{\vec{q}}{\lVert \vec{q} \rVert} &\succ \frac{\mathds{1}+\Pi}{2} \frac{\vec{b}_1}{\lVert \vec{b}_1 \rVert}{,}\label{qmajb14d}\\
    \frac{\vec{q}}{\lVert \vec{q} \rVert} &\succ \frac{\mathds{1}+\Pi^2}{2} \frac{\vec{b}_2}{\lVert \vec{b}_2 \rVert}= \frac{\vec{b}_2}{\lVert \vec{b}_2 \rVert}\label{qmajb24d}.
\end{align}
To {do so}, we {argue} that the majorisation relations
$ \frac{\vec{q}}{\lVert \vec{q} \rVert} \succ \frac{\mathds{1}+\Pi^i}{2} \frac{\vec{r}_i}{\lVert \vec{r}_i \rVert}$ and $\frac{\mathds{1}+\Pi^i}{2} \frac{\vec{r}_i}{\lVert \vec{r}_i \rVert}\succ \frac{\mathds{1}+\Pi^i}{2} \frac{\vec{b}_i}{\lVert \vec{b}_i \rVert}${,} for $i=1,2{,}$ {hold}.

{Let us} start with Eq.~(\ref{qmajb14d}). The condition {\(\delta_{i} \geq \delta_{i+1}\)} sets the ordering of \(\frac{\mathds{1}+\Pi}{2} \frac{{\vec{r}_{1}}}{\lVert {\vec{r}_{1}} \rVert}\). Still, some ambiguity remains as
\begin{equation}
p_{01}+p_{30} \geq p_{12}+p_{01} \geq p_{30}+p_{23} \geq p_{23}+p_{12}.
\end{equation}
Thus, we need to prove the following three inequalities:
\begin{subequations}
\begin{align}
\frac{p_{00}}{p_{00}+p_{11}+p_{22}+p_{33}} &\geq \frac{p_{01}+p_{30}}{2 (p_{01}+p_{12}+p_{23}+p_{30})}{,}\label{item a4}\\[1mm]
\frac{p_{00}+p_{11}}{p_{00}+p_{11}+p_{22}+p_{33}} &\geq \frac{p_{01}+p_{30}+p_{12}+p_{01}}{2 (p_{01}+p_{12}+p_{23}+p_{30})}{,}\label{item b4}\\[1mm]
\frac{p_{33}}{p_{00}+p_{11}+p_{22}+p_{33}} &\leq \frac{p_{32}+p_{12}}{2 (p_{01}+p_{12}+p_{23}+p_{30})}\label{item c4}.
\end{align}
\end{subequations}
Inequality~(\ref{item a4}) can be rewritten as
\begin{align}
    \frac{1}{1+e^{-2 \beta E_1}+e^{-2 \beta E_2}+e^{-2 \beta E_3}} &\geq \frac{1}{2 \left( 1+ \frac{p_{12} + p_{23}}{p_{01}+p_{30}}\right)},
\end{align}
which can further be turned into the inequality
\begin{align}
    1+1+e^{-\beta E_2}+e^{-\beta E_2}  &\geq 1+e^{-2 \beta E_1}+e^{-2 \beta E_2}+e^{-2 \beta E_3},
\end{align}
which {holds because} $E_1 \geq 0$ and $E_3 \geq E_2$.

{Similarly}, inequality~(\ref{item b4}) can be {recast} as
\begin{align}
    \frac{1}{1+\frac{p_{22}+p_{33}}{p_{00}+p_{11}}} &\geq \frac{1}{2} \left( \frac{1}{1+\frac{p_2}{p_0}}+\frac{1}{1+\frac{p_3}{p_1}} \right)
\end{align}
which implies the inequality
\begin{align}
    f\left( \frac{p_2^2+p_3^2}{p_0^2+p_1^2}\right) &\geq \frac{1}{2} \left[f\left( \frac{p_2}{p_0}\right)+f\left( \frac{p_3}{p_1}\right)\right],
\end{align}
where {we have introduced} \(f(x) := 1/(1+x)\). We {can see that} \(f'(x) \leq 0\) {and} \(\frac{p_2^2+p_3^2}{p_0^2+p_1^2} \leq \frac{p_2}{p_0}\), and in the case of \(E_3 \leq E_1 +E_2\) also \(\frac{p_2}{p_0} \leq \frac{p_3}{p_1}\){. Thus,}
\begin{equation}
f\left( \frac{p_3}{p_1} \right) \leq f \left( \frac{p_2}{p_0} \right) \leq f \left( \frac{p_2^2 +p_3^2}{p_0^2+p_1^2} \right),
\end{equation}
{from which}
\begin{equation}
\frac{1}{2} \left[f\left( \frac{p_3}{p_1} \right) + f \left( \frac{p_2}{p_0} \right)\right] \leq f \left( \frac{p_2^2 +p_3^2}{p_0^2+p_1^2} \right)
\end{equation}
follows. {This} proves inequality~(\ref{item b4}).

Inequality~(\ref{item c4}) is equivalent to
\begin{align}
    \frac{1}{e^{2 \beta E_3}+e^{2 \beta (E_3-E_1)}+ e^{2 \beta (E_3-E_2)}+1} & \leq \frac{1}{2 \left( 1 + \frac{p_{01}+p_{30}}{p_{32}+p_{12}} \right)},
\end{align}
which can be rewritten as
\begin{align}
    1 + 2\frac{p_0}{p_2} &\leq e^{2 \beta E_3}+e^{2 \beta (E_3-E_1)}+ e^{2 \beta (E_3-E_2)},
\end{align}
which is true since the right-hand side is larger or equal $e^{2 \beta E_2}+2$, which, in turn, is larger than (or equal to) the left-hand side, because $(1-e^{\beta E_2})^2 \geq 0$. This proves Inequality~(\ref{item c4}).

{Thus far, we have shown} that $ \frac{\vec{q}}{\lVert \vec{q} \rVert} \succ \frac{\mathds{1}+\Pi}{2} \frac{\vec{r}_1}{\lVert \vec{r}_1 \rVert}$. To complete the first part {of the} proof, we {also} need to show that
\begin{equation}
\frac{\mathds{1}+\Pi}{2} \frac{\vec{r}_1}{\lVert \vec{r}_1 \rVert} \succ \frac{\mathds{1}+\Pi}{2} \frac{\vec{b}_1}{\lVert \vec{b}_1 \rVert}.
\label{eq-tmp-a}
\end{equation}
To do so, we {first} prove that the derivative{s} with respect to \(\beta\) of the {two} largest elements {of \(\frac{\mathds{1}+\Pi}{2} \frac{\vec{r}_1}{\lVert \vec{r}_1 \rVert}\) are} positive and those of the two smallest elements {are} negative. {Let us represent} the elements of \(\frac{\mathds{1}+\Pi}{2} \frac{\vec{r}_1}{\lVert \vec{r}_1 \rVert}\) as
\begin{equation}
    \frac{\mathds{1}+\Pi}{2} \frac{\vec{r}_1}{\lVert \vec{r}_1 \rVert}=
    \bigl(f_{{0}}(\beta), f_{{1}}(\beta), f_{{2}}(\beta), f_{{3}}(\beta) \bigr)^{T}.
\end{equation}
{Note that}
\begin{equation}
\begin{aligned}
f_{{0}}(\beta)&= \frac{p_{01} + p_{30}}{2(p_{01}+p_{12}+p_{23}+p_{30})}=\frac{1}{2} \frac{p_0(p_1+p_3)}{(p_0+p_2)(p_1+p_3)}\\
&= \frac{1}{2} \frac{1}{1+e^{-\beta E_2}}{,}
\end{aligned}
\end{equation}
{which yields}
\begin{equation}\label{eq:f1deri}
f'_{{0}}(\beta)= \frac{1}{2} \frac{E_2\,e^{-\beta E_2}}{(1+e^{-\beta E_2})^2} \geq 0.
\end{equation}
Similarly{, we obtain}
\begin{equation}
\begin{aligned}\label{eq:fderi}
f_{{1}}(\beta)&= \frac{1}{2} \frac{1}{1+e^{-\beta (E_3-E_1)}}  , \,f'_{{1}}(\beta)=  \frac{(E_3-E_1)e^{-\beta (E_3-E_1)}}{2(1+e^{-\beta (E_3-E_1)})^2} \geq 0{,}\\
f_{{2}}(\beta)&= \frac{1}{2} \frac{1}{1+e^{\beta E_2}}  ,\,f'_{{2}}(\beta)= \frac{-E_2\,e^{\beta E_2}}{2(1+e^{\beta E_2})^2} \leq 0{,}\\
f_{{3}}(\beta)&= \frac{1}{2} \frac{1}{1+e^{\beta (E_3-E_1)}}  ,\,f'_{{3}}(\beta)=  \frac{-(E_3-E_1)e^{\beta (E_3-E_1)}}{2(1+e^{\beta (E_3-E_1)})^2} \leq 0.
\end{aligned}
\end{equation}

In the regime \(E_3 \leq E_1+E_2\), the components {are ordered as}
\begin{equation}
f_{{0}}(\beta) \geq f_{{4}}(\beta) \geq f_{{3}}(\beta) \geq f_{{2}} (\beta),
\end{equation}
and {thus} one needs to {show} for {\(\beta' \leq \beta\)} that
\begin{equation}
\begin{aligned}
f_{{0}}(\beta) &\geq f_{{0}}(\beta'){,}\\
f_{{0}}(\beta)+f_{{1}}(\beta) &\geq f_{{0}}(\beta')+f_{{1}}(\beta'){,}\\
f_{{2}}(\beta) &\leq f_{{2}}(\beta'){,}\\
f_{{2}}(\beta)+f_{{3}}(\beta) &\leq f_{{2}}(\beta')+f_{{3}}(\beta').
\end{aligned}
\end{equation}
These relations follow straightforwardly from Eqs.~(\ref{eq:f1deri}) and~(\ref{eq:fderi}), noting that a function $f_{i}(x)$ that fulfils $f'_{i}(x)\leq 0$ ($f'_{i}(x)\geq 0$) the relation $f_{i}(x_1)\leq f_{i}(x_{2})$ ($f_{i}(x_1)\geq f_{i}(x_{2})$) is also satisfied for $x_{1}\geq x_{2}$.

In {a similar fashion, in} the regime \(E_3 \geq E_1+E_2\) the ordering of the components is
\begin{equation}
f_{{1}}(\beta) \geq f_{{0}}(\beta) \geq f_{{2}}(\beta) \geq f_{{3}} (\beta),
\end{equation}
and {thus} one needs to {show} for {\(\beta' \leq \beta\)} that
\begin{equation}
\begin{aligned}
f_{{1}}(\beta) &\geq f_{{1}}(\beta'){,}\\
f_{{0}}(\beta)+f_{{1}}(\beta) &\geq f_{{0}}(\beta')+f_{{1}}(\beta'){,}\\
f_{{3}}(\beta) &\leq f_{{3}}(\beta'){,}\\
f_{{2}}(\beta)+f_{{3}}(\beta) &\leq f_{{2}}(\beta')+f_{{3}}(\beta').
\end{aligned}
\end{equation}
{These relations} also follow straightforwardly from Eqs.~(\ref{eq:f1deri}) and (\ref{eq:fderi}).

We next turn our attention to Eq.~(\ref{qmajb24d}). {Because} \(p_{02} \geq p_{13}\), we {need} to prove the following inequalities{:}
\begin{subequations}
\begin{align}
\frac{p_{00}}{p_{00}+p_{11}+p_{22}+p_{33}} &\geq \frac{p_{02}}{2 (p_{02}+p_{13})}{,}\label{item aa4}\\
\frac{p_{00}+p_{11}}{p_{00}+p_{11}+p_{22}+p_{33}} &\geq \frac{2p_{02}}{2 (p_{02}+p_{13})}{,}\label{item bb4}\\
\frac{p_{33}}{p_{00}+p_{11}+p_{22}+p_{33}} &\leq \frac{p_{13}}{2 (p_{02}+p_{13})}\label{item cc4}.
\end{align}
\end{subequations}
Inequality~(\ref{item aa4}) can be simplified to
\begin{align}
\frac{p_{0}}{p_{00}+p_{11}+p_{22}+p_{33}} &\geq \frac{p_{2}}{2 (p_{02}+p_{13})},
\end{align}
where we can use the inequalities $p_{0}p_{02}\geq p_{2}p_{11} $ and $p_{0}p_{13}\geq p_2 p_{33}$ to arrive at the condition $p_{0}p_{13}\geq  p_{2}p_{22}$.
Due to our assumption on the energy eigenvalues, we {have} $p_{03}\geq p_{22}$, which proves inequality~(\ref{item aa4}).

Inequality~(\ref{item bb4}) can be rewritten as
\begin{align}
    \frac{p_{00}+p_{11}}{p_{00}+p_{11}+p_{22}+p_{33}} &\geq \frac{2p_{02}}{2 (p_{02}+p_{13})},
\end{align}
which implies $p_{13}(p_{00}+p_{11}) \geq p_{02}(p_{22}+p_{33})$. Moreover, because $p_{00}+p_{11}\geq 2 p_{01}$, we have
\begin{align}
    p_{13} (p_{00}+p_{11}) &\geq 2 p_{01} p_{13}\geq p_{01} p_{13}+p_{02} p_{33}\nonumber\\
    &\geq p_{02} p_{22} + p_{02} p_{33},
\end{align}
where in the last step the condition $\delta_i \geq \delta_{i+1}$ was used. This proves inequality~(\ref{item bb4}).

Inequality~(\ref{item cc4}) can be {recast} as
\begin{align}
    \frac{p_{3}}{p_{00}+p_{11}+p_{22}+p_{33}} &\leq \frac{p_{1}}{2 (p_{02}+p_{13})},
\end{align}
from which we arrive at
\begin{align}
    2p_3 p_{02}+ p_3p_{13}  &\leq p_1 (p_{01}+p_{10}+p_{22})\nonumber\\
    &\leq p_1 (p_{00}+p_{11}+p_{22}),
\end{align}
which is true for any energy spectrum.
\end{proof}
Unfortunately, Eq.~(\ref{lemma:qmaj2ri4d}) in Lemma \ref{lemma:PassNorm4d} fails when \(\delta_i \ngeq \delta_{i+1}\).\\

\noindent\textbf{Counterexample to Eq.}~(\ref{lemma:qmaj2ri4d}) for $\delta_i \,{<}\, \delta_{i+1}$.\
The relation \(\frac{\vec{q}}{\lVert \vec{q} \rVert} \succ \frac{\mathds{1}+\Pi}{2} \frac{\vec{r}_1}{\lVert \vec{r}_1 \rVert}\) fails in the regime \(\delta_i \,{<}\, \delta_{i+1}\). In that regime, the greatest {element} of \(\frac{\mathds{1}+\Pi}{2} \frac{\vec{r}_1}{\lVert \vec{r}_1 \rVert}\) is \(\frac{p_{12}+p_{01}}{2(p_{01}+p_{12}+p_{23}+p_{30})}\) and so
\begin{equation}
\begin{aligned}
&\frac{p_{00}}{p_{00}+p_{11}+p_{22}+p_{33}} \geq \frac{p_{12}+p_{01}}{2(p_{01}+p_{12}+p_{23}+p_{30})}\\
\end{aligned}
\end{equation}
must hold in order for the relation to be true. {But} in the limit \(E_3 \rightarrow \infty, E_1 \rightarrow 0, E_2 \rightarrow 0\) the above inequality becomes
\begin{equation}
\begin{aligned}
\frac{p_{00}}{3 p_{00}} \geq \frac{2 p_{00}}{2(p_{00}+p_{00})}\Leftrightarrow \frac{1}{3} \geq \frac{1}{2},
\end{aligned}
\end{equation}
which is obviously {invalid}. The relation \(\frac{\vec{q}}{\lVert \vec{q} \rVert} \succ \frac{\vec{r}_1}{\lVert \vec{r}_1 \rVert}\) also fails sometimes. Looking at the limit \(E_3 \rightarrow \infty\), \(\frac{\vec{q}}{\lVert \vec{q} \rVert} \rightarrow (p_{00}, p_{11}, p_{22}, 0)\) and \(\frac{\vec{r}_1}{\lVert \vec{r}_1 \rVert} \rightarrow (p_{02}, 0 , p_{20}, 0)\). For the majorisation relation to hold, we would need \(p_{22} \leq 0\){,} which is clearly not true in general.\\

\noindent
\emph{Proof of Eq.}~\emph{(\ref{lemma:rismallerbi4d})} \emph{in Lemma}~\emph{\ref{lemma:PassNorm4d}}.\ We need to prove that $\lVert \vec{r}_i\rVert \leq \lVert \vec{b}_i\rVert${,} for $i=1,2${,} for $d=4$. To {do so, we again use} the first derivative of $\lVert \vec{r}_i\rVert$ and show that it is always negative. {From}
\begin{equation}
\lVert \vec{r}_i\rVert=\sum_{j}p_j p_{j+i}=\frac{\sum_j e^{-\beta(E_j+E_{j+i})}}{\sum_{m,n}e^{-\beta(E_m+E_n)}}{,}
\end{equation}
the partial derivative with respect to $\beta$ reads
\begin{align}
    &\partial_\beta \lVert\vec{r}_i\rVert =\sum_{j,m,n}(E_m+E_n-E_j-E_{j+i})\tfrac{e^{-\beta(E_m+E_n+E_j+E_{j+i})}}{z(\beta)^4}\nonumber\\
    &=\tfrac{1}{z(\beta)^3}\sum_{j,m}(2E_m-E_j-E_{j+i})e^{-\beta(E_m+E_j+E_{j+i})}.
\label{rjprime}
\end{align}
Now we need to show that for any $i\neq 0$, $\partial_\beta \lVert {\vec{r}_{i}}\rVert$ is negative. For $i=1$ we find
\begin{widetext}
\begin{align}
    \partial_\beta \lVert \vec{r}_1\rVert
    &=\frac{-1}{z(\beta)^3}\big[E_1(e^{-\beta E_1}+e^{-\beta(E_1+E_2)}-e^{-2\beta E_1})+(E_2-E_1)(-e^{-\beta(E_1+E_2)}+e^{-\beta(2E_1+E_2)}\nonumber\\[1mm]
    &\ +e^{-\beta(E_1+E_2+E_3)}-e^{-\beta(E_1+2E_2)})+(E_3-E_2)(e^{-\beta(E_2+E_3)}-e^{-\beta(E_1+E_2+E_3)}+e^{-\beta(2E_2+E_3)}-e^{-\beta(E_2+2E_3)})\nonumber\\[1mm]
    &\ +E_3(e^{-\beta(E_2+E_3)}+e^{-\beta E_3}-e^{-\beta(E_1+ E_3)}-e^{-2\beta E_3})\big]\leq 0.
\label{r1prime}
\end{align}
Since we have ordered the energy eigenvalues in {the} increasing order, it is straightforward to show that $\partial_\beta \lVert \vec{r}_1\rVert$ is always nonpositive for any set of energy eigenvalues in dimension $4$.

{Now similarly, we can} show that $ \lVert \vec{r}_2\rVert \leq \lVert \vec{b}_2\rVert$. By employing Eq.~(\ref{rjprime}), $\partial_\beta \lVert \vec{r}_2\rVert$ {is} obtained as
\begin{align}
\partial_\beta \lVert \vec{r}_2 \rVert
&=\frac{-2}{z(\beta)^3}\big[E_1( e^{-\beta E_2}+2e^{-\beta (E_1 +E_3)}-e^{-\beta (E_1+ E_2)}-e^{-2 \beta E_2}-e^{-\beta ( E_2+E_3)})+(E_2-E_1)(e^{-\beta E_2}+e^{-\beta (E_1 +E_3)}\nonumber\\[1mm]
&\ +e^{-\beta (E_1+ E_2)}+e^{-\beta (2E_1+ E_3)}-e^{-2 \beta E_2}-e^{-\beta (E_1+ E_2+E_3)}-e^{-\beta ( E_2+E_3)}-e^{-\beta ( E_1+2E_3)})\nonumber\\[1mm]
&\ +(E_3-E_2)(e^{-\beta (E_1 +E_3)}+e^{-\beta (2E_1+ E_3)}+e^{-\beta (E_1+ E_2+E_3)}-2e^{-\beta ( E_2+E_3)}-e^{-\beta ( E_1+2E_3)})\big].
\end{align}
The above expression is also nonpositive under {the} assumption $\delta_i\,{\geq}\, \delta_{i+1}$.\hfill \qed
\end{widetext}

\subsection{Proof of monotonicity and convexity of the thermal curve}\label{app:CEforConditionii}

In this section, we will show that the condition~(\ref{item ii}) mentioned in Sec.~\ref{sec:two-qutrit case_geometric} is true for every choice of the set $\{E_i\}_{i=0}^{d-1}$ and any initial inverse temperature.

Expressing a general point in the curve $\vec p(\beta\pr)=(x_0(\beta\pr),\dots,x_{d-2}(\beta\pr),-1)$ as a linear combination of the vertices $\vec{v}_i$ in condition~(\ref{item i}) we have
\begin{equation}
\begin{pmatrix}x_0(\beta\pr) \\
x_1(\beta\pr) \\
\vdots \\
x_{d-2}(\beta\pr) \\
-1
\end{pmatrix}=\begin{pmatrix} x_0(\beta) & 0 & 0 &\dots & 0 \\
x_1(\beta) & x_1(\beta) & 0 &\dots & 0 \\
\vdots & \ddots & x_2(\beta) &\ddots & 0 \\
x_{d-2}(\beta) & \dots & \ddots &\ddots & 0 \\
-1 & -1 & -1 &\dots & -1 \end{pmatrix}
\begin{pmatrix}
a_0 \\
a_1 \\
\vdots \\
a_{d-2} \\
a_{d-1}
\end{pmatrix},
\end{equation}
where the $(a_0,\dots,a_{d-1})$ are coefficients given by
\begin{equation}
    a_i = \frac{x_i(\beta\pr)}{x_i(\beta)} - \frac{x_{i-1}(\beta\pr)}{x_{i-1}(\beta)} = \frac{x_i(\beta\pr)}{x_{i-1}(\beta)} \left(\frac{x_{i-1}(\beta)}{x_i(\beta)} - \frac{x_{i-1}(\beta\pr)}{x_i(\beta\pr)} \right) ,
\end{equation}
which are positive for all $i$ and satisfy $\sum_i a_i=1$
(and thus below $1$ for all $i$) if the following condition holds
    \begin{align}
        \frac{\rm d}{{\rm d} \beta}\Bigl(\frac{x_{m}}{x_{m+1}} (\beta) \Bigr)  &=\frac 1 {x_{m+1}}\,\Bigl(\frac{\partial x_{m}}{\partial \beta} - \frac{x_{m}}{x_{m+1}}\frac{\partial x_{m+1}}{\partial \beta}\Bigr)\,\geq\,0\ \forall m ,
    \end{align}
which means that the function $x_i/x_{i+1} (\beta)$ is monotonically decreasing with $\beta$. This in turn means that the final point can be reached as a convex combination of the vertices.

With a convenient relabeling of the index we can write $x_{m+1}/x_{m}(\beta)$ in the form
\begin{align}
\tfrac{x_{m}}{x_{m-1}} &=\tfrac{(m+1)e^{-\beta E_{m+1}}-\sum_{i=0}^{m}e^{-\beta E_i}}{me^{-\beta E_{m}}-\sum_{i=0}^{m-1}e^{-\beta E_i}}\nonumber\\
 &=\tfrac{(m+1)e^{-\beta E_{m+1}}-me^{-\beta E_{m}}+me^{-\beta E_{m}}-\sum_{i=0}^{m-1}e^{-\beta E_i}-e^{-\beta E_{m}}}{me^{-\beta E_{m}}-\sum_{i=0}^{m-1}e^{-\beta E_i}}\nonumber\\
 &= 1+\tfrac{(m+1)(e^{-\beta E_{m+1}}-e^{-\beta E_{m}})}{me^{-\beta E_{m}}-\sum_{i=0}^{m-1}e^{-\beta E_i}}.
\end{align}
To prove our claim [condition~(\ref{item ii})], we need to show that for any $\beta\geq0$ the function $x_{m}/x_{m-1}$ is monotonically decreasing with $\beta$. That is, we have to show that
\begin{align}
 R(\beta)   &=\,\frac{\rm d}{\rm d \beta}\Bigl(\frac{x_{m}}{x_{m-1}} (\beta) \Bigr) \,\leq\, 0.
 \label{ineq R}
\end{align}
The derivative takes the form
\begin{widetext}
\begin{align}
   R(\beta) &=\tfrac{m+1}{x_{m-1}^2} \Bigl[-(E_{m+1}e^{-\beta E_{m+1}}-E_{m}e^{-\beta E_{m}})(me^{-\beta E_{m}}-\sum_{i=0}^{m-1}e^{-\beta E_i})+(e^{-\beta E_{m+1}}-e^{-\beta E_{m}})(m\,E_{m}e^{-\beta E_{m}}-\sum_{i=0}^{m-1}E_{i}e^{-\beta E_i})\Bigr]\nonumber\\
   &=\tfrac{m+1}{x_{m-1}^2}\Bigl(\sum_{i=0}^{m-1}[(E_{m+1}-E_i)e^{-\beta( E_{m+1}+E_i)}-(E_{m}-E_i)e^{-\beta( E_{m}+E_i)}]-m (E_{m+1}-E_m)e^{-\beta( E_{m+1}+E_m)}\Bigr)\nonumber\\
   &=\tfrac{m+1}{x_{m-1}^2}\Bigl(\sum_{i=0}^{m-1}[(E_{m+1}-E_i)e^{-\beta( E_{m+1}+E_i)}-(E_{m}-E_i)e^{-\beta( E_{m}+E_i)}]-m (E_{m+1}-E_i+E_i-E_m)e^{-\beta( E_{m+1}+E_m)}\Bigr)\nonumber\\
      &=\tfrac{m+1}{x_{m-1}^2}\sum_{i=0}^{m-1}\bigl[(E_{m+1}-E_i)(e^{-\beta( E_{m+1}+E_i)}-e^{-\beta( E_{m+1}+E_m)})-(E_{m}-E_i)(e^{-\beta( E_{m}+E_i)}-e^{-\beta( E_{m+1}+E_m)})\bigr]\nonumber\\
      &=\tfrac{m+1}{x_{m-1}^2}e^{-\beta( E_{m+1}+E_m)}\sum_{i=0}^{m-1}\bigl[(E_{m+1}-E_i)(e^{\beta( E_{m}-E_i)}-1)-(E_{m}-E_i)(e^{\beta( E_{m+1}-E_i)}-1)\bigr].
\end{align}
To prove the nonpositivity of $R(\beta)$ we need to show that
\begin{equation}
  (E_{m+1}-E_i)(e^{\beta( E_{m}-E_i)}-1)-(E_{m}-E_i)(e^{\beta( E_{m+1}-E_i)}-1)\leq 0 \ \ \ \text{for}\ \ E_i\leq E_m \leq E_{m+1}.
  \label{eq: negativity R}
\end{equation}
To do so, we define $y_m:=\beta(E_{m}-E_i)$, and write inequality~(\ref{eq: negativity R}) in the form
\begin{align}
   h(y_{m+1})   &=\, \frac{y_{m+1}}{e^{y_{m+1}}-1}\,\leq\, \frac{y_{m}}{e^{y_{m}}-1}.
\end{align}
This inequality is satisfied for all $y_{m+1}\geq y_{m}$ if the function $h(y)$ is monotonically decreasing. To see that this is the case, we calculate the derivative
\begin{align}
    \frac{\partial h(y)}{\partial y}    &=\,\frac{(1-y)e^{y}-1}{(e^{y}-1)^2},
\end{align}
which can be seen to be smaller or equal to zero since $(1-y)e^{y}$ is a monotonically decreasing function for $y\geq0$, i.e., $\tfrac{\partial}{\partial y}(1-y)e^{y}=-y e^{y}\leq0$ and hence has its maximum value of $1$ at $y=0$, confirming that $(1-y)e^{y}-1\leq0$ and that $h(y)$ is monotonically decreasing.
\end{widetext}


\subsubsection{Explicit partial derivatives in dimension 3}

To show explicitly that the curve $x(\beta)$ is monotonically increasing in $d=3$, we first calculate
\begin{subequations}\label{eq:partials of x and y wrt beta}
  \begin{align}
    \tfrac{\partial x}{\partial\beta}  &=\,-\,\mathcal{Z}^{-2}f_{x}(\beta), \label{eq:partial x wrt beta1}\\
    \tfrac{\partial y}{\partial\beta}  &=\,-\,3\mathcal{Z}^{-2}f_{y}(\beta), \label{eq:partial y wrt beta1}
  \end{align}
\end{subequations}
where the functions $f_{x}(\beta)$ and $f_{y}(\beta)$ are given by
\begin{subequations}\label{eq:fx and fy1}
  \begin{align}
    f_{x}(\beta)  &=\,E_{1}e^{-\beta E_{1}}(2+e^{-\beta E_{2}})+E_{2}e^{-\beta E_{2}}(1-e^{-\beta E_{1}})\,\geq\,0, \label{eq:fx1}\\
    f_{y}(\beta)  &=\,(E_{2}-E_{1})e^{-\beta (E_{1}+E_{2})}+E_{2}e^{-\beta E_{2}}\,\geq\,0. \label{eq:fy1}
  \end{align}
\end{subequations}
With this, we can then evaluate
\begin{align}
    \tfrac{\partial y}{\partial x}    &=\,\left(\tfrac{\partial x}{\partial\beta}\right)^{-1}\tfrac{\partial y}{\partial\beta}\,=\,
    \tfrac{3f_{y}(\beta)}{f_{x}(\beta)}\,\geq\,0,
    \label{eq:partial y wrt x1}
\end{align}
which allows us to confirm that the curve segment lies above (with respect to the coordinates $x$ and $y$ in the plane of the polytope) the line connecting $\bigl(x(\beta),y(\beta)\bigr)$ and $\bigl(\tilde{x},\tilde{y}\bigr)$.

To show that this curve is a convex function, we need to show that $\partial^{2}y/\partial x^{2}\geq0$, and calculate
\begin{align}
    \tfrac{\partial^{2}y}{\partial x^{2}}    &=\,\left(\tfrac{\partial x}{\partial\beta}\right)^{-1}\tfrac{\partial}{\partial\beta}\,\left[\left(\tfrac{\partial x}{\partial\beta}\right)^{-1}\tfrac{\partial y}{\partial\beta}\right]\nonumber\\[1.5mm]
    &=\,\left(\tfrac{\partial x}{\partial\beta}\right)^{-2}\left\{\tfrac{\partial^{2}y}{\partial \beta^{2}}-\left(\tfrac{\partial y}{\partial \beta}\right)\left(\tfrac{\partial x}{\partial\beta}\right)^{-1}\left(\tfrac{\partial^{2}x}{\partial \beta^{2}}\right)\right\}.
    \label{eq:2nd partial y wrt x}
\end{align}
The second partial derivatives with respect to $\beta$ are found to be
\begin{subequations}\label{eq:second partials of x and y wrt beta}
  \begin{align}
    \tfrac{\partial x^{2}}{\partial\beta^{2}}  &=\,\mathcal{Z}^{-3}\bigl[2\tfrac{\partial\mathcal{Z}}{\partial\beta}f_{x}-\mathcal{Z}\tfrac{\partial f_{x}}{\partial\beta}\bigr], \label{eq:2nd partial x wrt beta}\\[1.5mm]
    \tfrac{\partial y^{2}}{\partial\beta^{2}}  &=\,3\mathcal{Z}^{-3}\bigl[2\tfrac{\partial\mathcal{Z}}{\partial\beta}f_{y}-\mathcal{Z}\tfrac{\partial f_{y}}{\partial\beta}\bigr]. \label{eq:2nd partial y wrt beta}
  \end{align}
\end{subequations}
Combining this with Eq.~(\ref{eq:2nd partial y wrt x}) we find
\begin{align}
    \tfrac{\partial^{2}y}{\partial x^{2}}    &=\,3\,\mathcal{Z}^{2}f_{x}^{-3}\bigl(f_{y}\frac{\partial f_{x}}{\partial\beta}-f_{x}\tfrac{\partial f_{y}}{\partial\beta}\bigr).
    \label{eq:2nd partial y wrt x rewritten}
\end{align}
The derivatives of $f_{x}(\beta)$ and $f_{y}(\beta)$ are
\begin{subequations}\label{eq:derivatives of fx and fy}
  \begin{align}
    \tfrac{\partial f_{x}}{\partial\beta}  &=\,\bigl(E_{2}^{2}-E_{1}^{2}\bigr)e^{-\beta (E_{1}+E_{2})}-2E_{1}^{2}e^{-\beta E_{1}}-E_{2}^{2}e^{-\beta E_{2}},\label{eq:derivative of fx}\\
    \tfrac{\partial f_{y}}{\partial\beta}  &=\,-\,\bigl(E_{2}^{2}-E_{1}^{2}\bigr)e^{-\beta (E_{1}+E_{2})}\,-\,E_{2}^{2}e^{-\beta E_{2}}, \label{eq:derivative of fy}
  \end{align}
\end{subequations}
which we insert into Eq.~(\ref{eq:2nd partial y wrt x rewritten}) to arrive at
\begin{align}
    \tfrac{\partial^{2}y}{\partial x^{2}}    &=\,6\,\mathcal{Z}^{3}f_{x}^{-3}
    E_{1}E_{2}\bigl(E_{2}-E_{1}\bigr)e^{-\beta(E_{1}+E_{2})}\,\geq\,0,
    \label{eq:2nd partial y wrt x rewritten final}
\end{align}
where we have made use of the fact that $f_{x}\geq0$ (since $1\geq e^{-\beta E_{1}}$), along with $E_{1},E_{2},\mathcal{Z}\geq 0$ and $E_{2}\geq E_{1}$.


\subsubsection{Explicit partial derivatives in dimension 4}\label{app:derivatives}

In this appendix, we present the explicit expressions for the first and second partial derivatives. The derivatives of the coordinates $x(\beta)$, $y(\beta)$ and $z(\beta)$ along the curve of thermal states can be written as
\begin{subequations}\label{eq:partials of x and y wrt beta1}
  \begin{align}
    \frac{\partial x}{\partial\beta}  &=\,-\,\mathcal{Z}^{-2}f_{x}(\beta), \label{eq:partial x wrt beta}\\
    \frac{\partial y}{\partial\beta}  &=\,-\,\mathcal{Z}^{-2}f_{y}(\beta), \label{eq:partial y wrt beta}\\
    \frac{\partial z}{\partial\beta}  &=\,-\,\mathcal{Z}^{-2}f_{z}(\beta), \label{eq:partial z wrt beta}
  \end{align}
\end{subequations}
where the functions $f_{x}(\beta)$, $f_{y}(\beta)$, and $f_{z}(\beta)$ are given by
\begin{subequations}\label{eq:fx and fy}
  \begin{align}
    f_{x}(\beta)  &=\,E_{1}e^{-\beta E_{1}}(2+e^{-\beta E_{2}}+e^{-\beta E_{3}})\nonumber\\[1mm]
    & \ \ +(E_{2}e^{-\beta E_{2}}+E_{3}e^{-\beta E_{3}})(1-e^{-\beta E_{1}})\,\geq\,0, \label{eq:fx}\\[1.5mm]
    f_{y}(\beta)  &=\,3(E_{2}-E_{1})e^{-\beta (E_{1}+E_{2})}+3E_{2}e^{-\beta E_{2}}\nonumber\\[1mm]
    & \ \ +E_{3}e^{-\beta E_{3}}(1+e^{-\beta E_{1}}-2e^{-\beta E_{2}})\nonumber\\[1mm]
    & \ \ +e^{-\beta E_{3}}(2E_{2}e^{-\beta E_{2}}-E_{1}e^{-\beta E_{1}})\geq0,\label{eq:fy}\\[1.5mm]
    f_{z}(\beta)  &=\,4e^{-\beta E_{3}}\bigl[E_{3}+e^{-\beta E_{1}}(E_{3}-E_{1})\nonumber\\[1mm]
    & \ \ \
    +e^{-\beta E_{2}}(E_{3}-E_{2})\bigr]\geq0,
  \end{align}
\end{subequations}
where $f_{y}\geq0$ follows, since $E_{3}\geq E_{2}$ and $(1+e^{-\beta E_{1}}-2e^{-\beta E_{2}})\geq0$ which implies that the terms in $f_{y}$ proportional to $e^{-\beta E_{3}}$ can be bounded by $E_{3}e^{-\beta E_{3}}\bigl(E_{2}+e^{-\beta E_{1}}(E_{2}-E_{1})\bigr)\geq0$. Since $f_{x}$, $f_{y}$, and $f_{z}$ are nonnegative, all partial derivatives with respect to $\beta$ are nonpositive, but
\begin{align}
    \frac{\partial y}{\partial x}    &=\,\left(\frac{\partial x}{\partial\beta}\right)^{-1}\frac{\partial y}{\partial\beta}\,=\,
    \frac{f_{y}(\beta)}{f_{x}(\beta)}\,\geq\,0,
    \label{eq:partial y wrt x}\\[1mm]
    \frac{\partial z}{\partial x}    &=\,\left(\frac{\partial x}{\partial\beta}\right)^{-1}\frac{\partial z}{\partial\beta}\,=\,
    \frac{f_{z}(\beta)}{f_{x}(\beta)}\,\geq\,0,
    \label{eq:partial z wrt x}\\[1mm]
    \frac{\partial z}{\partial y}    &=\,\left(\frac{\partial y}{\partial\beta}\right)^{-1}\frac{\partial z}{\partial\beta}\,=\,
    \frac{f_{z}(\beta)}{f_{y}(\beta)}\,\geq\,0.
    \label{eq:partial z wrt y}
\end{align}
For the second partial derivatives, we have
\begin{align}
    \frac{\partial^{2}y}{\partial x^{2}}    &=\,\left(\frac{\partial x}{\partial\beta}\right)^{-1}\frac{\partial}{\partial\beta}\,\left[\left(\frac{\partial x}{\partial\beta}\right)^{-1}\frac{\partial y}{\partial\beta}\right]\nonumber\\[1mm]
    &=\,\left(\frac{\partial x}{\partial\beta}\right)^{-2}\left\{\frac{\partial^{2}y}{\partial \beta^{2}}-\left(\frac{\partial y}{\partial \beta}\right)\left(\frac{\partial x}{\partial\beta}\right)^{-1}\left(\frac{\partial^{2}x}{\partial \beta^{2}}\right)\right\}\nonumber\\[1mm]
    & =\,\mathcal{Z}^{2}f_{x}^{\nr -3}\bigl(f_{y}\frac{\partial f_{x}}{\partial\beta}-f_{x}\frac{\partial f_{y}}{\partial\beta}\bigr),
    \label{eq:2nd partial y wrt x alternative}\\[1.5mm]
    \frac{\partial^{2}z}{\partial x^{2}}    &=\,
    \mathcal{Z}^{2}f_{x}^{\nr -3}\bigl(f_{z}\frac{\partial f_{x}}{\partial\beta}-f_{x}\frac{\partial f_{z}}{\partial\beta}\bigr),\\[1.5mm]
    \frac{\partial^{2}z}{\partial y^{2}}    &=\,
    \mathcal{Z}^{2}f_{y}^{\nr -3}\bigl(f_{z}\frac{\partial f_{y}}{\partial\beta}-f_{y}\frac{\partial f_{z}}{\partial\beta}\bigr).
\end{align}
We thus have to calculate the derivatives of the functions $f_{x}$, $f_{y}$, and $f_{z}$ with respect to $\beta$, which evaluate to
\begin{align}
    \frac{\partial f_{x}}{\partial \beta}    &=\,-E_{1}^{2}e^{-\beta E_{1}}\bigl(2+e^{-\beta E_{2}}+e^{-\beta E_{3}}\bigr)
    \nonumber\\
    & \ \ -\bigl(E_{2}^{2}e^{-\beta E_{2}}+E_{3}^{2}e^{-\beta E_{3}}\bigr)\bigl(1-e^{-\beta E_{1}}\bigr)\,,
    \label{eq:partial fx wrt beta}\\[1.5mm]
    \frac{\partial f_{y}}{\partial \beta}    &=\,E_{1}^{2}e^{-\beta E_{1}} \bigl(3e^{-\beta E_{2}}+e^{-\beta E_{3}}\bigr)\nonumber\\
    & \ \ -E_{2}^{2}e^{-\beta E_{2}} \bigl(3+3e^{-\beta E_{1}}+2e^{-\beta E_{3}}\bigr)\nonumber\\
    & \ \ -E_{3}^{2}e^{-\beta E_{3}} \bigl(1+e^{-\beta E_{1}}-2e^{-\beta E_{2}}\bigr)\,,
    \label{eq:partial fy wrt beta}\\[1.5mm]
    \frac{\partial f_{z}}{\partial \beta}    &=\,4e^{-\beta E_{3}} \bigl[E_{1}^{2}e^{-\beta E_{1}}+E_{2}^{2}e^{-\beta E_{2}}
    \nonumber\\
    & \ \ \ -\,E_{3}^{2}\bigl(1+e^{-\beta E_{1}}+e^{-\beta E_{2}}\bigr)\bigr]\,.
    \label{eq:partial fz wrt beta}
\end{align}
Then, since $\mathcal{Z}^{2}\geq0$ and $f_{x},f_{y},f_{z}\geq0$, we just have to evaluate
\vspace*{-6mm}
\begin{widetext}
\begin{subequations}
\begin{align}
    f_{y}\frac{\partial f_{x}}{\partial\beta}-f_{x}\frac{\partial f_{y}}{\partial\beta}    &=\,
    2\bigl[E_{1}E_{2}e^{-\beta(E_{1}+E_{2})}(E_{2}-E_{1})  +  E_{1}E_{3}e^{-\beta(E_{1}+E_{3})}(E_{3}-E_{1})\bigr]\bigl(3+e^{-\beta E_{3}}\bigr)\mathcal{Z}\nonumber\\
    & \ \ \ +\,
    2E_{2}E_{3}e^{-\beta(E_{2}+E_{3})}(E_{3}-E_{2})\bigl(1+e^{-\beta E_{1}}\bigr)\bigl(2+2e^{-\beta E_{1}}-e^{-\beta E_{2}}+e^{-\beta E_{3}}\bigr)
     \,\geq\,0,\\[2.5mm]
    f_{z}\frac{\partial f_{x}}{\partial\beta}-f_{x}\frac{\partial f_{z}}{\partial\beta} &=\,
    4\mathcal{Z}e^{-\beta E_{3}}\,\bigl[
    E_{1}e^{-\beta E_{1}}(E_{2}-E_{1}) \bigl(e^{-\beta E_{2}}(E_{3}-E-{2})+2E_{3}\bigr)\nonumber\\
    & \ \ \ +(E_{3}-E_{2})E_{3}\bigl(E_{1}e^{-\beta(E_{1}+E_{2})}+2E_{1}e^{-\beta E_{1}}+E_{2}e^{-\beta E_{2}}(1-e^{-\beta E_{1}})\bigr)
    \bigr]
     \,\geq\,0,\\[2.5mm]
    f_{z}\frac{\partial f_{y}}{\partial\beta}-f_{y}\frac{\partial f_{z}}{\partial\beta} &=\,
    12\mathcal{Z}e^{-\beta (E_{2}+E_{3})}\bigl(E_{3}-E_{2}\bigr)\,\bigl[
    E_{1}^{2}e^{-\beta E_{1}}+E_{3}\bigl(E_{2}-E_{1}\bigr)e^{-\beta E_{1}}+E_{2}\bigl(E_{3}-e^{-\beta E_{1}}E_{1}\bigr)
    \bigr]
     \,\geq\,0.
\end{align}
\end{subequations}
\end{widetext}
\clearpage


\subsection{Geometry of equal thermal marginals for two qutrits}\label{sec:two qutrit case}

For local dimension $d=3$, the transformations we consider here to ensure symmetric marginals imply that the vectors collecting the diagonal entries of the final state marginals are of the form
\begin{align}
 \mathbf{\tilde{p}}=\mathbf{\tilde{p}}\SA=\mathbf{\tilde{p}}\SB
    &= M_q \mathbf{q} + (\mathds{1}+ \Pi)M_{r} \mathbf{r}\,,
\label{eq:EqualEvoMarginal d3}
\end{align}
where $\mathbf{q}=(q_{i})$ with $q_{i}=p_{i}^{2}$ and $(\mathds{1}+ \Pi)\mathbf{r}=\bigl(p_{i}(1-p_{i})\bigr)$. This means that the set of reachable states corresponds to the convex hull of the points $\tilde{\mathbf{p}}\suptiny{0}{0}{(i,j)}:=\Pi\suptiny{0}{0}{(i)}\mathbf{q}+(\mathds{1}+ \Pi)\Pi\suptiny{0}{0}{(j)}\mathbf{r}$ for $i,j\in\{1,\ldots,3!\}$. However, in our case the problem can be further reduced to that of showing that the points $\mathbf{p}(\beta\pr<\beta)$ lie within a restricted part of the polytope, i.e., a triangle between the points $\mathbf{p}(\beta)$, $\mathbf{p}(\beta\rightarrow0)$, and $\tilde{\mathbf{p}}\suptiny{0}{0}{(6,1)}=\Pi\suptiny{0}{0}{(6)}\mathbf{q}+(\mathds{1}+ \Pi)\Pi\suptiny{0}{0}{(1)}\mathbf{r}$, where
\begin{align}
    \Pi\suptiny{0}{0}{(1)}=\mathds{1},\ \
    \Pi\suptiny{0}{0}{(6)} \,=\,\begin{pmatrix} 0 & 1 & 0 \\ 1 & 0 & 0 \\ 0 & 0 & 1 \end{pmatrix},\ \
    \Pi = \begin{pmatrix} 0 & 0 & 1 \\ 1 & 0 & 0 \\ 0 & 1 & 0 \end{pmatrix}.
    .
\end{align}
Note that the point $\vec{p}(\beta\rightarrow0)$ trivially lies in the polytope, since it can be obtained from an equally weighted convex combination of all cyclic permutations obtained for $\Pi\suptiny{0}{0}{(i)}=\Pi\suptiny{0}{0}{(j)}=\mathds{1},\Pi,\Pi^{2}$. Then we switch from the coordinates $\{p_{0},p_{1},p_{2}\}$ to the coordinates $\{x,y,z\}$, given by
\begin{subequations}
\label{eq:new coords3d}
    \begin{align}
        x   &=\,-\,p_{0}\,+\,p_{1},\\
        y   &=\,-\,p_{0}\,-\,p_{1}\,+\,2p_{2},\\
        z   &=\,-p_{0}\,-\,p_{1}\,-\,p_{2}\,=\,-1.
    \end{align}
\end{subequations}
The plane containing the polytope is defined by $z=-1$ such that $p_{2}=1-p_{0}-p_{1}$, and we therefore have $y=2-3p_{0}-3p_{1}$. The corners $(p_{0},p_{1},p_{2})=(1,0,0)$, $(0,1,0)$, and $(0,0,1)$ of the outer triangle containing all (diagonal) qutrit states are mapped to $(x,y)=(-1,-1)$, $(1,-1)$, and $(0,2)$, respectively, see Fig.~\ref{fig:polytope}. Moreover, the curve of thermal states is parameterised as $\bigl(x(\beta),y(\beta)\bigr)$ with $x(\beta)=-p_{0}(\beta)+p_{1}(\beta)$ and $y(\beta)=2-3p_{0}(\beta)-3p_{1}(\beta)$, and connects the points $(x,y)=(-1,-1)$ corresponding to $\mathbf{p}(\beta\rightarrow\infty)$ and $(x,y)=(0,0)$ corresponding to $\mathbf{p}(\beta\rightarrow0)$.

Next, we note that for any given thermal state with coordinates $\bigl(x(\beta),y(\beta)\bigr)$, the coordinates $\bigl(x=\tilde{x},y=\tilde{y}\bigr)$ of the corresponding point $\tilde{\vec{p}}\suptiny{0}{0}{(6,1)}$ are given by
\begin{subequations}
\label{eq:new coords tilde point}
\begin{align}
    \tilde{x}&=\,(p_{0}-p_{1})\bigl[2(p_{0}+p_{1})-1\bigr]\,\geq\,0,\\
    \tilde{y}&=\,2\,-\,3p_{0}\,-\,3p_{1}\,=\,y(\beta),
\end{align}
\end{subequations}
and consequently this point always lies on the same height $\tilde{y}=y(\beta)$ as the point corresponding to $\vec{p}(\beta)$, but to the right, i.e., $\tilde{x}\geq0\geq x(\beta)$ of any of the points corresponding to thermal states (in particular $\mathbf{p}(\beta\rightarrow0)$). This shows that the point with coordinates $(x_{0},x_{1})=(0,y(\beta))$ is included in the polytope, since it can be obtained as a convex combination of $\vec{p}(\beta)$ and $\tilde{\vec{p}}\suptiny{0}{0}{(6,1)}$.

Moreover, evaluating the partial derivatives, we have found $(\partial y(\beta)/\partial x(\beta)\geq0)$ and $(\partial^{2} y/\partial x^{2}\geq0)$, see Appendix~\ref{app:derivatives}, i.e., the curve described by $\bigl(x(\beta),y(\beta)\bigr)$ is a monotonically increasing $(\partial y(\beta)/\partial x(\beta)\geq0)$ and convex $(\partial^{2} y/\partial x^{2}\geq0)$ function, implying that the curve segment corresponding to all points $\bigl(x(\beta\pr),y(\beta\pr)\bigr)$ with $\beta\pr\leq\beta$ is contained within the triangle $\mathbf{p}(\beta)$, $\tilde{\mathbf{p}}\suptiny{0}{0}{(6,1)}$, and $\mathbf{p}(\beta\rightarrow0)$.


\subsection{Geometry of equal thermal marginals for two ququarts}\label{sec:two ququart case}

\begin{figure*}[ht!]
(a)\includegraphics[width=0.46\textwidth,trim={0cm 5mm 0cm 0mm}]{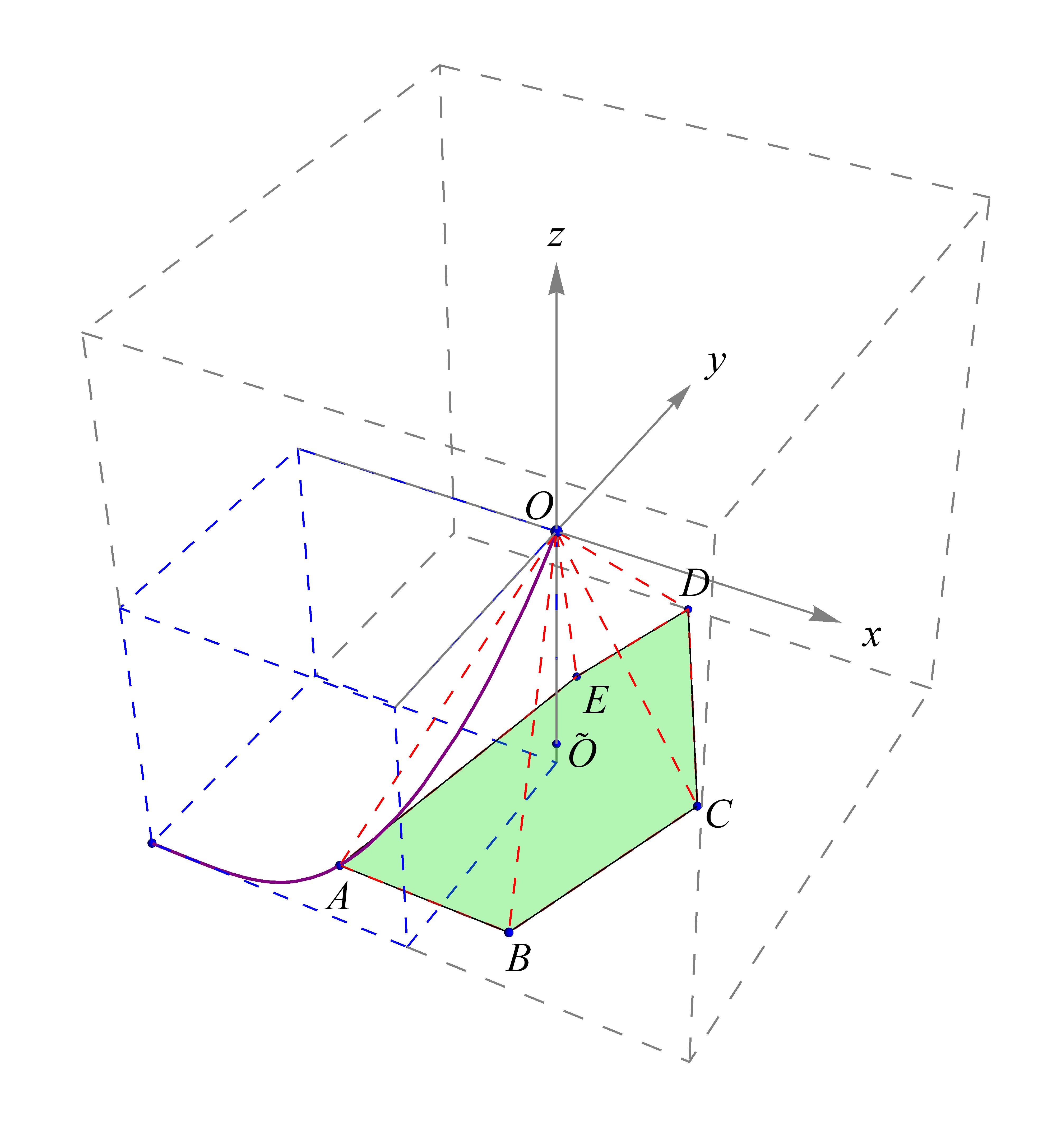}
(b) \includegraphics[width=0.45\textwidth,trim={0cm 0mm 0cm 0mm}]{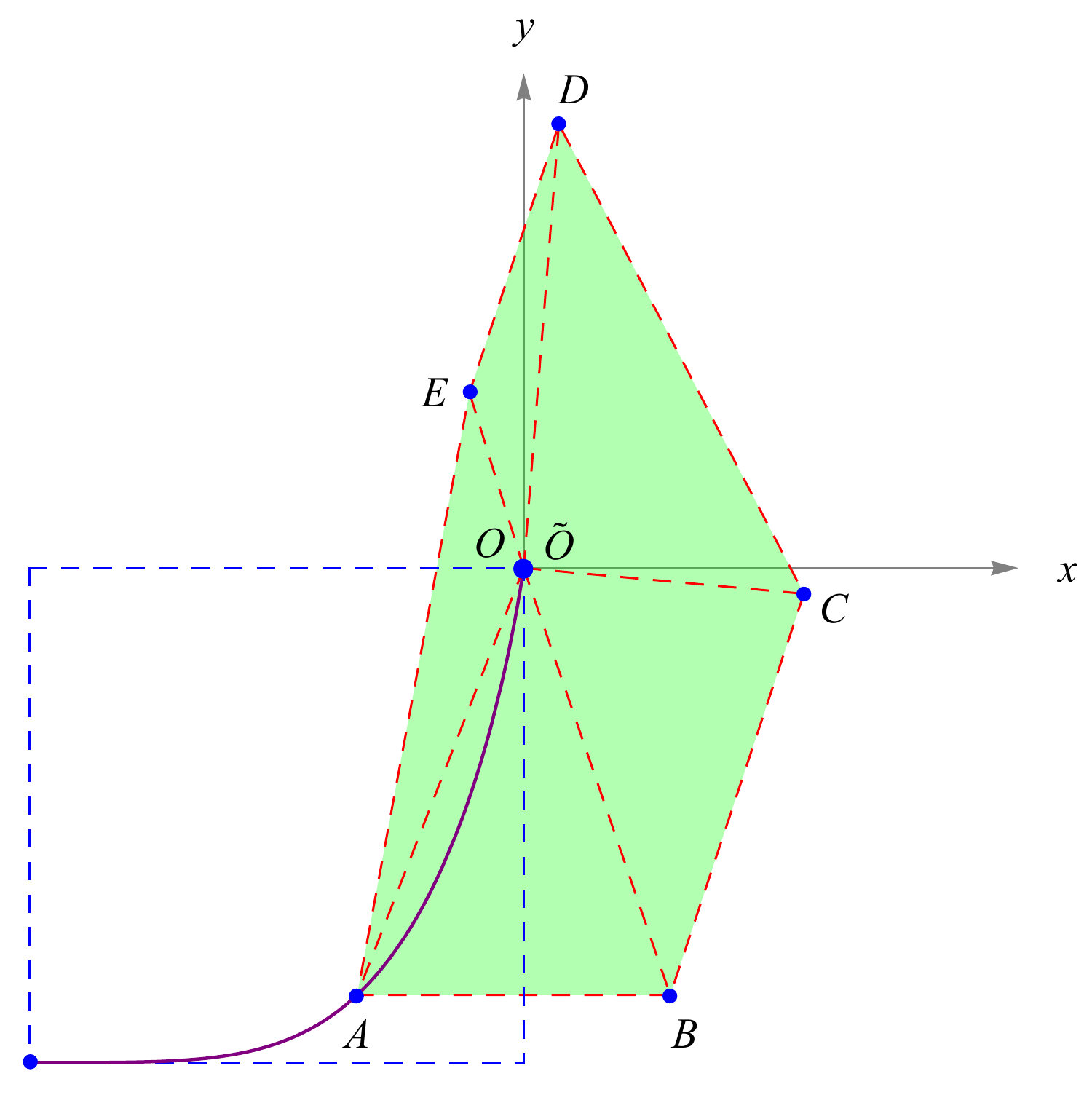}\\
(c)\includegraphics[width=0.44\textwidth,trim={0cm 0mm 0cm 0mm}]{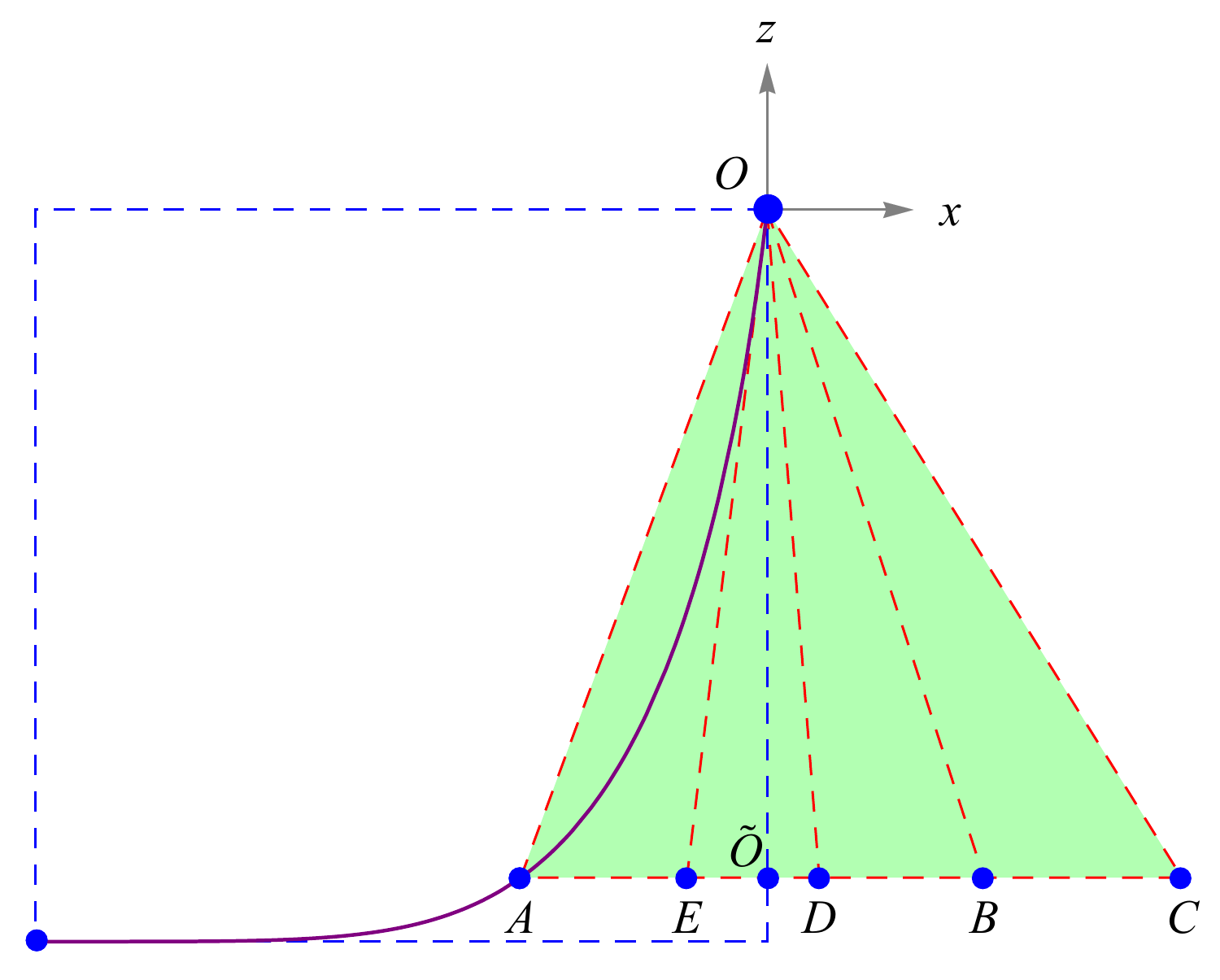}\ \ \
(d)\includegraphics[width=0.44\textwidth,trim={0cm 0mm 0cm 0mm}]{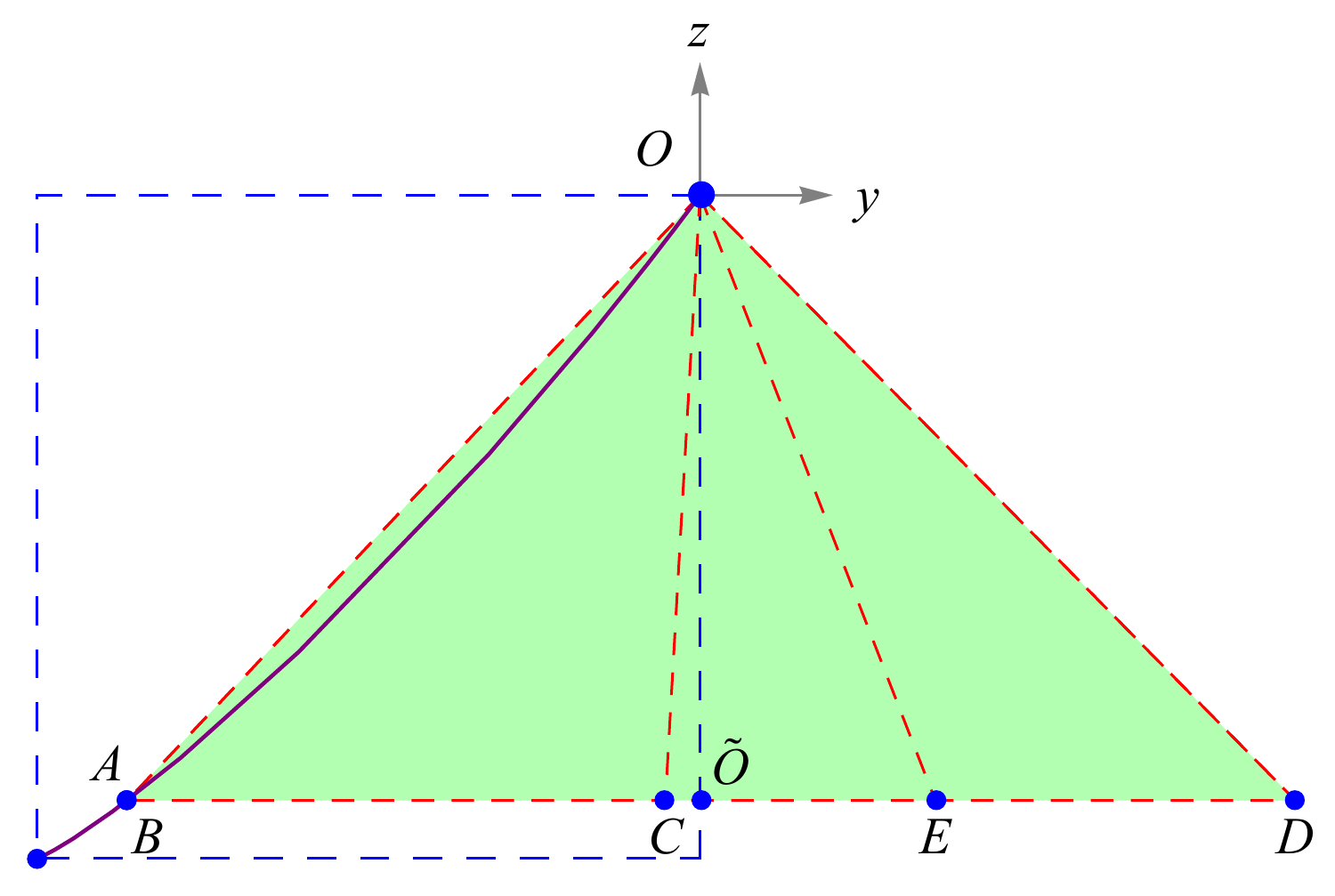}
\caption{\textbf{Illustration of geometric method for $d=4$}. (a) In local dimension $d=4$, the space of marginal probability distributions can be mapped to the hyperplane of constant coordinate $w=-1$, leaving a three-dimensional space with coordinates $x$, $y$, and $z$. For any fixed Hamiltonian, the thermal states form a one-parameter family of states along a line connecting $(x,y,z)=(-1,-1,-1)^{T}$ with $(0,0,0)^{T}$, where the curve is strictly confined to the region $x,y,z\leq0$ (dashed blue). The dashed red lines delineate the minimal polytope spanned by the points $A$, $B$, $C$, $D$, $E$ (which all share the same $z$-coordinate $z=z(\beta)$) and the origin $O$. (b), (c), (d) show the projections onto the $x$-$y$, $x$-$z$ and $y$-$z$ plane, respectively.}
\label{fig:polytopedimfour}
\end{figure*}

For local dimension $d=4$, we have marginal vectors of the form
\begin{align}
 \mathbf{\tilde{p}}=\mathbf{\tilde{p}}\SA=\mathbf{\tilde{p}}\SB
    &= M_q \mathbf{q} + (\mathds{1}+ \Pi)M_{r_{1}} \mathbf{r}_{1}\,+\,\tfrac{1}{2}(\mathds{1}+ \Pi^{2})M_{r_{2}} \mathbf{r}_{2}.
\label{eq:EqualEvoMarginal d4}
\end{align}
The set of reachable states corresponds to the convex hull of the points $\tilde{\mathbf{p}}\suptiny{0}{0}{(i,j,k)}:=\Pi\suptiny{0}{0}{(i)}\mathbf{q}+(\mathds{1}+ \Pi)\Pi\suptiny{0}{0}{(j)}\mathbf{r}_{1}+\tfrac{1}{2}(\mathds{1}+ \Pi^{2})\Pi\suptiny{0}{0}{(k)}\mathbf{r}_{2}$ for $i,j,k\in\{1,\ldots,4!\}$.
Switching from the coordinates $\{p_{0},p_{1},p_{2},p_{3}\}$ to the new coordinates $\{x,y,z,w\}$, given by
\begin{subequations}
\label{eq:new coords d4}
    \begin{align}
        x   &=\,-\,p_{0}\,+\,p_{1},\\
        y   &=\,-\,p_{0}\,-\,p_{1}\,+\,2p_{2},\\
        z   &=\,-\,p_{0}\,-\,p_{1}\,-\,p_{2}\,+\,3p_{3},\\
        w   &=\,-\,p_{0}\,-\,p_{1}\,-\,p_{2}\,-\,p_{3}\,=\,-1\,.
    \end{align}
\end{subequations}

We here show that the points $(x,y,z)=(0,y(\beta),z(\beta))$ and $(0,0,z(\beta))$ are included. To achieve this, we identify a set of 5 points that all lie in the plane of constant $z=z(\beta)$, and show that the two-dimensional polytope spanned by these points encloses both $(0,y(\beta),z(\beta))$ and $(0,0,z(\beta))$, see Fig.~\ref{fig:polytopedimfour}. These five points are $A=\tilde{\mathbf{p}}\suptiny{0}{0}{(1,1,1)}$, $B=\tilde{\mathbf{p}}\suptiny{0}{0}{(7,1,1)}$, $C=\tilde{\mathbf{p}}\suptiny{0}{0}{(13,7,1)}$, $D=\tilde{\mathbf{p}}\suptiny{0}{0}{(9,7,1)}$, and $E=\tilde{\mathbf{p}}\suptiny{0}{0}{(9,1,1)}$, where
\begin{align}
    \Pi &=\,\begin{pmatrix} 0 & 0 & 0 & 1 \\ 1 & 0 & 0 & 0 \\ 0 & 1 & 0 & 0 \\ 0 & 0 & 1 & 0 \end{pmatrix},\
    \Pi\suptiny{0}{0}{(7)} \,=\,\begin{pmatrix} 0 & 1 & 0 & 0 \\ 1 & 0 & 0 & 0 \\ 0 & 0 & 1 & 0 \\ 0 & 0 & 0 & 1 \end{pmatrix},\nonumber\\[1mm]
    \Pi\suptiny{0}{0}{(9)} &=\,\begin{pmatrix} 0 & 1 & 0 & 0 \\ 0 & 0 & 1 & 0 \\ 1 & 0 & 0 & 0 \\ 0 & 0 & 0 & 1 \end{pmatrix},\
    \Pi\suptiny{0}{0}{(13)} \,=\,\begin{pmatrix} 0 & 0 & 1 & 0 \\ 1 & 0 & 0 & 0 \\ 0 & 1 & 0 & 0 \\ 0 & 0 & 0 & 1 \end{pmatrix}.
\end{align}
The point $B$ can be seen to have the same $y$-coordinate as the starting point $A$, $y(B)=y(A)=y(\beta)$, while
\begin{align}
    x(B)    &=\,\bigl(p_{0}(\beta) - p_{1}(\beta)\bigr)\bigl(2p_{0}(\beta) + 2p_{1}(\beta)-1\bigr)\,\geq\,0,
\end{align}
since $p_{0}(\beta)\geq p_{1}(\beta)$ and $p_{0}(\beta)+p_{1}(\beta)\geq\tfrac{1}{2}$ for any Hamiltonian and any temperature. This means the point $(0,y(\beta),z(\beta))$ is contained in polytope since it can be obtained as a convex combination of $A$ and $B$.

Finally, to show that the point $\tilde{O}=(0,0,z(\beta))$ is contained within the polytope, we show that the points $C$, $D$ and $E$ satisfy, $x(C)\geq0$, $y(D)\geq0$ and $x(E)\leq 0$, while $(A\tilde{O}\times AE)_{z}\geq0$, $(E\tilde{O}\times ED)_{z}\geq0$ and $(D\tilde{O}\times DC)_{z}\geq0$, where for any two points $X$ and $Y$ we use the notation $XY=Y-X$. In other words, the closed path connecting the five points $A$, $B$, $C$, $D$ and $E$ encircles the point $\tilde{O}$, meaning that $\tilde{O}$ can be written as a convex combination of these points.

In a slight abuse of notation, we now use the shorthand $p_{i}\equiv p_{i}(\beta)$ for the thermal state diagonal components to express the relevant coordinates for the points $C$, $D$ and $E$ as
\begin{subequations}
\begin{align}
    x(C)    &=\,\bigl(p_{0} - p_{1}\bigr) \bigl(2 p_{0} + 2p_{1} + p_{2} - 1\bigr)\nonumber\\
    & \ \ + \bigl(p_{1} - p_{2}\bigr) \bigl(p_{0} + p_{1} + p_{2}\bigr)\,\geq\,0,\\
    y(D) &=\,3\bigl(p_{0} - p_{1}\bigr) \bigl(p_{0} + p_{1} + p_{2} - \tfrac{1}{3}\bigr)\nonumber\\
    & \ \ + 3\bigl(p_{1} - p_{2}\bigr) \bigl(p_{0} + p_{1} + p_{2}-\tfrac{2}{3}\bigr)\,\geq\,0,\\
    x(E) &=\,-\bigl[\bigl(p_{0} - p_{1}\bigr)\bigl(1 - p_{0} - p_{1}\bigr)\nonumber\\
    & \ \ \ \ + \bigl(p_{1} - p_{2}\bigr) \bigl(p_{1} + p_{2}\bigr)\bigr]\,\leq\,0,
\end{align}
\end{subequations}
while the $z$-components of the relevant cross products are given explicitly below. In Appendix~\ref{app:derivatives} we have further shown that the partial derivatives $\partial y/\partial x$, $\partial z/\partial x$, and $\partial z/\partial y$ as well as second derivatives $\partial^{2} y/\partial x^{2}$, $\partial^{2} z/\partial x^{2}$, and $\partial^{2} z/\partial y^{2}$ are nonnegative along the curve of thermal states, meaning that conditions~(\ref{item i}) and~(\ref{item ii}) are satisfied for $d=4$.


\subsubsection{Relevant cross products}\label{app:cross products}

Finally, the $z$-components of the relevant cross products are given
\clearpage
\begin{widetext}
\begin{subequations}
\begin{align}
    (A\tilde{O}\times AE)_{z}    &=\,\bigl(p_{0} - p_{1}\bigr)^{2} \bigl(2 p_{0} - 2p_{1} + 3p_{2}\bigr)
     \,+\, \bigl(p_{0} - p_{1}\bigr)\bigl(p_{1} - p_{2}\bigr) \bigl(p_{0} - p_{1} + 4p_{2}\bigr)
     \,+\, 2\bigl(p_{1} - p_{2}\bigr)^{2} \bigl(p_{1} + p_{2}\bigr)
     \,\geq\,0,\\
    (E\tilde{O}\times ED)_{z} &=\,2p_{1}\bigl(p_{0}-p_{2}\bigr)
    \Bigl[  \bigl(p_{0}-p_{1}\bigr) \bigl(1-p_{1}+p_{2}\bigr)
        + \bigl(p_{1}-p_{2}\bigr)\bigl(p_{1}+2p_{2}-p_{3}\bigr)\Bigr]\,\geq\,0,\\
    (D\tilde{O}\times DC)_{z} &=\,
    \bigl(p_{0} - p_{1}\bigr)^{2} \bigl(p_{0} + p_{2}\bigr)\bigl((3 p_{0} + 3 p_{1} - 1) + 3 (p_{1} - p_{3})\bigr)
     \nonumber\\
     & \ \ +\, \bigl(p_{0} - p_{1}\bigr)\bigl(p_{1} - p_{2}\bigr)
     \bigl(3 (p_{0} - p_{1}) (p_{0} + p_{1}) + 6 p_{0} (p_{2} - p_{3}) + 2 p_{2} (2 p_{0} + 5 p_{1} + 2 p_{2} - p_{3})\bigr)
     \nonumber\\
     &\ \ +\, \bigl(p_{1} - p_{2}\bigr)^{2} \bigl(
     \bigl[p_{0} - 3 p_{3} (1 - p_{3}) + 6 p_{1} p_{2}\bigr] + (p_{1} - p_{2}) + 3 p_{2}^{2} + 2 (p_{0} - p_{1}) (1 + 3 p_{0})\bigr)
     \,\geq\,0,
\end{align}
\end{subequations}
where the last inequality follows since the term in angled brackets in the last line can be shown to nonnegative using $6p_{1}p_{2}\geq 6p_{3}^{2}$ and $3 p_{3} (1 - 3p_{3})\leq \tfrac{1}{4}\leq p_{0}$.
\end{widetext}


\newpage
\subsection{Outlook on the geometric method in higher dimensions}\label{app:higherdgeom}

Here we present a possible recursive generalisation of the geometric approach, based on re-expressing condition~(\ref{item i}) in terms of majorisation relations rather than trying to prove it in a full geometric way.

To illustrate this method we refer back to the proof of Theorem~\ref{theorem:dim3result} for the case $d=4$. We have seen that the point $\vec{v}_1 =(0,x_1(\beta),x_2(\beta),-1)$ can be reached with a transformation of the type $\vec{v}_1 =M_{q_1}\vec{q}+(\openone+\Pi)\openone\vec{r}_1 +\frac{1}{2} (\openone+\Pi^2)\openone \vec{r}_2$, where
\begin{equation}
M_{q_1}=\begin{pmatrix} m & 1-m & 0 & 0 \\ 1-m & m & 0 & 0 \\
    0 & 0 & 1 & 0 \\
    0 & 0 & 0 & 1 \end{pmatrix}  ,
\end{equation}
with $m=1-1/2(p_0+p_1)$, which is doubly stochastic since $p_0+p_1\geq 1/2$ holds for $d=4$.

To prove that $\vec{v}_2=(0,0,x_2(\beta),-1)$ can be reached, we also considered $M_{r_2}=\openone$, which turns out to be a big simplification, due to the following reasoning. Assuming $M_{r_2}=\openone$ we can rewrite the general transformation that (potentially) reaches $\vec{v}_2$ as
\begin{equation}\label{eq:geomproof1}
    \vec{v}_2+\vec q - \vec p - (\openone+\Pi)(M_{r_1}-\openone) \vec{r}_1 = M_q \vec q ,
\end{equation}
where we used the fact that $\tfrac 1 2 (\openone+\Pi^2)\vec{r}_2 = \vec p-\vec q -(\openone+\Pi) \vec{r}_1$. Thus, with this simplification of assuming $M_{r_2}=\openone$, the problem reduces to showing that there exist some pair $(M_{q},M_{r_1})$ such that Eq.~(\ref{eq:geomproof1}) holds. In particular, we use the fact that $(1+\Pi)(M_{r_1}-\openone)$ has norm equal to zero,
which in turn implies that the sum of vectors on the left hand side of Eq.~(\ref{eq:geomproof1}) has the same norm as $\vec q$, since both $\vec{v}_2$ and $\vec p$ have norm $1$ (i.e., they are probabilities). Note also that the vector $(\openone+\Pi)(M_{r_1}-\openone)$ by itself need not be ordered with decreasing components, as well as the whole expression in the left hand side of Eq.~(\ref{eq:geomproof1}). More precisely, in this case the greatest component of expression $\vec{v}_2+\vec q - \vec p$ is the third.

To prove the statement, then, we can show that $\vec{v}_2+\vec q - \vec p - (\openone+\Pi)(M_{r_1}-\openone) \vec{r}_1$ is majorised by $\vec q$, or, more precisely, we have to find a proper stochastic matrix $M_{r_1}$ such that the above majorisation holds. This is achieved again with a matrix of the form $M_{r_1}=M_{q_1}$ that leads to
\begin{equation}\label{eq:geomproof2}
(\openone+\Pi)(M_{r_1}-\openone) \vec{r}_1=(-a,0,a,0) ,
\end{equation}
with $a=(m-1)p_1(p_0-p_2)$.
Then, the first condition for the majorisation relation
to hold is that the third component of the left hand side of  Eq.~(\ref{eq:geomproof1}) is smaller than the first component of $\vec q$, i.e., that
\begin{align}
\frac{p_0+p_1-2p_2}{3}+p_2^2+(m-1)p_1(p_0-p_2) \leq p_0^2 ,
\end{align}
which implies that we have to choose
\begin{align}
1-m= \max \{0, \frac{p_0+p_1-2p_2-3(p_0^2-p_2^2)}{3p_1(p_0-p_2)} \} ,
\end{align}
such that the majorisation relation is correctly satisfied, since the other conditions are also trivially satisfied.

The last step to check is that with this choice of the parameter $1-m$, $M_{r_1}$ is indeed a doubly stochastic matrix. This is true because
\begin{equation}\label{eq:1md4}
 1-m \leq 1 \iff p_0+p_1-2p_2-3(p_0^2-p_2^2) \leq 3p_1(p_0-p_2)
\end{equation}
also holds, as well as $1-m\geq 0$, which holds by construction.

For higher dimensions, we can try a sort of recursive reasoning, that we work out here for the case $d=5$.
Let us hypotesise that the three vertices $(\vec v_1, \vec v_2, \vec v_{3})$ can be found with a set of transformations $\{(M_{q_i}, (M_{r_1})_i , \openone)\}_{i=1,2,3}$, i.e., that are such that the last stochastic matrix is always the identity $(M_{r_{2}})_i=\openone$ for all $i=1,2,3$.
This allows us to decompose the last vector as
\begin{equation}
     (\mathds{1}+ \Pi^2)\vec{r}_2 = \vec{p}- \vec{q} - (\mathds{1}+ \Pi)\vec{r}_1 ,
\end{equation}
that is $\vec p - \vec q$ minus a sum of terms that have zero norm. Then, we rewrite the condition for a vertex $v_i$ to be reached as
\begin{align}\label{eq:geomproof3}
 \vec{v}_{i}- \vec p + \vec q - (\mathds{1}+ \Pi)((M_{r_1})_i-\openone) \vec{r}_1 &= M_{q_i}\vec{q} ,
\end{align}
which is satisfied if and only if the left hand side is majorised by $\vec q$. So now, the task would be to find suitable doubly stochastic matrices $(M_{r_1})_i$ such that the left hand side is majorised by $\vec q$ for each of the $\vec v_i$. Let us then start by considering $\vec v_1$, which in the probability coordinates is $\vec v_1=((p_0+p_1)/2,(p_0+p_1)/2,p_2,p_3,p_4)$. In this case,
the greatest component of $\vec{v}_{i}- \vec p + \vec q$ is the second, namely $\tfrac{p_0-p_1} 2 +p_1^2$. Unlike
lower dimensions, however, in this case we have that $\tfrac{p_0-p_1} 2 +p_1^2\leq p_0^2 \iff p_0+p_1 \geq 1/2$ does not always hold, meaning that this time $(M_{r_1})_1=\openone$ does not suffice. Thus, in general we have to find $(M_{r_1})_1$ such that the second component of the vector $\vec{\tilde r}_1:=(\mathds{1}+ \Pi)((M_{r_1})_1-\openone) \vec{r}_1$ is negative. The simplest would be if $\vec{\tilde r}_1=(-a,a,0,0,0)$ for some positive $a$. However, it can be readily checked that this is not possible. Nevertheless, we can look for vectors, for example, of the form $\vec{\tilde r}_1=(a_1,a_2,-a_1-a_2,0,0)$, for positive $a_2$ such that we can satisfy $\tfrac{p_0-p_1} 2 +p_1^2-a_2\leq p_0^2$ by choosing
\begin{align}
a_2 = \max \{ 0, (p_0-p_1)(\tfrac 1 2-p_0-p_1)\} .
\end{align}
Afterwards, the next condition for the majorisation to hold is either
\begin{align}
 p_0^2 +p_1^2 - a_1 -a_2 \leq p_0^2+p_1^2 ,
\end{align}
or
\begin{equation}\label{eq:lastappcond3}
    \tfrac{p_0-p_1} 2 +p_1^2+p_2^2+a_1\leq p_0^2+p_1^2 .
\end{equation}
While the first is trivially satisfied if $a_1$ and $a_2$ are both positive, the second needs to be checked.
After this, all other conditions for the majorisation will be trivially satisfied.
This can be achieved with the following matrix:
\begin{equation}
\begin{aligned}
(M_{r_1})_1&=\begin{pmatrix} 1-m_1 & m_1 & 0 & 0 &  0\\ 0 & 1-m_1 & m_1 & 0 & 0 \\
    m_1 & 0 & 1-m_1 & 0 & 0 \\
    0 & 0 & 0 & 1-m_2 & m_2 \\
    0 & 0 & 0 & m_2 & 1-m_2
    \end{pmatrix} ,
\end{aligned}
\end{equation}
with
\begin{equation}
\begin{aligned}
    a_2&=m_2p_4(p_0-p_3) = m_1(p_0p_1-p_2p_3) , \\
    a_1&=m_1p_2(p_1-p_3) ,
\end{aligned}
\end{equation}
and we have to ensure the following conditions (only in the case $1/2\geq p_0+p_1$), coming from $m_1\leq 1$ and $m_2\leq 1$
\begin{equation}
    \begin{gathered}\label{eq:1md5II}
     (p_0-p_1)(\tfrac 1 2-p_0-p_1)/(p_0p_1-p_2p_3) \leq 1 , \\
    (p_0-p_1)(\tfrac 1 2-p_0-p_1)/(p_4(p_0-p_3))  \leq 1 ,
    \end{gathered}
    \end{equation}
    plus we have to guarantee that Eq.~(\ref{eq:lastappcond3}) is satisfied, which leads to
    \begin{equation}\label{eq:lastappcond31}
       (p_0-p_1)(\tfrac 1 2-p_0-p_1)p_2(p_1-p_3)/(p_0p_1-p_2p_3) \leq p_0^2-p_2^2-\tfrac{p_0-p_1} 2 ,
    \end{equation}
    since $a_1=(p_0-p_1)(\tfrac 1 2-p_0-p_1)p_2(p_1-p_3)/(p_0p_1-p_2p_3)$.
The above conditions are nontrivial, but still satisfied, as it is shown below in Sec.~\ref{app:higherdgeomsub}.

Let us now come to $\vec v_2$. For this case we can try to re-use the result of $d=4$, and show that the majorisation relation
\begin{align}\label{eq:geomproof4}
 \vec{v}_{2}- \vec p + \vec q + (\mathds{1}+ \Pi)((M_{r_1})_2-\openone) \vec{r}_1 &= M_{q_2}\vec{q} ,
\end{align}
holds for
\begin{equation}
\begin{aligned}
(M_{r_1})_2&=\begin{pmatrix} m & 1-m & 0 & 0 &  0\\ 1-m & m & 0 & 0 & 0 \\
    0 & 0 & 1 & 0 & 0 \\
    0 & 0 & 0 & 1 & 0 \\
    0 & 0 & 0 & 0 & 1
    \end{pmatrix} ,
\end{aligned}
\end{equation}
with the appropriate coefficient $1-m$, which is exactly the same as in the previous case, namely
\begin{align}
1-m= \max \{0, \frac{p_0+p_1-2p_2-3(p_0^2-p_2^2)}{3p_1(p_0-p_2)} \} ,
\end{align}
and again we have to ensure that
\begin{equation}\label{eq:reld5d4}
 1-m \leq 1 \iff p_0+p_1-2p_2-3(p_0^2-p_2^2) \leq 3p_1(p_0-p_2) .
\end{equation}
However, while Eq.~(\ref{eq:reld5d4}) holds alwasy true for $d=4$, this is no longer the case for $d=5$. Still, there is a range of parameters $\{E_i\}_{i=1}^5$ where it holds (see also Sec.~\ref{app:higherdgeomsub} below).

Last, for the point $\vec{v}_{2}$ we can consider a matrix similar to the above, namely
\begin{align}
(M_{r_1})_3&=\begin{pmatrix} 1 & 0 & 0 & 0 &  0\\ 0 & m & 1-m & 0 & 0 \\
    0 & 1-m & m & 0 & 0 \\
    0 & 0 & 0 & 1 & 0 \\
    0 & 0 & 0 & 0 & 1
    \end{pmatrix} ,
\end{align}
but now the condition to be satisfied is $\tfrac{p_0+p_1+p_2+p_3} 4 -p_3+p_3^2-(1-m)p_2(p_1-p_3)\leq p_0^2$ which leads to
\begin{align}
    1-m\,=\,\max \{0, \frac{p_0+p_1+p_2-3p_3-4(p_0^2-p_3^2)}{4p_2(p_1-p_3)} \} .
\end{align}
The requirement $1-m \leq 1$ thus translates to
\begin{equation}\label{eq:proofd5last}
  p_0+p_1+p_2-3p_3-4(p_0^2-p_3^2) \leq 4p_2(p_1-p_3) .
\end{equation}
While this last relation might not hold in general, it can be satisfied in some appropriately chosen range of the parameters $\{E_i\}_{i=1}^4$. Thus, in order to fully prove the $d=5$ case with this method, one would need to find other possible $(M_{r_1})_2$ and $(M_{r_1})_3$ for the respective complementary parameter regions. Nevertheless, this can possibly still be achieved while keeping $M_{r_2}=\openone$.

In conclusion, we have seen a possible strategy to prove that all the vertices listed in condition~(\ref{item i}) can be reached with the transformations preserving equal marginals. This exploits the geometry of such achievable thermal marginals and tries to avoid the difficulties encountered in the ''passing on the norm´´ approach, by considering the assumption that it is always possible to choose one of the matrices $M_{r_i}$ as the identity. This assumption is in fact supported by the geometric proofs in dimensions $d=3$ and $d=4$.


\subsubsection{Validity of doubly stochasticity of $M_{r_i}$ for $d=5$}\label{app:higherdgeomsub}

Here we study the validity of the relations that guarantee that the matrices $M_{r_i}$ needed for the proof of $d=5$ are indeed doubly stochastic. We start with Eq.~(\ref{eq:1md4}), which was also needed for the $d=4$ case. To prove Eq.~(\ref{eq:1md4}) we can see that
\begin{equation}
    p_0+p_1-2p_2-3(p_0^2-p_2^2)\leq 3p_1(p_0-p_2) .
\end{equation}
The above relation in fact holds since we can rewrite it as
\begin{equation}
    (p_0-p_2)(1-3(p_0+p_1+p_2))+(p_1-p_2)\leq 0 ,
\end{equation}
and we can see that indeed
\begin{equation}
\begin{gathered}
    (p_0-p_2)(1-3(p_0+p_1+p_2))+(p_1-p_2)\leq \\ (p_0-p_2)(2-3(p_0+p_1+p_2)) \leq 0 ,
    \end{gathered}
\end{equation}
since $p_0+p_1+p_2 \geq 2/3$ holds for $d=4$, but not always $d=5$. In fact, for $d=5$ we can have points close to $p_0=p_1=p_2=1/5$.

To prove Eq.~(\ref{eq:1md5II}), we have to prove that, whenever $\tfrac 1 2 \geq p_0+p_1$ it holds that
\begin{equation}
    (p_0-p_1)(\tfrac 1 2-p_0-p_1) \leq p_4(p_0-p_3) \leq  p_0p_1-p_2p_3 .
\end{equation}
This is indeed true, since we can rewrite it as
\begin{equation}
    (p_0-p_1)(\tfrac 1 2-p_0-p_1) - p_4(p_0-p_3) \leq (p_0-p_1)(\tfrac 1 2-p_0-p_1-p_4) \leq 0 ,
\end{equation}
where the last relation holds whenever $\tfrac 1 2 \geq p_0+p_1$ because we can write
\begin{equation}
    \tfrac 1 2-p_0-p_1-p_4 = p_2+p_3 - \tfrac 1 2 \leq  p_0+p_1 - \tfrac 1 2 \leq 0 .
\end{equation}

Let us know consider Eq.~(\ref{eq:lastappcond31}). Since
$p_2(p_1-p_3)\leq (p_0p_1-p_2p_3)$ we can rewrite it as
\begin{equation}
      (p_0-p_1)(\tfrac 1 2-p_0-p_1) \leq p_0^2-p_2^2-\tfrac{p_0-p_1} 2 ,
\end{equation}
which is always satisfied since $p_0-p_1 \leq 2(p_0^2-p_1^2) \leq 2p_0^2-p_2^2-p_1^2$.

Finally, let us consider Eq.~(\ref{eq:proofd5last}). We can rewrite the expression as
\begin{equation}
    (p_0-p_3)(1-4(p_0+p_3))+(p_1-p_3)(1-4p_2)+p_2-p_3 \leq 0 ,
\end{equation}
and again, we can see that
\begin{equation}
\begin{gathered}
    (p_0-p_3)(3-4(p_0+p_2+p_3)) \leq \\ (p_0-p_3)(1-4(p_0+p_3))+(p_1-p_3)(1-4p_2)+p_2-p_3
    \end{gathered}
\end{equation}
and in particular both the above expressions are surely negative whenever
\begin{equation}
    p_0+p_2+p_3 \geq \tfrac{3}{4},
\end{equation}
which is not always true, but it is so in a certain range of parameters.

\end{document}